\documentclass[submission,copyright,creativecommons]{eptcs}
\providecommand{\event}{Linearity/TLLA 2020} 
\usepackage{underscore}           

\usepackage{graphicx}
\usepackage{mathrsfs}
\usepackage{xspace}
\usepackage[dvips]{epsfig}
\usepackage{color}
\usepackage{amsmath}
\usepackage{amsthm}
\usepackage{amssymb}
\usepackage{mathtools}
\usepackage{thm-restate}
\usepackage{multirow}
\usepackage{multicol}

\usepackage{wasysym}
\usepackage[dvips]{graphicx}
\usepackage{stmaryrd}
\usepackage{ebproof} 
\usepackage{tikz}
\usepackage{cmll}
\usepackage{virginialake}
\usepackage{marginnote}

\usepackage{pzccal}
\usepackage{cleveref}
\usepackage{adjustbox}

\usetikzlibrary{graphs}

\def\N{\mathbb N}

\def\fequiv{\simeq}
   
\def\cons#1{\{#1\}}
\def\symbol{\bullet}
\def\contextsym{\mathsf C}
\newcommand{\ctx}{\contextsym \cons{}}
\newcommand{\ctxp}[1]{\contextsym \cons{#1}}
\newcommand{\ctxpp}[1]{\contextsym' \cons{#1}}

\def\sizeof#1{|#1|}

\def\cneg#1{\bar #1}
\def\widecneg#1{\overline{#1}}

\def\rclr#1{{\color{red}#1}}
\def\bclr#1{{\color{blue}#1}}

\def\set#1{\{#1\}}

\newtheorem{theorem}{Theorem}[section]

\theoremstyle{plain}

\newtheorem{fact} [theorem]{Fact}

\newtheorem{corollary}  [theorem]{Corollary}
\newtheorem{lemma}      [theorem]{Lemma}
\newtheorem{proposition}[theorem]{Proposition}

\theoremstyle{definition}
\newtheorem{definition}[theorem]{Definition}
\newtheorem{remark}[theorem]{Remark}

\def\LL{\mathsf{LL}}
\def\MLL{\mathsf{MLL}}
\def\SDill{\mathsf{DiLL}_0}
\def\DDI{\mathsf{DDI}}

\def\MELL{\mathsf{MELL}}
\def\SDDI{\mathsf{SDDI}}
\def\SDDIs{\SDDI}
\def\DILL{\mathsf{DiLL}}
\def\DiLL{\DILL}

\def\ie{\textit{i.e.}\xspace}
\def\resp{resp.\xspace}

\def\cS{\mathcal S}

\def\upfrag{\set{\aiur, \ocdur, \wndur, \ocwur, \allowbreak\wnwur, \allowbreak\occur, \allowbreak\wncur, \allowbreak\pur, \zur, \swir}}
\def\downfrag{\set{\aidr, \allowbreak\ocddr, \allowbreak\wnddr, \ocwdr, \wnwdr, \occdr, \wncdr, \pdr, \zdr, \swir}}

\def\SDDIup{\DDI^\uparrow} 
\def\SDDIdown{\DDI^\downarrow} 
\def\DDIdown{\SDDIdown} 
\def\DDIup{\SDDIup} 
\def\DMELL{\DDIdown_{-}} 

\def\isneg#1{{#1}^\bullet}

\def\lplus{\mathsf{+}}
\def\lone{{\mathsf{1}}}
\def\lbot{\mathsf{\bot}}
\def\lzero{\mathsf{0}}
\newcommand{\provesym}{\tikz[baseline=-.65ex]{\draw[very thick] (0,0)--(.4,0);\draw[thin] (0,-.1)--(0,.1);}}
\newcommand{\provevia}[1]{\mathrel{\stackrel{\scriptscriptstyle#1}{\provesym}}}

\def\inttrans#1{\mathsf{#1}}
\def\lones{\inttrans{n}}
\newcommand\loness[1]{\inttrans{#1}}

\def\IH{\mathsf{IH}}
\odbackgroundfalse
\odframetrue

\def\la{a}
\def\lb{b}
\def\lc{c}
\def\lA{A}
\def\lB{B}
\def\lC{C}
\def\lD{D}

\def\llone{\lone}
\def\llbot{\lbot}

\def\lplus{+}
\def\lpar{\parr}
\def\ltens{\otimes}

\def\axrule{\mathsf {ax}}
\def\cutr{\mathsf {cut}}
\def\cutrule{\cutr}

\def\wnwrule{\wn\mathsf{w}}
\def\ocwrule{\oc\mathsf{w}}
\def\wncrule{\wn\mathsf{c}}
\def\occrule{\oc\mathsf{c}}
\def\ocdrule{\oc\mathsf{d}}
\def\wndrule{\wn\mathsf{d}}
\def\ocprule{\oc\mathsf{p}}

\def\sumrule{\mathsf{sum}}

\def\zerorule{\mathsf{zero}}

\def\mixr{\mathsf{mix}}

\def\onerule{\llone}
\def\botrule{\llbot}

\def\prule{\mathsf{\oc p}}

\def\unaryrule{{\mathsf{r}_{1}}}
\def\binaryrule{{\mathsf{r}_{2}}}

\def\downr#1{{#1}^\downarrow}
\def\upr#1{{#1}^\uparrow}

\def\ruler{\mathsf{\rho}}
\def\rules{\mathsf \sigma}
\def\rulet{\mathsf \tau}

\def\aidr{\downr {\mathsf{ai}}}
\def\idr{\downr {\mathsf {i}}}

\def\swir{\mathsf {s}}

\def\wnwdr{\downr {\mathsf{\wn w}}}
\def\ocwdr{\downr {\mathsf{\oc w}}}

\def\wncdr{\downr {\mathsf{\wn c}}}
\def\occdr{\downr {\mathsf{\oc c}}}

\def\wnddr{\downr {\mathsf{\wn d}}}
\def\ocddr{\downr {\mathsf{\oc d}}}

\def\pdr{\downr {\lplus}}
\def\zdr{\downr {\lzero}}

\def\aiur{\upr {\mathsf{ai}}}
\def\iur{\upr {\mathsf{i}}}

\def\wnwur{\upr {\mathsf{\wn w}}}
\def\ocwur{\upr {\mathsf{\oc w}}}

\def\wncur{\upr {\mathsf{\wn c}}}
\def\occur{\upr {\mathsf{\oc c}}}

\def\wndur{\upr {\mathsf{\wn d}}}
\def\ocdur{\upr {\mathsf{\oc d}}}

\def\pur{\upr {\lplus}}
\def\zur{\upr {\lzero}}

\def\dD{\mathcal D}
\newcommand{\deriv}[4]{#1 \mathbin{\,\triangleright\,} #2 \provevia{#3} #4}

\def\etared{\rightsquigarrow_\eta}
\def\cutred{\rightsquigarrow_\cutr}
\def\normSym{\mathsf{norm}}
\def\normred{\rightsquigarrow_{\normSym}}
\def\normeq{\simeq_\normSym}

\def\translatesto{\overset{\toform{\,}}\to}
\def\toform#1{\llbracket#1\rrbracket}
\def\toformnew#1{\llparenthesis#1\rrparenthesis}
\def\translatesnewto{\overset{\toformnew{\,}}\to}

\tikzstyle{edgestyle}=[>=stealth,overlay,remember picture,thin, opacity=1]

\tikzstyle{edgestyle}=[>=stealth,overlay,remember picture,thin, opacity=1]  
\tikzstyle{dashedgestyle}=[>=stealth,overlay,remember picture,thin, opacity=1,dashed]

\title{A Deep Inference System for Differential Linear Logic}

\author{Matteo Acclavio 
	\institute{University of Luxembourg, Belval, Luxembourg}
	\and
	Giulio Guerrieri 
	\institute{University of Bath, Bath, United Kingdom}
}
\def\titlerunning{A Deep Inference System for Differential Linear Logic}
\def\authorrunning{M. Acclavio, G. Guerrieri}

\begin{document}

\maketitle

\begin{abstract}
Differential linear logic ($\mathsf{DiLL}$) provides a fine analysis of resource consumption in cut-elimination. 
We investigate the subsystem of $\mathsf{DiLL}$ without promotion 
in a deep inference formalism, where cuts are at an atomic level.
In our system every provable formula admits a derivation in normal form, via a normalization procedure that commutes with the translation from sequent calculus to deep inference. 

\end{abstract}

\section{Introduction}
\label{sect:introduction}

Girard \cite{Girard87} introduced linear logic ($\LL$) as a refinement of intuitionistic and classical logics, built around cut-elimination. 
In $\LL$, a pair of dual modalities (the \emph{exponentials} $\oc$ and $\wn$) give a logical status to the operations of erasing and copying (sub-)proofs in the cut-elimination procedure. 
The idea is that \emph{linear} proofs (\ie~proofs without exponentials) use their hypotheses exactly once,
whilst \emph{exponential} proofs may use their hypotheses at will. 
In particular, the \emph{promotion rule} makes a (sub-)proof available to be erased or copied an unbounded number of times, provided that its hypotheses are as well (it is a contextual rule).
Via Curry--Howard correspondence between programs and proofs, $\LL$ gives a logical status to the operations of erasing and copying data in the evaluation process. 
Linear proofs correspond to programs which call their arguments exactly once, exponential proofs to programs which call their arguments at will. 
The study of $\LL$ contributed to unveil the logical nature of resource consumption.

\paragraph{The importance of being differential.} 
A further tool for the analysis of resource consumption in cut-elimination
came from Ehrhard and Regnier's work on 
differential $\lambda$-calculus \cite{EhrhardRegnier03} and \emph{differential linear logic} ($\DiLL$, \cite{EhrhardRegnier06,Pagani09}). 
Despite the fact that $\DiLL$ is inconsistent (every sequent $\vdash \Gamma$ can be proved), it has a cut-elimination theorem \cite{Pagani09,Gimenez11} and internalizes notions from denotational models of $\LL$ into the syntax.
In particular, 
$\SDill$ (the \emph{promotion-free} fragment of $\DiLL$, \cite{EhrhardRegnier06}) is a logic 
corresponding to the semantic constructions defined by Ehrhard's finiteness spaces \cite{Ehrhard05}. 
Finiteness spaces interpret linear proofs as linear functions on certain topological vector spaces, on which one can define an operation of derivative.
Exponential proofs are interpreted as analytic functions, in the sense that they can be arbitrarily approximated by the semantic equivalent of a Taylor expansion \cite{Ehrhard05,Ehrhard18}, which becomes available thanks to the presence of a derivative operator.
In syntactic terms, 
these constructions take an interesting form: 
they correspond to ``symmetrizing'' the exponential modalities, \ie~in $\SDill$ 
the rules handling the dual exponential modalities $\oc$ and $\wn$ are perfectly symmetrical, although the logic is not self-dual.
%
Indeed, in $\LL$, only the promotion rule introduces the $\oc$ modality, creating inputs that can be called an unbounded number of times. 
In $\SDill$ the rules handling the $\oc$ modality ($\oc$-dereliction $\ocdrule$, $\oc$-contraction $\occrule$, $\oc$-weakening $\ocwrule$) are the duals of the usual rules dealing with the $\wn$ modality ($\wn$-dereliction $\wndrule$, $\wn$-contraction $\wncrule$, $\wn$-weakening $\wnwrule$). 
In particular, \emph{$\oc$-dereliction} expresses in the syntax the semantic derivative: it releases inputs of type $\oc \lA$ that must be called exactly \emph{once}, so that executing a program $f$ on a ``$\oc$-derelicted'' input $x$ (\ie performing cut-elimination on a proof $f$ cut with a ``$\oc$-derelicted'' sub-proof $x$) amounts to compute the best linear approximation of $f$ on $x$. 
This imposes non-deterministic choices: if in an evaluation the program $f$ 
needs several copies of the input $x$ (\ie~if the proof $f$ uses several times the hypothesis $\oc \lA$), then there are different executions of $f$ on $x$, depending on which sub-routine (\ie hypothesis) of $f$ is fed with the unique available copy of $x$. 
Thus we get a formal \emph{sum}, where each term represents a possibility.
The rules $\oc$-contraction and $\oc$-weakening put together a finite (possibly $0$) number of copies of an input, so that it can be called a \emph{bounded} number of times during execution. 

What is also interesting is that $\LL$ promotion rule  can be encoded in $\SDill$ through the notion of syntactic \emph{Taylor expansion} \cite{EhrhardRegnier06b,EhrhardRegnier08,MazzaPagani07,PaganiTasson09,GuerrieriPellissierTortora16,Carvalho18,GuerrieriPellissierTortora19,GuerrieriPellissierTortora20}:
a proof in $\LL$ can be decomposed into a possibly infinite set of (promotion-free) proofs in $\SDill$. 
Given a proof in $\LL$ with exactly one promotion rule $\ocprule$, the idea is to replace $\ocprule$ (which makes the resource $\pi$ available at will)
with an infinite set of ``differential'' proofs in $\SDill$, each of them taking $n \in \mathbb{N}$ copies of $\pi$ so as to make the resource $\pi$ available exactly $n$ times.
The potential infinity of the promotion rule becomes an actual infinite  via the Taylor expansion.

\paragraph{Nets vs. sequents.}
The system $\SDill$ is usually presented in two formalisms: sequent calculus and  Lafont's interaction nets \cite{Lafont90} (a graphical representation of proofs similar to $\LL$ proof-nets).
The \emph{symmetry} of the rules handling the dual exponentials 
$\oc$ and $\wn$ in $\SDill$ is evident in interaction nets, 
but not at all in the sequent calculus. 
In interaction nets for $\SDill$, the rules for $\wn$ and $\oc$ have the \emph{same~geometry}:

\vspace*{-2\baselineskip}
\begin{center}
\includegraphics{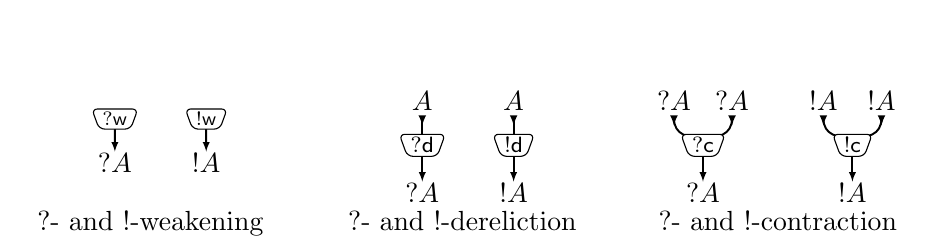}
\end{center}

So, the distinction between $\oc$ and $\wn$ is given only by their different behaviors in correctness graphs (a geometrical characterization of the interaction nets corresponding to proofs in the $\SDill$ sequent calculus).
But meaningful operations in $\SDill$ such as cut-elimination can be defined directly on interaction nets, regardless of being correct or not.
The benefit is that $\SDill$ cut-elimination steps defined on interaction nets are perfectly symmetric: for instance, the step for a cut $\wncrule$/$\ocdrule$ is exactly the \emph{dual} of the step for a cut $\occrule$/$\wndrule$, and similarly for the other steps (see \cite[Fig. 4]{Pagani09}).

This elegant symmetry in the presentation of cut-elimination steps is lost in $\SDill$ sequent calculus, see our \Cref{fig:cut-elim}.
Moreover, cut-elimination in $\SDill$ sequent calculus has to deal with (many) uninteresting commutative steps, while interaction nets get rid of them.
Thus, interaction nets allow one to express $\SDill$ cut-elimination with a sharper account than in sequent calculus. 
Not by chance, all papers dealing with $\SDill$ cut-elimination use only interaction nets \cite{EhrhardRegnier06,Pagani09,Tranquilli09,Gimenez11,Tranquilli11,PaganiTranquilli17}.

However, the interaction net presentation of $\SDill$ has some flaws that do not affect the sequent calculus:
interaction nets do not have an \emph{inductive tree-like structure} and so it is not easy to handle them. 
Moreover, not all 
interaction nets correspond to a derivation in $\SDill$ sequent calculus, a \emph{global} geometrical correctness criterion is required to identify them.

\paragraph{Our contribution.} 
We define a proof system for $\SDill$ in the formalism of \emph{open deduction}~\cite{gug:gun:par:2010} following the principles of \emph{deep inference} ~\cite{gug:str:01,brunnler:tiu:01,gug:SIS,ralph:phd,tub:str:esslli19}.
Such a formalism, which allows rules to be applied deep in a context, provides a more flexible composition of derivations and makes explicit the behavior of the cut-elimination process in $\SDill$ in a more fine-grained way, since it pushes cut-elimination at an \emph{atomic} level.
Besides, our deep inference system for $\SDill$ gathers good qualities of both sequent calculus and interaction nets formalisms: it restores the interaction net \emph{symmetries} lost in the sequent calculus and its derivations keep a handy \emph{inductive} tree-like (or better, sequence-like) structure as in the sequent calculus, without the need for a global correctness criterion like in~interaction~nets.

A first attempt in the direction of a deep inference system for $\SDill$ is in \cite{gim:phd} where, however, the $\sumrule$-rule is absent and, as a consequence, it is not suitable to represent the dynamic behavior of $\SDill$.

To fully recover the expressiveness of this logic, we design our system to include a binary connective $\lplus$ which represents the sum operation.
The rules for $\lplus$ (and for its unit $\lzero$) 
prevent the use of Guglielmi and Tubella's general result \cite{tubella:phd,TubellaGuglielmi18} to show cut-elimination.
However, we are able to define a 
normalization procedure by rule permutations which fully captures the dynamics of $\SDill$ cut-elimination, in a way similar to the one in 
\cite{str:phd,str:MELLinCoS}.
Our system is sound and complete with respect to $\SDill$ sequent calculus, through a translation that commutes with cut-elimination/normalization.

In the normalization procedure, we can classify our rule permutations depending on their behavior: some rule permutations correspond to 
multiplicative cut-elimination steps, 
other permutations correspond to ``resource management'' cut-elimination steps 
(involving the $\wn$ and $\oc$ rules),
other permutations correspond to ``slice management'' operations 
(involving the propagation of $\lplus$ and $\lzero$).


\section{Differential Linear Logic}
\label{sec:DiLL}
We present here the classical, propositional, one-sided sequent calculus for \emph{differential linear logic} without promotion ($\SDill$).
The formulas of $\SDill$ are exactly the same as in the multiplicative exponential fragment of linear logic ($\MELL$).
\emph{$\MELL$ formulas} are defined by the grammar below, where $\la, \lb, \lc, \dots$ range over a countably infinite set of propositional variables:
\begin{center}
	$\lA, \lB  \Coloneqq \la \mid \widecneg \la \mid \llone \mid \llbot \mid \lA\ltens \lB \mid \lA\lpar \lB \mid \oc \lA \mid \wn \lA $
\end{center}
Linear negation $\widecneg{(\cdot)}$ is defined through De Morgan laws so as to be involutive \mbox{($\widecneg{\widecneg{\lA}} = \lA$ for any $\lA$):}
\begin{center}
	$
	\widecneg{(\la)} = \widecneg{\la} \quad
	\widecneg{(\widecneg{\la})} = \la \quad 
	\widecneg{\lA \otimes \lB} = \widecneg{\lA} \parr \widecneg{\lB} \quad
	\widecneg{\lA \parr \lB} = \widecneg{\lA} \otimes \widecneg{\lB} \quad
	\widecneg{\llone} = \llbot \quad
	\widecneg{\llbot} = \llone \quad
	\widecneg{\oc \lA} = \wn \widecneg{\lA} \quad
	\widecneg{\wn \lA} = \oc \widecneg{\lA}
	$
\end{center}

Variables and their negations 
are \emph{atomic}; $\otimes, \parr$ are \emph{multiplicative connectives} and $\llone, \llbot$  are their respective \emph{units}; $\oc, \wn$ are \emph{exponential modalities}. 
A \emph{$\MELL$ sequent} is a (finite) multiset of $\MELL$ formulas $\lA_1, \dots, \lA_n$ (for any $n \geq 0$), and it is ranged over by~$\Gamma, \Delta, \Sigma$.

\begin{figure}[!t]
  \def\myskip{\hskip1em}
   \hbox to\textwidth{\hfil 
   	\scriptsize
    \begin{tabular}{c@{\myskip}c@{\myskip}c@{\myskip}c@{\myskip}c@{\myskip}c@{\myskip}c@{\myskip}c@{\myskip}|@{\myskip}c}
$\vlinf{}{\axrule}{\vdash \lA , \widecneg{\lA}}{}$
&
$\vlinf{}{\lpar}{\vdash \Gamma, \lA\lpar  \lB}{\vdash \Gamma, \lA, \lB}$
&
$\vliinf{}{\ltens}{\vdash \Gamma, \lA \ltens   \lB, \Delta}{\vdash \Gamma, \lA}{\vdash \lB, \Delta}$
&
$\vlinf{}{\onerule}{\vdash\llone}{}$
&
$\vlinf{}{\botrule}{\vdash \Gamma , \llbot}{\vdash \Gamma}$
&
$\vliinf{}{\cutr}{\vdash \Gamma, \Delta}{\vdash \Gamma, \lA}{\vdash \widecneg{\lA}, \Delta}$
&
&
\\
$\vlinf{}{\ocdrule}{\vdash \Gamma, \oc \lA}{\vdash \Gamma, \lA}$
&
$\vlinf{}{\wndrule}{\vdash \Gamma, \wn \lA}{\vdash \Gamma, \lA}$
&
$\vliinf{}{\occrule}{\vdash \Gamma,  \oc \lA,\Delta}{\vdash \Gamma, \oc\lA}{\vdash 	\oc \lA, \Delta}$
&
$\vlinf{}{\wncrule}{\vdash \Gamma, \wn \lA}{\vdash \Gamma, \wn \lA , \wn \lA}$
&
$\vlinf{}{\ocwrule}{\vdash \oc \lA}{}$
&
$\vlinf{}{\wnwrule}{\vdash \Gamma, \wn \lA}{\vdash \Gamma}$
&
$\vliinf{}{\sumrule}{\vdash \Gamma}{\vdash \Gamma}{\vdash \Gamma}$
&
$\vlinf{}{\zerorule}{\vdash \Gamma}{}$
&
$\vlinf{}{\prule}{\vdash \wn\Gamma,  \oc \lA}{\vdash \wn\Gamma, \lA}$
\end{tabular}    
          \hfil}  
	\caption{Sequent calculus rules for $\SDill$ (on the left) 
		and the promotion rule (on the right) \cite{Pagani09}.}
	\label{fig:scrules}
\end{figure}

\Cref{fig:scrules} gives the sequent calculus rules\footnotemark
\footnotetext{\label{note:co}Usually, in the literature on $\LL$ and $\DiLL$, the rules $\wnwrule$, $\wndrule$, $\wncrule$, $\ocwrule$, $\ocdrule$, $\occrule$ are called respectively weakening, dereliction, contraction, co-weakening, co-dereliction and co-contraction. 
	To avoid clashes with the usual terminology in deep inference (see \Cref{note:co-deep}), we call them $\wn$-weakening, $\wn$-dereliction, $\wn$-contraction, $\oc$-weakening, $\oc$-dereliction and $\oc$-contraction, respectively.}
for \emph{differential linear logic} $\SDill$ (without \emph{promotion} $\prule$);
	the rules on the first line correspond to the  \emph{multiplicative linear logic} fragment $\MLL$.
We set:
$$
\begin{array}{l@{\hskip2em}l@{\hskip2em}l}
	\MELL=\MLL\cup\set{\wnwrule, \wndrule, \wncrule,\prule}
 &
	&
	\SDill^-=\SDill\setminus\set{\zerorule, \sumrule}
\end{array}
$$

%
%

We define $\equiv$ as the least congruence on $\SDill$ derivations generated by the relations in \Cref{fig:zero}.
Roughly, the rule $\zerorule$ plays the role of annihilating element with respect to all the other rules but $\sumrule$, for which it is a neutral element; 
whilst the rule $\sumrule$ commutes with any rule below it.
Clearly, $\equiv$ preserves conclusions and can be oriented so as to define a terminating rewriting relation that pushes down the rules $\zerorule$ and $\sumrule$ in a derivation. 
As a consequence, every derivation in $\SDill$ can be rewritten in a \emph{canonical form} (with the same conclusion).

\begin{definition}[Canonical form, slice]
	\label{def:canonical}
	Let $\pi$ be a derivation in $\SDill$:
	\begin{enumerate}
		\item $\pi$ is a \emph{slice} if it is in $\SDill^-$ (\ie~the rules $\zerorule$ and $\sumrule$ do not occur in $\pi$);
		
		\item $\pi$ is \emph{canonical} or in \emph{canonical form} if 
		either it consists of a $\zerorule$ rule, or it is a slice, or if its last rule is $\sumrule$ 
		with a canonical form as left premise and a slice as right premise.
	\end{enumerate}
	
	A \emph{canonical form} of $\pi$ is any canonical derivation $\pi'$ in $\SDill$ such that $\pi \equiv \pi'$.
\end{definition}

\begin{fact}[Canonicity]
	\label{fact:canon}
	Any derivation in $\SDill$ has a canonical form (with the same conclusion).
\end{fact}


Intuitively, considering only canonical derivations, slices---\ie~derivations in the subsystem $\SDill^-$---are the ``real and meaningful'' proofs in $\SDill$ (corresponding to simple nets in \cite{EhrhardRegnier06,MazzaPagani07,Pagani09,Tranquilli09}),
while the rules $\sumrule$ and $\zerorule$ are needed to define cut-elimination in $\SDill$, in particular they ensure 
that the conclusion of a derivation is preserved after a cut-elimination step (the subsystem $\SDill^-$ is not closed under cut-elimination, see below).
The rule $\sumrule$ puts together slices with the same conclusions $\vdash\!\Gamma$, similarly to multiset union: it expresses the possibility of several ``real proofs'' of $\vdash\!\Gamma$.
The rule $\zerorule$ then corresponds to the empty multiset of ``real proofs'' of $\vdash\!\Gamma$: it claims $\vdash \!\Gamma$ without a proof (it is reminiscent of daimon in  ludics \cite{Girard01}).
Because of 
it, any $\MELL$ sequent (also the empty one) is provable in~$\SDill$.

\begin{figure}[!t]
	\ebproofset{right label template=\tiny$\inserttext$, label separation=0.2em, separation = 1.2em}
	\scriptsize
	\begin{prooftree}[label separation = 0.1em]
		\infer0[\zerorule]{\vdash \Gamma}
		\infer1[\unaryrule]{\vdash \Delta}
	\end{prooftree}
	\!\!$\equiv$
	\begin{prooftree}[label separation = 0.1em]
		\infer0[\zerorule]{\vdash \Delta}
	\end{prooftree}	
	\qquad
	\begin{prooftree}[separation=0.5em, label separation = 0.1em]
		\infer0[\zerorule]{\vdash \Gamma}
		\hypo{\vdash \Delta}
		\infer2[\binaryrule]{\vdash \Sigma}
	\end{prooftree}
	$\equiv$
	\begin{prooftree}[separation=1em, label separation = 0.1em]
		\infer0[\zerorule]{\vdash \Sigma}
	\end{prooftree}
	$\equiv$
	\begin{prooftree}[separation=1em, label separation = 0.1em]
		\hypo{\vdash \Gamma}
		\infer0[\zerorule]{\vdash \Delta}
		\infer2[\binaryrule]{\vdash \Sigma}
	\end{prooftree}
	\quad \ \ 
	\begin{prooftree}[separation=0.5em, label separation = 0.1em]
		\infer0[\zerorule]{\vdash \Gamma}
		\hypo{\vdash \Gamma}
		\infer2[\sumrule]{\vdash \Gamma}
	\end{prooftree}
	$\equiv$
	\begin{prooftree}
		\hypo{\vdash \Gamma}
	\end{prooftree}	
	$\equiv$
	\begin{prooftree}[separation=1em, label separation = 0.1em]
		\hypo{\vdash \Gamma}
		\infer0[\zerorule]{\vdash \Gamma}
		\infer2[\sumrule]{\vdash \Gamma}
	\end{prooftree}
	\qquad
	\begin{prooftree}[separation=1em, label separation = 0.1em]
		\hypo{\vdash \Gamma}
		\hypo{\vdash \Gamma}
		\infer2[\sumrule]{\vdash \Gamma}
		\infer1[\unaryrule]{\vdash \Delta}
	\end{prooftree}
	\!\!$\equiv$
	\begin{prooftree}[separation=1em, label separation = 0.1em]
		\hypo{\vdash \Gamma}
		\infer1[\unaryrule]{\vdash \Delta}
		\hypo{\vdash \Gamma}
		\infer1[\unaryrule]{\vdash \Delta}
		\infer2[\sumrule]{\vdash \Delta}
	\end{prooftree}
	\\[5pt]
	\begin{prooftree}
		\hypo{\vdash \Gamma}
		\hypo{\vdash \Gamma}
		\infer2[\sumrule]{\vdash \Gamma}
		\hypo{\vdash \Delta}
		\infer2[\binaryrule]{\vdash \Sigma}
	\end{prooftree}
	$\equiv$
	\begin{prooftree}
		\hypo{\vdash \Gamma}
		\hypo{\vdash \Delta}
		\infer2[\binaryrule]{\vdash \Sigma}
		\hypo{\vdash \Gamma}
		\hypo{\vdash \Delta}
		\infer2[\binaryrule]{\vdash \Sigma}
		\infer2[\sumrule]{\vdash \Sigma}
	\end{prooftree}
	$\equiv$
	\begin{prooftree}[separation=1em]
		\hypo{\vdash \Gamma}
		\hypo{\vdash \Delta}
		\hypo{\vdash \Delta}
		\infer2[\sumrule]{\vdash \Delta}
		\infer2[\binaryrule]{\vdash \Sigma}
	\end{prooftree}
	\qquad
	\begin{prooftree}
		\hypo{\vdash \Gamma}
		\hypo{\vdash \Gamma}
		\hypo{\vdash \Gamma}
		\infer2[\sumrule]{\vdash \Gamma}
		\infer2[\sumrule]{\vdash \Gamma}
	\end{prooftree}
	\ $\equiv$ \
	\begin{prooftree}
		\hypo{\vdash \Gamma}
		\hypo{\vdash \Gamma}
		\infer2[\sumrule]{\vdash \Gamma}
		\hypo{\vdash \Gamma}
		\infer2[\sumrule]{\vdash \Gamma}
	\end{prooftree}
	\caption{The congruence $\equiv$ on derivations generated by the rules $\zerorule$ and $\sumrule$ in $\SDill$, where $\unaryrule$ is any unary rule in $\SDill$, and $\binaryrule$ is any binary rule in $\SDill$ but $\sumrule$.}
	\label{fig:zero}
\end{figure}

Let $\pi$ be a derivation in $\SDill$. 
We say that $\pi$ is \emph{with atomic axioms} (or \emph{$\eta$-expanded}) if every instance of the rule $\axrule$ introduces a $\MELL$ sequent of the form $\vdash \la, \widecneg{\la}$, where $\la$ is a propositional variable.
We say that $\pi$ is \emph{cut-free} is it does not contain any instance of the rule $\cutr$, \ie $\pi$ is a derivation in $\SDill \setminus \set{\cutr}$.

\begin{proposition}[Atomic axioms]
	\label{prop:atomic}
	For every derivation $\pi$ in $\SDill$ with conclusion $\vdash \Gamma$, there exists a $\eta$-expanded derivation $\pi'$ in $\SDill$ with conclusion $\vdash \Gamma$.
	If, moreover, $\pi$ is canonical (resp.~a slice) then $\pi'$ is canonical (resp.~a slice).
\end{proposition}
\begin{proof}
	Rewrite any non-atomic instance of the rule $\axrule$ according to the \emph{$\eta$-expansion} relation $\etared$ below:
	
	\vspace{-\baselineskip}
\begin{equation}\label{eq:etaexp}
	\ebproofset{right label template=\tiny$\inserttext$, label separation=0.3em, separation = 1.2em}
	\footnotesize
		{\begin{prooftree}
				\infer0[\axrule]{\vdash \lA \ltens \lB, \widecneg{\lA} \lpar \widecneg{\lB}}
		\end{prooftree}}
		\etared
		{\begin{prooftree}
			\infer0[\axrule]{\vdash \lA, \widecneg{\lA}}
			\infer0[\axrule]{\vdash \lB, \widecneg{\lB}}
			\infer2[\ltens]{\vdash \lA \ltens \lB, \widecneg{\lA}, \widecneg{\lB}}
			\infer1[\lpar]{\vdash \lA \ltens \lB, \widecneg{\lA} \lpar \widecneg{\lB}}
		\end{prooftree}}
		\qquad		\qquad
		{\begin{prooftree}
			\infer0[\axrule]{\vdash \lone, \lbot}
		\end{prooftree}}
		\etared~
		{\begin{prooftree}
			\infer0[\onerule]{\vdash \lone}
			\infer1[\lbot]{\vdash \lone, \lbot}
		\end{prooftree}}
		\qquad		\qquad
		{\begin{prooftree}
			\infer0[\axrule]{\vdash \oc \lA, \wn \widecneg{\lA}}
		\end{prooftree}}
				\etared~
		{\begin{prooftree}
			\infer0[\axrule]{\vdash \lA, \widecneg{\lA}}
			\infer1[\wndrule]{\vdash \lA, \wn \widecneg{\lA}}
			\infer1[\ocdrule]{\vdash \oc \lA, \wn \widecneg{\lA}}
		\end{prooftree}}
\end{equation}
	It is immediate to prove that the relation $\etared$ on the derivations of $\SDill$ is terminating.
\end{proof}


\begin{figure}[!ht]
	\ebproofset{right label template=\tiny$\inserttext$, label separation=0.2em, separation = 1.2em}
	\scriptsize
	\begin{prooftree}
		\hypo{\vdash \Gamma, \lA, \lB}
		\infer1[\parr]{\vdash \Gamma, \lA \parr \lB}
		\hypo{\vdash \Delta, \widecneg{\lA}}
		\hypo{\vdash \Sigma, \widecneg{\lB}}
		\infer2[\otimes]{\vdash \Delta, \Sigma, \widecneg{\lA} \otimes \widecneg{\lB}}
		\infer2[\cutr]{\vdash \Gamma, \Delta, \Sigma}
	\end{prooftree}
	$\cutred$~
	\begin{prooftree}
		\hypo{\vdash \Gamma, \lA, \lB}
		\hypo{\vdash \Delta, \widecneg{\lA}}
		\infer2[\cutr]{\vdash \Gamma, \Delta, \lB}
		\hypo{\vdash \Sigma, \widecneg{\lB}}
		\infer2[\cutr]{\vdash \Gamma, \Delta, \Sigma}
	\end{prooftree}
	\qquad
	\begin{prooftree}
		\hypo{\vdash \Gamma}
		\infer1[\bot]{\vdash \Gamma, \bot}
		\infer0[1]{1}
		\infer2[\cutr]{\vdash \Gamma}
	\end{prooftree}
	$\cutred$~
	\begin{prooftree}
		\hypo{\vdash \Gamma}
	\end{prooftree}
	\\[5pt]
	\begin{adjustbox}{max width = \textwidth}
		\begin{prooftree}
			\infer0[\axrule]{\vdash \lA, \widecneg{\lA}}
			\hypo{\vdash \Gamma, \lA}
			\infer2[\cutrule]{\vdash \Gamma, \lA~}
		\end{prooftree}
		$\cutred$
		\begin{prooftree}
			\hypo{\vdash \Gamma, \lA}
		\end{prooftree}
		\qquad
		\begin{prooftree}
			\hypo{\vdash \Gamma}
			\infer1[\wnwrule]{\vdash \Gamma, \wn\lA}
			\infer0[\ocwrule]{\vdash \oc \widecneg{\lA}}
			\infer2[\cutr]{\vdash \Gamma}
		\end{prooftree}
		$\cutred$
		\begin{prooftree}
			\hypo{\vdash \Gamma}
		\end{prooftree}
		\qquad
		\begin{prooftree}
			\hypo{\vdash \Gamma, \lA}
			\infer1[\wndrule]{\vdash \Gamma, \wn \lA}
			\hypo{\vdash \Delta, \widecneg{\lA}}
			\infer1[\ocdrule]{\vdash \Delta, \oc\widecneg{\lA}}
			\infer2[\cutr]{\vdash \Gamma, \Delta}
		\end{prooftree}
		$\cutred$
		\begin{prooftree}
			\hypo{\vdash \Gamma, \lA}
			\hypo{\vdash \Delta, \widecneg{\lA}}
			\infer2[\cutr]{\vdash \Gamma, \Delta}
		\end{prooftree}
	\end{adjustbox}
	\\[5pt]
	\begin{adjustbox}{max width = 1.0\textwidth}
		\begin{prooftree}
			\hypo{\vdash \Gamma}
			\infer1[\wnwrule]{\vdash \Gamma, \wn \lA}
			\hypo{\vdash \Delta, \oc\widecneg{\lA}}
			\hypo{\vdash \Sigma, \oc\widecneg{\lA}}
			\infer2[\occrule]{\vdash \Delta, \Sigma, \oc\widecneg{\lA}}
			\infer2[\cutr]{\vdash\Gamma, \Delta, \Sigma}
		\end{prooftree}
		$\cutred$
		\begin{prooftree}[separation=0.8em, label separation =0.2em]
			\hypo{\vdash \Gamma}
			\infer1[\wnwrule]{\vdash \Gamma, \wn \lA}
			\hypo{\vdash \Delta, \oc\widecneg{\lA}}
			\infer2[\cutr]{\vdash \Gamma, \Delta}
			\infer1[\wnwrule]{\vdash \Gamma, \Delta, \wn\lA}
			\hypo{\vdash \Sigma, \oc\widecneg{\lA}}
			\infer[separation=0.3em]2[\cutr]{\vdash \Gamma, \Delta, \Sigma}
		\end{prooftree}
		\qquad
		\mbox{
			\begin{prooftree}
				\hypo{\vdash \Gamma, \wn\lA, \wn\lA}
				\infer1[\wncrule]{\vdash \Gamma, \wn \lA}
				\infer0[\ocwrule]{\vdash \oc\widecneg{\lA}}
				\infer2[\cutr]{\vdash \Gamma}
			\end{prooftree}
			$\cutred$~
			\begin{prooftree}
				\hypo{\vdash \Gamma, \wn\lA, \wn\lA}
				\infer0[\ocwrule]{\vdash \oc\widecneg{\lA}}
				\infer2[\cutr]{\vdash \Gamma, \wn\lA}
				\infer0[\ocwrule]{\vdash \oc\widecneg{\lA}}
				\infer2[\cutr]{\vdash \Gamma}
			\end{prooftree}
		}
	\end{adjustbox}
	\\[5pt]
	\begin{prooftree}
		\hypo{\vdash \Gamma, \lA}
		\infer1[\wndrule]{\vdash \Gamma, \wn \lA}
		\infer0[\ocwrule]{\vdash \oc\widecneg{\lA}}
		\infer2[\cutr]{\vdash\Gamma}
	\end{prooftree}
	$\cutred$~
	\begin{prooftree}
		\infer0[\zerorule]{\vdash\Gamma}
	\end{prooftree}
	\qquad
	\begin{prooftree}
		\hypo{\vdash \Gamma}
		\infer1[\wnwrule]{\vdash \Gamma, \wn \lA}
		\hypo{\vdash \Delta, \widecneg{\lA}}
		\infer1[\ocdrule]{\vdash \Delta, \oc\widecneg{\lA}}
		\infer2[\cutr]{\vdash\Gamma, \Delta}
	\end{prooftree}
	$\cutred$~
	\begin{prooftree}
		\infer0[\zerorule]{\vdash\Gamma, \Delta}
	\end{prooftree}
	\\[5pt]
	\begin{prooftree}
		\hypo{\vdash \Gamma, \lA}
		\infer1[\wndrule]{\vdash \Gamma, \wn \lA}
		\hypo{\vdash \Delta, \bclr{\oc\widecneg{\lA}}}
		\hypo{\vdash \Sigma, \rclr{\oc\widecneg{\lA}}}
		\infer2[\occrule]{\vdash \Delta, \Sigma, \oc\widecneg{\lA}}
		\infer2[\cutr]{\vdash\Gamma, \Delta, \Sigma}
	\end{prooftree}
	$\cutred$
	\begin{prooftree}[separation=0.8em, label separation =0.2em]
		\hypo{\vdash \Gamma, \lA}
		\infer1[\wndrule]{\vdash \Gamma, \bclr{\wn \lA}}
		\hypo{\vdash \Delta, \bclr{\oc\widecneg{\lA}}}
		\infer2[\cutr]{\vdash \Gamma, \Delta}
		\infer1[\wnwrule]{\vdash \Gamma, \Delta, \rclr{\wn\lA}}
		\hypo{\vdash \Sigma, \rclr{\oc\widecneg{\lA}}}
		\infer2[\cutr]{\vdash \Gamma, \Delta, \Sigma}
		
		\hypo{\vdash \Gamma, \lA}
		\infer1[\wndrule]{\vdash \Gamma, \rclr{\wn \lA}}
		\hypo{\vdash \Sigma, \rclr{\oc\widecneg{\lA}}}
		\infer2[\cutr]{\vdash \Gamma, \Sigma}
		\infer1[\wnwrule]{\vdash \Gamma, \Sigma, \bclr{\wn\lA}}
		\hypo{\vdash \Delta, \bclr{\oc\widecneg{\lA}}}
		\infer2[\cutr]{\vdash \Gamma, \Delta, \Sigma}
		
		\infer2[\sumrule]{\vdash \Gamma, \Delta, \Sigma}
	\end{prooftree}
	\\[5pt]
	\begin{prooftree}
		\hypo{\vdash \Gamma, \rclr{\wn\lA}, \bclr{\wn\lA}}
		\infer1[\wncrule]{\vdash \Gamma, \wn \lA}
		\hypo{\vdash \Delta, \widecneg{\lA}}
		\infer1[\ocdrule]{\vdash \Delta,\oc\widecneg{\lA}}
		\infer2[\cutr]{\vdash \Gamma, \Delta}
	\end{prooftree}
	$\cutred$
	\begin{prooftree}
		\hypo{\vdash \Gamma, \rclr{\wn\lA}, \bclr{\wn\lA}}
		\hypo{\vdash \Delta, \widecneg{\lA}}
		\infer1[\ocdrule]{\vdash \Delta, \bclr{\oc\widecneg{\lA}}}
		\infer2[\cutr]{\vdash \Gamma, \Delta, \rclr{\wn \lA}}
		\infer0[\ocwrule]{\vdash \rclr{\oc\widecneg{\lA}}}
		\infer2[\cutr]{\vdash \Gamma, \Delta}

		\hypo{\vdash \Gamma, \rclr{\wn\lA}, \bclr{\wn\lA}}
		\infer0[\ocwrule]{\vdash \bclr{\oc\widecneg{\lA}}}
		\infer2[\cutr]{\vdash \Gamma, \rclr{\wn\lA}}
		\hypo{\vdash \Delta, \widecneg{\lA}}
		\infer1[\ocdrule]{\vdash \Delta, \rclr{\oc\widecneg{\lA}}}
		\infer2[\cutr]{\vdash \Gamma, \Delta}
		
		\infer2[\sumrule]{\vdash \Gamma, \Delta}
	\end{prooftree}
	\\[5pt]
	\begin{adjustbox}{max width = \textwidth}
		\begin{prooftree}
			\hypo{\vdash \Gamma, \rclr{\wn\lA}, \bclr{\wn\lA}}
			\infer1[\wncrule]{\vdash \Gamma, \wn \lA}
			\hypo{\vdash \Delta, {\oc\widecneg{\lA}}}
			\hypo{\vdash \Sigma, {\oc\widecneg{\lA}}}	
			\infer2[\occrule]{\vdash \Delta, \Sigma, \oc\widecneg{\lA}}
			\infer2[\cutr]{\vdash \Gamma, \Delta, \Sigma}
		\end{prooftree}
		$\cutred$
		\begin{prooftree}
			\hypo{\vdash \Gamma, \rclr{\wn\lA}, \bclr{\wn\lA}}
			\hypo{~}
			\infer1[\axrule]{\vdash \bclr{\wn \lA}, \bclr{\oc\widecneg{\lA}}}
			\hypo{~}
			\infer1[\axrule]{\vdash \bclr{\wn \lA}, \bclr{\oc\widecneg{\lA}}}
			\infer2[\occrule]{\vdash \bclr{\wn \lA}, \bclr{\wn \lA}, \bclr{\oc\widecneg{\lA}}}
			\infer2[\cutr]{\vdash \Gamma, \rclr{\wn \lA}, \bclr{\wn \lA}, \bclr{\wn \lA}}
			\hypo{~}
			\infer1[\axrule]{\vdash \rclr{\wn \lA}, \rclr{\oc\widecneg{\lA}}}
			\hypo{~}
			\infer1[\axrule]{\vdash \rclr{\wn \lA}, \rclr{\oc\widecneg{\lA}}}
			\infer2[\occrule]{\vdash \rclr{\wn \lA}, \rclr{\wn \lA}, \rclr{\oc\widecneg{\lA}}}
			\infer2[\cutr]{\vdash \Gamma, \bclr{\wn \lA}, \bclr{\wn \lA}, \rclr{\wn \lA}, \rclr{\wn \lA}}
			\infer1[\wncrule]{\vdash \Gamma, \bclr{\wn \lA}, {\wn \lA}, \rclr{\wn \lA}}
			\hypo{\vdash \Delta, {\oc\widecneg{\lA}}}
			\infer2[\cutr]{\vdash \Gamma, \Delta, \bclr{\wn \lA}, \rclr{\wn \lA}}
			\infer1[\wncrule]{\vdash \Gamma, \Delta, {\wn \lA}}
			\hypo{\vdash \Sigma, {\oc\widecneg{\lA}}}
			\infer2[\cutr]{\vdash \Gamma, \Delta, \Sigma}
		\end{prooftree}
	\end{adjustbox}
	\caption{Key cases of cut-elimination rewriting rules for $\SDill$ sequent calculus (colors highlight $\cutr$-relations between formula occurrences).
%
}
	\label{fig:cut-elim}
\end{figure}


\paragraph{Cut-elimination.}
Despite its incoherence, $\SDill$ provides a fine analysis of resource consumption in cut-elimination.
Rewriting rules $\cutred$ for cut-elimination in $\SDill$ sequent calculus are defined in 
\Cref{fig:cut-elim}.
They are just the formulation in the sequent calculus formalism of the cut-elimination steps defined and studied in \cite{EhrhardRegnier06,Tranquilli09,Gimenez11} and \cite[Fig. 4]{Pagani09} within the interaction nets formalism.
We represent in \Cref{fig:cut-elim} only the \emph{key cases}, where the principal connectives in the cut formulas are dual (the pairs of dual connectives are $\ltens/\lpar$, $\lone/\lbot$, $\oc/\wn$). 
The way $\SDill$ deals with the \emph{commutative cases} is omitted since is analogous to usual sequent calculi.
With these cut-elimination steps it has been proved in \cite{EhrhardRegnier06,Pagani09,Gimenez11} that the rule $\cutrule$ is admissible in $\SDill$ (and even in $\DiLL$, \ie, the system $\SDill$ plus $\MELL$ promotion~$\prule$).

\begin{theorem}[Cut-elimination, \cite{EhrhardRegnier06,Pagani09,Gimenez11}]\label{thm:cutElimSequent}
	For every derivation $\pi$ in $\SDill$ with conclusion $\vdash \Gamma$, there exists a cut-free derivation $\pi'$ in $\SDill$ with conclusion $\vdash \Gamma$ such that 
	$\pi \cutred^* \pi'$.
\end{theorem}

Cut-elimination preserves atomic axioms: if $\pi \cutred \pi'$ and $\pi$ is $\eta$-expanded, then $\pi'$ is $\eta$-expanded.
Note that if $\pi \cutred \pi'$ with $\pi$ canonical then $\pi'$ is not necessarily canonical (\textit{e.g.} if in $\pi$ a cut $\wncrule/\ocdrule$ or $\wndrule/\ocwrule$ is above another rule), but $\pi'$ can be rewritten in a canonical form (see \Cref{fact:canon} above).
Indeed, $\SDill^-$ is not closed under cut-elimination: steps $\wncrule/\ocdrule$ or $\wndrule/\ocwrule$ create instances of the rule $\sumrule$ or $\zerorule$.

To explain the importance of the rules $\sumrule$ and $\zerorule$ as \emph{resource management}, we give an informal account of the cut-elimination steps in \Cref{fig:cut-elim} for the key cases involving $\oc/\wn$.
Roughly, they follow the ``law of supply and demand'' so as to be \emph{resource-sensitive}: in each slice no duplication or erasure is allowed.
The rules for $\wn$ ($\wnwrule$, $\wndrule$, $\wncrule$) \emph{ask for} a number of resources of type $\oc \lA$ ($0$, $1$, and the sum of the numbers asked by its premises, respectively), while the rules for $\oc$ ($\ocwrule$, $\ocdrule$, $\occrule$) \emph{supply} a number of resources of type $\oc \lA$ ($0$, $1$, and the sum of the numbers supplied by its premises, respectively). 
Cases:
\begin{enumerate}
	\item If the numbers of demanded and supplied resources match, the cut-elimination proceeds normally (see the steps $\wndrule/\ocdrule$ and $\wnwrule/\ocwrule$). 
	
	\item The step $\wncrule/\occrule$ is slightly more complex: 
	intuitively, it connects the dual premises of a $\wn$-contraction and of a $\oc$-contraction in all possible ways.
	
	\item The step $\wncrule/\ocwrule$ duplicates the rule $\ocwrule$, spreading the information that there are no available resources to the premises of $\wncrule$. 
	
	\item The step $\wndrule/\ocwrule$ represents a \emph{mismatch} in supply and demand: $\wn$-dereliction asks for a resource but $\oc$-weakening says that it is not available;
	the rule $\zerorule$ in the resulting derivation keeps track of this mismatch, as a sort of error in computation, and ensures that the conclusion is preserved.
	
	\item In the step $\wncrule/\ocdrule$, $\wn$-contraction represents 
	two possible demands for a resource, but according to  $\oc$-dereliction only one resource is available, so there is a \emph{non-deterministic choice} 
	on which request will be fed, the other one will receive a $\oc$-weakening; the rule $\sumrule$ has to be intended as a way to keep track of all possible choices, \emph{not} as a way to \emph{duplicate resources}; said differently, in the step $\wncrule/\ocdrule$ a derivation reduces to a pair of derivations (of slices, if we consider their canonical forms).
\end{enumerate}
By duality, the discussion above about resource management is similar for the steps $\wnwrule/\occrule$, $\wnwrule/\ocdrule$ and $\wndrule/\occrule$, respectively.
\Cref{fig:ex:cutelim1} provides an example of the cut-elimination procedure in $\SDill$.

\begin{figure}[!t]
\begin{center}
	\footnotesize
	\ebproofset{right label template=\tiny$\inserttext$, label separation=0.3em, separation = 1.2em}
	{\begin{prooftree}
			\hypo{}
			\ellipsis{$\pi_1$}{\vdash \lA, \lA}
			\infer1[\wndrule]{\vdash \lA, \wn \lA}
			\infer1[\wndrule]{\vdash \wn \lA, \wn \lA}
			\infer1[\wncrule]{\vdash \wn \lA}
			\hypo{}
			\ellipsis{$\pi$}{\vdash \Gamma,  \widecneg{\lA}}
			\infer1[\ocdrule]{\vdash \Gamma, \oc \widecneg{\lA}}
			\infer2[\cutrule]{\vdash \Gamma}
	\end{prooftree}}
	$\cutred~$
	{\begin{prooftree}[separation=1.2em, label separation=0.2em]
			\hypo{}
			\ellipsis{$\pi_1$}{\vdash \lA, \lA}
			\infer1[\wndrule]{\vdash \lA, \wn \lA}
			\infer1[\wndrule]{\vdash \wn \lA, \wn \lA}
			\hypo{}
			\ellipsis{$\pi$}{\vdash \Gamma, \widecneg{\lA}}
			\infer1[\ocdrule]{\vdash \Gamma, \oc\widecneg{\lA}}
			\infer2[\cutr]{\vdash \Gamma, \wn \lA}
			\infer0[\ocwrule]{\oc \widecneg{\lA}}
			\infer2[\cutr]{\vdash \Gamma}
			\hypo{}
			\ellipsis{$\pi_1$}{\vdash \lA, \lA}
			\infer1[\wndrule]{\vdash \lA, \wn \lA}
			\infer1[\wndrule]{\vdash \wn \lA, \wn \lA}
			\infer0[\ocwrule]{\oc \widecneg{\lA}}
			\infer2[\cutr]{\vdash \wn \lA}
			\hypo{}
			\ellipsis{$\pi$}{\vdash \Gamma, \widecneg{\lA}}
			\infer1[\ocdrule]{\vdash \Gamma, \oc\widecneg{\lA}}
			\infer2[\cutr]{\vdash \Gamma}
			\infer2[\sumrule]{\vdash \Gamma}
	\end{prooftree}}
	$\cutred$
	{\begin{prooftree}[separation=1.2em, label separation=0.2em]
			\hypo{}
			\ellipsis{$\pi_1$}{\vdash \lA, \lA}
			\infer1[\wndrule]{\vdash \lA, \wn \lA}
			\infer1[\wndrule]{\vdash \wn \lA, \wn \lA}
			\hypo{}
			\ellipsis{$\pi$}{\vdash \Gamma, \widecneg{\lA}}
			\infer1[\ocdrule]{\vdash \Gamma, \oc\widecneg{\lA}}
			\infer2[\cutr]{\vdash \Gamma, \wn \lA}
			\infer0[\ocwrule]{\oc \widecneg{\lA}}
			\infer2[\cutr]{\vdash \Gamma}
			\infer0[\zerorule]{\vdash \wn \lA}
			\hypo{}
			\ellipsis{$\pi$}{\vdash \Gamma, \widecneg{\lA}}
			\infer1[\ocdrule]{\vdash \Gamma, \oc\widecneg{\lA}}
			\infer2[\cutr]{\vdash \Gamma}
			\infer2[\sumrule]{\vdash \Gamma}
	\end{prooftree}}
	$\equiv~$
	{\begin{prooftree}[separation=1.2em, label separation=0.2em]
			\hypo{}
			\ellipsis{$\pi_1$}{\vdash \lA, \lA}
			\infer1[\wndrule]{\vdash \lA, \wn \lA}
			\infer1[\wndrule]{\vdash \wn \lA, \wn \lA}
			\hypo{}
			\ellipsis{$\pi$}{\vdash \Gamma, \widecneg{\lA}}
			\infer1[\ocdrule]{\vdash \Gamma, \oc\widecneg{\lA}}
			\infer2[\cutr]{\vdash \Gamma, \wn \lA}
			\infer0[\ocwrule]{\oc \widecneg{\lA}}
			\infer2[\cutr]{\vdash \Gamma}
	\end{prooftree}}
	$\cutred~$
	{\begin{prooftree}[separation=1.2em, label separation=0.2em]
			\hypo{}
			\ellipsis{$\pi_1$}{\vdash \lA, \lA}
			\infer1[\wndrule]{\vdash \lA, \wn \lA}
			\hypo{}
			\ellipsis{$\pi$}{\vdash \Gamma, \widecneg{\lA}}
			\infer2[\cutr]{\vdash \Gamma, \wn \lA}
			\infer0[\ocwrule]{\oc \widecneg{\lA}}
			\infer2[\cutr]{\vdash \Gamma}
	\end{prooftree}}
	$\cutred~$
	{\begin{prooftree}[separation=1.2em, label separation=0.2em]
			\hypo{}
			\ellipsis{$\pi_1$}{\vdash \lA, \lA}
			\infer1[\wndrule]{\vdash \lA, \wn \lA}
			\infer0[\ocwrule]{\oc \widecneg{\lA}}
			\infer2[\cutr]{\vdash  \lA}
			\hypo{}
			\ellipsis{$\pi$}{\vdash \Gamma, \widecneg{\lA}}
			\infer2[\cutr]{\vdash \Gamma}
	\end{prooftree}}
	$\cutred~$
	{\begin{prooftree}
			\infer0[\zerorule]{\vdash \lA}
			\hypo{}
			\ellipsis{$\pi$}{\vdash \Gamma, \widecneg{\lA}}
			\infer2[\cutr]{\vdash \Gamma}
	\end{prooftree}}
	$\equiv~$
	{\begin{prooftree}
			\infer0[\zerorule]{\vdash \Gamma}
	\end{prooftree}}
\end{center}
\caption{An example of the cut-elimination procedure in $\SDill$ sequent calculus.}
\label{fig:ex:cutelim1}
\end{figure}


It is worth comparing cut-elimination steps as defined for $\SDill$ sequent calculus (\Cref{fig:cut-elim}) and for $\SDill$ interaction nets (\cite[Sect.~2]{EhrhardRegnier06}, \cite[Fig.~4]{Pagani09}): symmetry and duality in the latter are lost in the former.

As there is no promotion rule $\prule$, in $\SDill$ transforming a derivation in $\SDill$ into one with atomic axioms does not commute with cut-elimination. 
For instance, derivation $\pi$ below reduces to $\pi'$ via cut-elimination; 
but derivation $\pi_\eta$ with atomic axioms, obtained from $\pi$ through $\eta$-expansion (the procedure described in the proof of \Cref{prop:atomic}) reduces to $\pi_\eta' \neq \pi'$ via cut-elimination.

\vspace*{-\baselineskip}
\begin{align*}
\small
\ebproofset{right label template=\tiny$\inserttext$, label separation=0.2em, separation = 1em}
	\pi =
	{\begin{prooftree}
	\infer0[\axrule]{\vdash \oc \widecneg{a}, \wn a}
	\infer0[\wnwrule]{\vdash \wn a}
	\infer2[\cutr]{\vdash \wn {a}}
	\end{prooftree}}
\ \
\cutred
\ \
	{\begin{prooftree}
	\infer0[\wnwrule]{\vdash \wn {a}}
	\end{prooftree}}
	= \pi'
	\qquad\qquad
	\pi_\eta = 
	{\begin{prooftree}
	\infer0[\axrule]{\vdash \widecneg{a},  a}
	\infer1[\wndrule]{\vdash \widecneg{a}, \wn a}
	\infer1[\ocdrule]{\vdash \oc \widecneg{a}, \wn a}
	\infer0[\wnwrule]{\vdash \wn a}
	\infer2[\cutr]{\vdash \wn {a}}
	\end{prooftree}}
	\ \
	\cutred
	\ \
	{\begin{prooftree}
	\infer0[\zerorule]{\vdash \wn {a}}
	\end{prooftree}}
	= \pi_\eta'
\end{align*}


\section{A Calculus of Structures for $\SDill$}\label{sec:SDDI}
%

In this section,
we introduce a deep inference system \cite{gug:str:01,brunnler:tiu:01,gug:SIS,tub:str:esslli19} suitable for $\SDill$, using the \emph{open deduction} formalism \cite{gug:gun:par:2010,TubellaGuglielmi18}.
As a first novelty, we internalize the rules $\zerorule$ and $\sumrule$ of $\SDill$ sequent calculus at the level of formulas.
In fact, derivations in deep inference systems have a sequence structure instead of the more general tree-like structure of sequent calculus: every rule in deep inference has exactly one premise, consisting of one formula. 
This because the meta-connectives for sequent composition (the comma) and sequent juxtaposition (derivation branching) are internalized by $\lpar$ and $\ltens$, respectively. 
To internalize the $\SDill$  meta-connective for sum, together with its unit, we introduce the (commutative and associative) binary connective $\lplus$ and its unit $\lzero$.\footnotemark
\footnotetext{Here, the new connective $\lplus$ has nothing to do with the additive disjunction $\oplus$ in $\LL$; and the unit $\lzero$ for $\lplus$ must not be confused with the additive unit $0$ for $\oplus$ in $\LL$.}
In this way, the rule $\sumrule$ branches the derivation tree with a connective, $\lplus$;
and similarly, the rule $\zerorule$ has is own premise, $\lzero$.
Thus, \emph{formulas} are~defined~by: 
\begin{center}
	$A, B \Coloneqq a \mid \cneg a \mid A\ltens B \mid A\lpar B \mid \lone \mid \lbot \mid \oc A \mid \wn B \mid \lzero \mid A \lplus B$
\end{center}
where $a, b, c, \dots$ range over the usual countably infinite set of propositional 
variables (so, a $\MELL$ formula as defined on p.~\pageref{sec:DiLL} is a formula 
with no occurrences of $\lplus$ and $\lzero$).
Formulas are identified up to the \emph{equivalence} $\fequiv$ defined as the least congruence on formulas generated by the relations in \eqref{eq:form}.
\begin{equation}\label{eq:form}
	\small
	\begin{gathered}
		\begin{array}{c@{\quad\quad }c@{\quad\quad}c}
			A\lpar B \fequiv B\lpar A				&
			A\ltens B \fequiv B\ltens A 			&
			A\lplus B \fequiv B \lplus A           
			\\
			A\lpar(B\lpar C)  \fequiv (A\lpar B)\lpar C 	&
			A\ltens(B\ltens C) \fequiv (A\ltens B)\ltens C &
			A \lplus (B\lplus C) \fequiv (A\lplus B)\lplus C
			\\
			A\lpar\lbot \fequiv A 	\quad\ \ \,		&
			A\ltens\lone  \fequiv A \quad\ \		&
			A \lplus \lzero \fequiv A \quad\ \  
		\end{array}
		\\[-2pt]
		\begin{array}{c@{\qquad\qquad}c@{\qquad}c}
			A\lpar(B\lplus C) \fequiv (A\lpar  B)\lplus (A\lpar C) 	
			&
			A\ltens (B\lplus C) \fequiv(A\ltens B)\lplus (A\ltens C )
			\\
			\oc(A\lplus B) \fequiv \oc A \lplus \oc B
			&
			\wn(A\lplus B) \fequiv \wn A \lplus \wn B
			\\    
			\lzero \ltens A \fequiv \lzero
			\qquad \quad
			\lzero \lpar A \fequiv  \lzero 
			& 
			\wn \lzero \fequiv \lzero 
			\qquad \quad
			\oc \lzero \fequiv \lzero
		\end{array}
	\end{gathered}
\end{equation}  

Some equivalences in \eqref{eq:form} correspond to well-known isomorphisms in $\MLL$.
With respect to $\fequiv$, the formula 
$\lzero$ is an annihilating element for all other connectives but $\lplus$, for which it is a neutral element; 
every connective other than $\lplus$ distributes over $\lplus$.

An \emph{additive normal} formula $A$ is a sum of $\MELL$ formulas,
\ie $A = \lA_1 \lplus \cdots \lplus \lA_n$ ($n \in \N$) where all $\lA_i$'s are $\MELL$ formulas ($A = \lzero$ for $n =0$).
For any $n \in \N$, we set $\lones=\underbrace{\lone+\cdots +\lone}_{n \textup{ times}}$. 
Note that, by the equivalences in \eqref{eq:form},  $\loness n\ltens \loness m =\loness{k}$ where $k = n \times m$.

A \emph{context} (\resp~$\MELL$ \emph{context}) $\ctx$ is a formula (\resp~$\MELL$ formula) with exactly one occurrence of the hole $\cons{\,}$ (which can be thought of as a special propositional variable).
We write $\ctxp{A}$ for the formula obtained from the context $\ctx$ by replacing its hole with the formula $A$.

\begin{remark}[Additive normal form]
	\label{rmk:additive}
	By definition of $\fequiv$, if $\ctx$ is a context, $\ctxp{A \lplus \lzero} \fequiv \ctxp{A}$ and $\ctxp{A \lplus B} \fequiv \ctxp{A} \lplus \ctxp B$. 
	If $\ctxp{}$ is a $\MELL$ context, $\ctxp{\lzero} \fequiv \lzero$.
	In general, 
	any formula $A$ has an additive normal formula $A'$ such that $A' \fequiv A$. 
	Indeed,
	equivalences in \eqref{eq:form} but the ones on the 
	first line can be oriented to define a terminating rewriting relation whose normal forms are \mbox{additive normal} formulas. 
\end{remark}

\paragraph{Derivations.}
We present deep inference derivations in the open deduction formalism, according to their ``\emph{synchronal}'' form, where there is maximal parallelism between inference steps (see \cite{gug:gun:par:2010,TubellaGuglielmi18}).

A \emph{deep inference system} $\cS$ is a set of unary inference rules. 
A \emph{derivation} $\dD$ from a \emph{premise} $B$ to a \emph{conclusion} $A$ in a deep inference  system $\cS$, noted $\vlderivation{\vlde{\dD}{\cS}{A}{\vlhy{B}}}$ or $\deriv{\dD}{B}{\cS}{A}$, is defined as follows:
\begin{itemize}
	\item (\emph{assumption}) a formula $A$ is a derivation (denoted by $A$) with premise and conclusion $A$;
	
	\item (\emph{horizontal composition}) 
	if for all $i \in \{1,2\}$ $\dD_i$ is a derivation 
	from $B_i$ 
	to $A_i$, 
	then for any $\symbol \in \set{\lpar, \ltens, \lplus}$, $\dD_1 \symbol \dD_2$ is a derivation 
	from  $B_1\symbol B_2$ 
	to $A_1\symbol A_2$ (see \eqref{eq:composition} below on the left);

	\item (\emph{vertical composition}) if 
	$\vlinf{\rho}{}{B_2}{A_1} \,\in \cS$
	and, for all $i \in\{1,2\}$, $\dD_i$ is a derivation from $B_i$ 	to $A_i$, then 
	$\dD_1 \circ_\ruler \dD_2$ 	is a derivation from $B_1$ to $A_2$ (see \eqref{eq:composition} below on the right).\footnotemark
	\footnotetext{\label{footnote:equivalence} We can write $\vlinf{\fequiv}{}{B}{A}$
		as a rule in a derivation if $A \fequiv B$, although formally its use is implicit as formulas are identified up to $\fequiv$.
		Said differently, our formalism for derivations can be seen as a \emph{calculus of structures} (in the sense of \cite{TubellaGuglielmi18}) that extends the equivalence relation $\fequiv$ on formulas\,---\,defined in \eqref{eq:form}\,---\,to an equivalence relation on derivations.
	}
\end{itemize}
\begin{small}
	\begin{equation}\label{eq:composition}
		\vlderivation{\odd{\odh{B_1\symbol B_2}}{\dD_1\symbol \dD_2}{A_1\symbol A_2}{\cS}}
		=
		\vlderivation
		{
			\odh{
				\od{\odd{\odh{B_1}}{\dD_1}{A_1}{\cS}}
				\symbol
				\od{\odd{\odh{B_2}}{\dD_2}{A_2}{\cS}}
			}
		}
		\quad\mbox{ for $\symbol \in \set{\lpar, \ltens, \lplus}$}
		\qquad\qquad
		\vlderivation{\odd{\odh{B_1}}{\dD_1 \circ_\ruler \dD_2}{A_2}{\cS}}
		=
		\vlinf
		{\ruler}
		{}
		{\vlderivation{\odd{\odh{B_2}}{\dD_2}{A_2}{\cS}}}
		{\vlderivation{\odd{\odh{B_1}}{\dD_1}{A_1}{\cS}}}
		\quad   
		\mbox{ for $\ruler \in \cS$}
	\end{equation}
\end{small}

We write $B \provevia{\cS} A$ if there is a derivation $\deriv{\dD}{B}{\cS}{A}$. 
	A rule 	$\vlinf{\rho}{}{A}{B}$
	is \emph{derivable} in $\cS$  if  $B \provevia{\cS} A$.

The system $\SDDIs$ is defined by the rules in \Cref{fig:rules}.
All rules in $\SDDIs$ have exactly one premise, as usual in deep inference.
The \emph{down-fragment} and \emph{up-fragment}\footnotemark
\footnotetext{\label{note:co-deep}Usually in the literature on deep inference, the dual rule $\upr{\mathsf{r}}$ of a rule $\downr{\mathsf{r}}$  is called ``co-$\mathsf{r}$''. 
	We avoid these names because they clash with the usual terminology in the literature on $\SDill$, see \Cref{note:co}.
} 
of $\SDDIs$ are the following sets of rules:
$$
\begin{array}{c}
	\SDDIdown=\downfrag
	\qquad
\SDDIup=\upfrag
\end{array}
$$

Note the 
up/down symmetry between $\SDDIdown$ and $\SDDIup$, and that 
$\SDDIs = \SDDIdown \cup \SDDIup$ with $\SDDIdown \cap \SDDIup = \set{\swir}$.
We set $\DMELL = \SDDIdown \setminus \set{\pdr, \zdr} $. 
Note that in a $\DMELL$ derivation only $\MELL$ formulas occur.

Roughly, rules in $\SDDIdown$ somehow mimic the ones in $\SDill \setminus \{\cutr\}$.
Rules in $\SDDIup$ are their duals, turning them upside down.
Derivations in $\DMELL$ correspond to cut-free slices in $\SDill^-$ \mbox{(see \Cref{thm:SDillSDDIs})}.

\begin{remark}[Deep]
	\label{rmk:deep}
	The idea of deep inference is that inference rules can be applied ``deep'' in any context: 
	in a deep inference system $\cS$, if $\vlinf{\ruler}{}{A}{B} \in \cS$ then, for any context $\ctx$, $\vlinf{\ruler}{}{\ctxp A}{\ctxp B}$ is derivable in $\cS$.
	Therefore, a derivation in $\cS$ can be seen as a finite \emph{sequence} of ``deep'' 
	rules: for instance,
	the derivation 
	$\small
	\vlsmash{\vlderivation{
		\odh{
			\od{\odi{\odh{\la}}{\wnddr}{\wn \la}{}}
			\ltens
			\od{\odi{\odh{\lb}}{\ocddr}{\oc \lb}{}}
		}
	}}
	$ 
	in $\SDDIdown$ (with parallel $\wnddr$ and $\ocddr$) can be ``sequenced'' as both
	\begin{prooftree}[label separation = 0.2em]
		\hypo{\la \ltens \lb}
		\infer[left label = \scriptsize$\ocddr$]1{\la \ltens \oc \lb}
		\infer[left label = \scriptsize$\wnddr$]1{\wn \la \ltens \oc \lb}
	\end{prooftree}
	and 
	\begin{prooftree}[label separation = 0.2em]
		\hypo{\la \ltens \lb}
		\infer[left label = \scriptsize$\wnddr$]1{\wn \la \ltens \lb}
		\infer[left label = \scriptsize$\ocddr$]1{\wn \la \ltens \oc \lb}
	\end{prooftree}. 
	Often we implicitly identify a derivation in a deep inference system $\cS$ with its sequenced presentations.
\end{remark}

\Cref{rmk:deep} means that every derivation can be also presented in \emph{sequenced} form, where inference rules are in a total order.
This is useful to have a notion of \emph{last} rule of a derivation, and to work by induction on the number of inference rules in a derivation.
Clearly, a derivation may have several sequenced forms, but each of them has the same number of inference rules as in its ``synchronal'' open deduction form.
A formal definition of how to sequence a derivation in oped deduction is in \cite{TubellaGuglielmi18}.

\begin{figure}[!t]
	\small
	\def\myskip{\hskip1.2em}  
	\hbox to\textwidth{\hfil 
		\begin{tabular}{c@{\myskip}c@{\myskip}c@{\myskip}c@{\myskip}c@{\myskip}c@{\myskip}c@{\myskip}c@{\myskip}c@{\myskip}c@{\myskip}c@{\myskip}c}
			$\vlinf{\aidr}{}{a\lpar \cneg a}{\lone}$
			&
			$\vlinf{\ocddr}{}{\oc \lA}{\lA}$
			&
			$\vlinf{\wnddr}{}{\wn \lA}{\lA}$
			&
			$\vlinf{\ocwdr}{}{\oc \lA}{\lone}$
			&
			$\vlinf{\wnwdr}{}{\wn \lA}{\lbot}$
			&
			$\vlinf{\occdr}{}{\oc \lA}{\oc \lA \ltens \oc \lA}$
			&
			$\vlinf{\wncdr}{}{\wn \lA}{\wn \lA \lpar \wn \lA}$
			&
			$\vlinf{\pdr}{}{\lA}{\lA \lplus \lA}$
			&
			$\vlinf{\zdr}{}{\lA}{\lzero}$
			&
			\multirow[b]{2}{*}{$\vlinf{\swir}{}{(A\ltens B)\lpar C}{A \ltens (B \lpar C)}$}
			\\[10pt]
			$\vlinf{\aiur}{}{\lbot}{a\ltens \cneg a}$
			&
			$\vlinf{\ocdur}{}{\lA}{\wn \lA}$
			&
			$\vlinf{\wndur}{}{\lA}{\oc \lA}$
			&
			$\vlinf{\ocwur}{}{\lbot}{\wn \lA}$
			&
			$\vlinf{\wnwur}{}{\lone}{\oc \lA}$
			&
			$\vlinf{\occur}{}{\wn \lA  \lpar \wn \lA}{\wn \lA}$
			&
			$\vlinf{\wncur}{}{\oc \lA \ltens \oc \lA}{\oc \lA}$
			&
			$\vlinf{\pur}{}{\lA \lplus \lA}{\lA}$
			&
			$\vlinf{\zur}{}{\lzero}{\lA}$
		\end{tabular}  
		\hfil}  
	\caption{The rules of the deep inference system $\SDDIs$ ($A,B,C$ are $\MELL$ formulas).
}
	\label{fig:rules}
\end{figure}

\begin{remark}[Big one]
	\label{rmk:one}
	For any formula $\lA$ and $n \in \N$, if $\deriv{\dD}{\lones}{\SDDI}{\lA}$, then there is a derivation $\deriv{\dD'}{\lone}{\SDDI}{\lA}$.
	Indeed, $\dD'$ is built from $\dD$ by adding one rule $\zur$ if $n = 0$, or $n-1$ rules $\pur$ if $n > 1$, on top of $\dD$.
\end{remark}

System $\SDDIs$ has only the \emph{atomic} introduction rules $\aidr$ and $\aiur$ (indeed $\la$ is a propositional variable in \Cref{fig:rules}): they can be seen as the atomic version of $\axrule$- and $\cutr$-rules of sequent calculus, respectively.
The \emph{non-atomic} versions of the rules $\aidr$ and $\aiur$ are respectively:
$$
\qquad
\vlsmash{\vlinf{\idr}{}{A\lpar \cneg A}{\lone}}
%
\qquad\qquad\qquad 
\vlsmash{\vlinf{\iur}{}{\lbot}{A\ltens \cneg A}}
\qquad\qquad\qquad 
\text{(where } A \text{ is a $\MELL$ formula)}
$$

However, the rules $\idr$ and $\iur$ are derivable in $\SDDIs$ (\Cref{lemma:idDer}).
Derivability of $\idr$ 
is analogous to the fact that a derivation can be transformed to one with atomic axioms in $\SDill$ sequent calculus (\Cref{prop:atomic}), but derivability of $\iur$ is a typical result in deep inference systems that does not have a corresponding result in the sequent calculus: it says that restricting cuts to an \emph{atomic} level is not limiting.

\begin{lemma}[Atomic axioms and atomic cuts]	
	\label{lemma:idDer}
	The rule $\idr$ is derivable in $\set{\aidr, \swir, \wnddr, \ocddr 	}$;
	and the rule $\iur$ is derivable in $\set{\aiur, \swir, \wndur, \ocdur	}$.
\end{lemma}
\begin{proof}
Concerning  $\idr$, the proof is by induction on 
the $\MELL$ formula $A$ in $\vlsmash{\vlinf{\idr}{}{A \lpar \cneg A}{\lone}}$~:
\begin{itemize}
	\item if $A=a$ is a propositional variable (and similarly if $\lA = \cneg{a}$), then {\small$\vlinf{\aidr}{}{a \lpar \cneg a}{ \lone}$}~;
	
	\item if $A= \lone$  (and similarly if $A = \lbot$), then 
	{\small \begin{prooftree}
			\hypo{\lone}
			\infer[left label = \scriptsize$\fequiv$, label separation = 0.2em]1{\lone \lpar \lbot}
	\end{prooftree}}~;
	
	\item if $A = B \ltens C$ (and similarly for $A = B \lpar C$), then 
	{\small
		$
		\od{
			\odi{
				\odh{
					\vlderivation{
						\vlin
						{\fequiv}
						{}
						{
							\od{\odd{\odh{\lone}}{\IH}{B \lpar \cneg B}{}} 
							\ltens 
							\od{\odd{\odh{\lone}}{\IH}{C \lpar \cneg C}{}}
						}
						{\vlhy{\lone}}
					}
				}
			}
			{2\times \swir}
			{ ((B \ltens C)) \lpar ((\cneg B \lpar \cneg C))}
			{}
		}
		$}~; 
	
	\item if $A=\oc B$ (and similarly if $A = \wn B$), then 
	{\small
		$
		\vlderivation{
			\odd 
			{\odh{\lone}}
			{\IH}
			{
				\boxed{\vlderivation{\odi{\odh{B}}{\ocddr}{\oc B}{}}}
				\lpar 
				\boxed{\vlderivation{\odi{\odh{\cneg B}}{\wnddr}{\wn \cneg B}{}}}
			}
			{}
		}
		$
	}~.
	
	
	
\end{itemize}
The proof for $\iur$ is dual, using $\aiur$, $\ocdur$ and $\wndur$ instead of $\aidr$, $\ocddr$ and $\wnddr$, respectively.
\end{proof}

%

The rule $\iur$ plays a special role in deep inference systems, as the $\cutr$ does in sequent calculi.
Thanks to $\fequiv$ and $\zdr$, it makes superfluous all the rules in $\DDIup$ (second line in \Cref{fig:rules}) but $\zur$ and $\swir$.
Note that $\aiur$ is not enough for that, because $\iur$ needs $\ocdur$ and $\wndur$ to be simulated by $\aiur$, as seen in \Cref{lemma:idDer}.

\begin{lemma}
	[Getting rid of up-rules via $\iur$ and $\zdr$]
	\label{lemma:puDer}	\label{thm:SDDIcut}\label{prop:up-derivability}
	Any rule $\upr{\ruler} \!\in\! \{\ocdur,\wndur,\occur,\wncur,\ocwur,\wnwur\}$	is derivable 
	in $\set{\downr{\ruler}, \iur, \idr, \swir}$; 
	the rule $\pur$ is derivable in \mbox{$\{ \zdr\}$}. 	
\end{lemma}
\begin{proof}
		For a rule 
		$\vldownsmash{\vlinf{\upr{\ruler}}{}{\cneg A}{\cneg B}}$
		with 
		$\upr{\ruler} \in \set{\ocdur,\wndur,\occur,\wncur,\ocwur,\wnwur}$, see \eqref{eq:derivable-rho}. 
		For the rule $\vldownsmash{\vlinf{\pur}{}{A \lplus A}{A}}$, see \eqref{eq:derivable-plus}.

		\begin{minipage}{0.4\textwidth}
			\small
		\begin{equation}
		\label{eq:derivable-rho}	
			{\vlderivation{
					\vlin{\fequiv}{}
					{
						\vlin{\swir}{}
						{\cneg A \lpar \od{\odn{\odn{A}{\downr{\ruler}}{B}{}\ltens \cneg B}{\iur}{\lbot}{}}}
						{\cneg A}
					}
					{	\vlin{\fequiv}{}{
							\od{\odn{\lone}{\idr}{\cneg A\lpar A}{}}
							\ltens\cneg B
						}
						{\vlhy{\cneg B}}
					}	
			}}
		\end{equation}
		\end{minipage}
	\qquad
	\begin{minipage}{0.4\textwidth}
		\begin{equation}
		\label{eq:derivable-plus}
		{\vlupsmash{\vlderivation{
				\vlinf
				{\fequiv}
				{}
				{
					\odn{\lzero}{\zdr}{A }{}
					\lplus A
				}
				{A}
		}}
		}
		\end{equation}
	\end{minipage}

\vspace{-\baselineskip}
\end{proof}

\section{Correspondence between \texorpdfstring{$\SDill$}{DiLL0} and \texorpdfstring{$\SDDIs$}{SDDI}}\label{sec:correspondence}



In this section we prove that $\SDDIs$  is a \emph{sound} and \emph{complete} proof system for $\SDill$ 
sequent calculus.
At first sight, this result is obvious because the rules $\zerorule$ in $\SDill$ and $\zdr$ in $\SDDI$ make everything provable. 
But the interest is to show that the fragments without $\zerorule$ and $\zdr$ correspond to each other.

If $\Gamma= \lA_1, \dots, \lA_n$ (with $n \in \N$) is a $\MELL$ sequent, we set $\toform \Gamma=\lA_1 \lpar \cdots \lpar \lA_n$ (so, $\toform \Gamma= \bot$ for $n = 0$).

\begin{figure}[t]
	\begin{center}
		\footnotesize
		$
		\vlinf{}{\axrule}{\vdash \la, \cneg \la}{}
		\translatesto \
		\vlinf{\aidr}{}{a \lpar \cneg a}{\lone}
		$
		\quad \ \ 
		$
		\vlderivation{\vliin{}{\cutr}{\vdash \Gamma, \Delta}{\vlpr{}{\pi_1}{\vdash \Gamma, \lA}}{\vlpr{}{\pi_2}{\vdash \Delta , \widecneg \lA}}} 
		\translatesto  \ 
		\vlderivation{
			\vlin
			{2\times \swir}
			{}
			{ \toform{\Gamma} \lpar \od{\odn{\lA \ltens \widecneg{\lA}}{\iur}{\bot}{}\lpar \toform{\Delta}}} 
			{
				\vlin
				{\fequiv}
				{}
				{
					\od{\od{\odd{\odh{\lones_1}}{\toform{\pi_1}}{\toform{\Gamma} \lpar \lA}{}}} 
					\ltens 
					\od{\od{\odd{\odh{\lones_2}}{\toform{\pi_2}}{\toform{\Delta} \lpar \widecneg\lA}{}}} 
				}
				{\vlhy{\loness{m}}}
			}
		}
		$
		\quad \ \
		$
		\vlderivation{\vliin{}{\ltens}{\vdash \Gamma, \lA \ltens \lB , \Delta}{\vlpr{}{\pi_1}{\vdash \Gamma, \lA}}{\vlpr{}{\pi_2}{\vdash \Delta , \lB}}} 
		\translatesto \
		\vlderivation{
			\vlin
			{2\times \swir}
			{}
			{\toform{\Gamma} \lpar (\lA \ltens \lB) \lpar \toform{\Delta}} 
			{
				\vlin
				{\fequiv}
				{}
				{
					\od{\od{\odd{\odh{\lones_1}}{\toform{\pi_1}}{\toform{\Gamma} \lpar \lA}{}}} 
					\ltens 
					\od{\od{\odd{\odh{\lones_2}}{\toform{\pi_2}}{\toform{\Delta} \lpar \lB}{}}} 
				}
				{\vlhy{\loness{m}}}
			}
		}
		$
		
		$
		\vlderivation{\vlin{}{\lpar}{\vdash \Gamma, \lA \lpar \lB}{\vlpr{}{\pi}{\vdash \Gamma, \lA, \lB}}}
		\translatesto \
		\vlderivation{\od{\odd{\odh{\lones}}{\toform{\pi}}{\toform{\Gamma} \lpar \lA \lpar \lB}{}} }
		$
		\qquad
		$
		\vlderivation{\vlin{}{\onerule}{\vdash \llone}{\vlhy{}}}
		\translatesto \ 
		\vlderivation{\lone}
		$
		\qquad
		$
		\vlderivation{\vlin{}{\lbot}{\vdash \Gamma, \llbot}{\vlpr{}{\pi}{\vdash \Gamma}}}
		\translatesto \ 
		\vlderivation{
			\vlin
			{\fequiv}
			{}
			{\toform{\Gamma} \lpar \lbot}
			{
				\odd{\odh{\lones}}{\toform{\pi}}{\toform{\Gamma}}{}
			}
		}
		$
		\qquad
		$
		\vlderivation{\vliin{}{\sumrule}{\vdash \Gamma}{\vlpr{}{\pi_1}{\vdash \Gamma}}{\vlpr{}{\pi_2}{\vdash \Gamma}}} 
		\translatesto 
		\vlderivation{\vlin{\pdr}{}{ \toform{\Gamma} } {\vlhy{\od{				\od{\odd{\odh{\lones_1}}{\toform{\pi_1}}{\toform{\Gamma}}{}}} \lplus \od{\od{\odd{\odh{\lones_2}}{\toform{\pi_2}}{\toform{\Gamma}}{}}} }}}
		$
		
		$
		\vlderivation{\vlin{}{\wnwrule}{\vdash \Gamma, \wn \lA}{\vlpr{}{\pi}{\vdash \Gamma}}}
		\translatesto
		\vlderivation{
			\vlin
			{\fequiv}
			{}
			{\toform \Gamma \lpar \od{\odn{\lbot}{\wnwdr}{\wn  \lA}{}}}
			{\odd{\odh{\lones}}{\toform{\pi}}{\toform{\Gamma}}{} }
		}
		$
		\qquad
		$
		\vlinf{}{\ocwrule}{\vdash \oc \lA}{} 
		\translatesto \
		\vlinf{\ocwdr}{}{\oc  \lA}{\lone}
		$
		\qquad
		$
		\vlderivation{\vlin{}{\ocdrule}{\vdash \Gamma, \oc \lA}{\vlpr{}{\pi}{\vdash \Gamma, \lA}}}
		\translatesto
		\vlderivation{\od{\odd{\odh{\lones}}{\toform{\pi}}{\toform{\Gamma} \lpar \odn{\lA}{\ocddr}{\oc \lA}{} }{}} }
		$
		\qquad
		$
		\vlderivation{\vlin{}{\wndrule}{\vdash \Gamma, \wn \lA}{\vlpr{}{\pi}{\vdash \Gamma, \lA}}}
		\translatesto
		\od{\odd{\odh{\lones}}{\toform{\pi}}{\toform{\Gamma} \lpar \odn{\lA}{\wnddr}{\wn \lA}{} }{}} 
		$
		
		$
		\vlderivation{\vlin{}{\wncrule}{\vdash \Gamma, \wn \lA}{\vlpr{}{\pi}{\vdash \Gamma,  \wn \lA, \wn \lA}}}
		\translatesto
		\vlderivation{\od{\odd{\odh{\lones}}{\toform{\pi}}{\toform{\Gamma} \lpar \odn{\wn \lA \lpar \wn \lA}{\wncdr}{\wn \lA}{} }{}}}
		$
		\qquad
		$
		\vlderivation{\vliin{}{\occrule}{\vdash \Gamma, \oc \lA , \Delta}{\vlpr{}{\pi_1}{\vdash \Gamma, \oc \lA}}{\vlpr{}{\pi_2}{\vdash \Delta ,\oc \lA}}} 
		\translatesto \ 
		\vlderivation{
			\vlin
			{2\times \swir}
			{}
			{ \toform{\Gamma} \lpar \od{\odn{\oc \lA \ltens \oc{\lA}}{\occdr}{\oc \lA}{}\lpar \toform{\Delta}} } 
			{
				\vlin
				{\fequiv}
				{}
				{
					\vlderivation{
						\od{\od{\odd{\odh{\lones_1}}{\toform{\pi_1}}{\toform{\Gamma} \lpar \oc\lA}{}}}
					} 
					\ltens 
					\vlderivation{
						\od{\od{\odd{\odh{\lones_2}}{\toform{\pi_2}}{\toform{\Delta} \lpar \oc\lA}{}}}
					} 
				}
				{\vlhy{\loness{m}}}	
			}
		}
		$
		\qquad
		$
		\vlderivation{\vlin{}{\zerorule}{\vdash \Gamma}{\vlhy{}}}
		\translatesto \
		\vlderivation{\vlin{\zdr}{}{\toform \Gamma}{\vlhy{\lzero}}}
		$
	\end{center}
	\caption{
		Translation of $\eta$-expanded $\SDill$ sequent calculus 
		derivations 
		into  $\DDIdown \cup \{\iur\}$ derivations ($m = n_1 \times n_2$).
	}
	\label{fig:seqToDeep}
\end{figure}


\begin{theorem}[Completeness]\label{thm:SDilltoSDDIs}
	Let $\Gamma$ be a $\MELL$ sequent.
	If $\provevia{\SDill} \Gamma $ (\resp\!\!$\provevia{\SDill\setminus \set{\cutr}} \Gamma $), then  $\lones \!\provevia{\DDIdown\cup\set{\iur}}\! \toform \Gamma$ (\resp $\lones \!\provevia{\DDIdown}\! \toform \Gamma$) for some $n \in \N$,	and $\lone \provevia {\SDDIs} \toform \Gamma$.
	Moreover, 
if $\provevia{\SDill^-} \Gamma $ then $\lone \!\provevia{\DDIdown_{-}\cup\set{\iur}}\! \toform \Gamma$;
if $\provevia{\SDill^-\setminus\set{\cutr}} \Gamma $ then $\lone \!\provevia{\DDIdown_{-}}\! \toform \Gamma$;
if $\provevia{\set{\zerorule}} \Gamma$ then $\lzero \provevia{\set{\zdr}} \toform{\Gamma}$.
\end{theorem}
\begin{proof}

	If we show that $ \provevia{\SDill} \Gamma$ implies $\lones \provevia {\DDIdown \cup\set{\iur}} \toform\Gamma$, then $\lones \provevia{\SDDIs} \toform\Gamma$ by \Cref{lemma:idDer}, and thus $\lone \provevia{\SDDIs} \toform{\Gamma}$ by \Cref{rmk:one}. 
	So, let $\pi$ be a derivation of $\vdash \!\Gamma$ in $\SDill$. 
	By  \Cref{prop:atomic} we can assume that $\pi$ is $\eta$-expanded.
	By induction on $\pi$, we define a derivation  $\deriv{\toform \pi}{\lones}{\DDIdown\cup\set{\iur}}{\toform{\Gamma}}$ 
	(for some $n \in \N$) as shown in \Cref{fig:seqToDeep}.
	According to this translation, if $\pi$ is in $\SDill\setminus\set{\cutr}$ (\resp in $\SDill^-\setminus\set{\cutr}$; in $\SDill^-$; in $\set{\zerorule}$) then $\toform{\pi}$ is in $\DDIdown$ (\resp in $\DMELL$ and $n = 1$; in $\DMELL \cup \{\iur\}$ and $n=1$; in $\set{\zdr}$~and~$n = 0$).
\end{proof}


Completeness (\Cref{thm:SDilltoSDDIs}) says that slices of a derivation in $\SDill$ (\ie~derivations in $\SDill^-$) with atomic axioms correspond to derivations in $\DMELL \cup \{\iur\}$ (and so in $\SDDIs \setminus \{\pur,\zur,\pdr,\zdr\}$, by rewriting $\iur$ according to \Cref{lemma:idDer}) with only $\MELL$ formulas, via the translation $\toform{\cdot}$ defined in 
\Cref{fig:seqToDeep}.
Soundness (\Cref{thm:SDDIstoSDill}) says somehow that the converse holds too.

\def\cotranslatesto{\to}
\begin{figure}[t]
	\footnotesize
	$$
	\vlinf{\aidr}{}{a\lpar \cneg a}{\lone}
	\ \cotranslatesto \
	\vlderivation{\vlin{}{\botrule}{\vdash \la\lpar \cneg \la, \llbot}{\vlin{}{\lpar}{\vdash \la\lpar \cneg \la}{\vlin{}{\axrule}{\vdash \la, \cneg \la}{\vlhy{}}}}}
	\qquad
	\vlinf{\iur}{}{\lbot}{\lA \ltens \widecneg \lA}
	\ \cotranslatesto \ 
	\vlderivation{\vlin{}{\botrule}{\vdash \widecneg\lA\lpar  \lA, \llbot}{\vlin{}{\lpar}{\vdash\widecneg \lA\lpar  \lA}{\vlin{}{\axrule}{\vdash  \widecneg \lA, \lA}{\vlhy{}}}}}
	\qquad 
	\vlinf{\ocwdr}{}{\oc \lA}{\lone}
	\ \cotranslatesto \ 
	\vlderivation{\vlin{}{\botrule}{\vdash \oc \lA, \llbot}{\vlin{}{\ocwrule}{\vdash\oc \lA}{\vlhy{}}}}
	\qquad
	\vlinf{\wnwdr}{}{\wn \lA}{\lbot}
	\ \cotranslatesto \ 
	\vlderivation{\vlin{}{\wnwrule}{\vdash \llone, \wn \lA}{\vlin{}{\onerule}{\vdash \llone}{\vlhy{}}}}
	$$
	
	\vspace{-\baselineskip}
	$$
	\vlinf{\wnddr}{}{\wn \lA}{\lA}
	\ \cotranslatesto \ 
	\vlderivation{
		\vlin{}{\wndrule}{\vdash \cneg \lA, \wn \lA}{ \vlin{}{\axrule}{\vdash  \widecneg \lA, \lA}{\vlhy{}}}
	}
	\qquad
	\vlinf{\ocddr}{}{\oc \lA}{\lA}
	\ \cotranslatesto \ 
	\vlderivation{
		\vlin{}{\ocdrule}{\vdash \widecneg \lA, \oc \lA}{\vlin{}{\axrule}{\vdash  \widecneg \lA, \lA}{\vlhy{}}}
	}
	\qquad
	\vlinf{\swir}{}{(\lA\ltens \lB)\lpar \lC}{\lA \ltens (\lB \lpar \lC)}
	\ \cotranslatesto \
	\vlderivation{
		\vliq
		{}
		{2\times \lpar}
		{\vdash(\lA\ltens \lB)\lpar \lC, \widecneg \lA \lpar  (\widecneg \lB \ltens \widecneg \lC)}
		{
			\vliin
			{}
			{\ltens}
			{\vdash\lA \ltens \lB, \lC, \widecneg \lA,  \widecneg \lB \ltens \cneg \lC}
			{\vlin{}{\axrule}{\vdash  \widecneg \lA, \lA}{\vlhy{}}}
			{
				\vliin
				{}
				{\ltens}
				{\vdash\lB, \lC, \widecneg \lB \ltens \widecneg \lC}
				{\vlin{}{\axrule}{\vdash  \widecneg \lB, \lB}{\vlhy{}}}
				{\vlin{}{\axrule}{\vdash  \widecneg \lC, \lC}{\vlhy{}}}
			}
		}
	}
	$$
	
	\vspace{-\baselineskip}
	$$
	\vlinf{\wncdr}{}{\wn \lA}{\wn \lA \lpar \wn \lA}
	\ \cotranslatesto \ 
	\vlderivation{
		\vlin
		{}
		{\wncrule}
		{\vdash \oc \widecneg \lA \ltens \oc \widecneg \lA, \wn \lA}
		{
			\vliin
			{}
			{\ltens }
			{\vdash \oc \widecneg \lA \ltens \oc \widecneg \lA, \wn \lA, \wn\lA}
			{
				\vlin
				{}
				{\axrule}
				{\vdash \oc  \widecneg \lA, \wn \lA}
				{\vlhy{}}
			}
			{
				\vlin
				{}
				{\axrule}
				{\vdash  \oc \widecneg \lA, \wn\lA}
				{\vlhy{}}
			}
		}
	}
	\qquad
	\vlinf{\occdr}{}{\oc \lA}{\oc \lA \ltens \oc \lA}
	\ \cotranslatesto \
	\vlderivation{
		\vlin{}{\lpar }{\vdash \oc  \lA , \wn\widecneg \lA\lpar \wn \widecneg\lA}{
			\vliin{}{\occrule }{\vdash  \wn \widecneg\lA, \wn \widecneg\lA, \oc  \lA }{\vlin{}{\axrule}{\vdash  \wn\widecneg \lA, \oc \lA}{\vlhy{}}}{\vlin{}{\axrule}{\vdash  \wn \widecneg \lA, \oc \lA}{\vlhy{}}}
	}}
	$$
	
	\vspace{-\baselineskip}
	$$
	\vlinf
	{\fequiv}
	{}
	{(\lA \lpar \lB)\lpar \lC}
	{\lA \lpar (\lB \lpar \lC)}
	\ \cotranslatesto \
	\vlderivation{
		\vliq
		{}
		{2\times \lpar}
		{\vdash \widecneg\lA \ltens (\widecneg \lB \ltens \widecneg \lC), (\lA \lpar \lB) \lpar \lC}
		{
			\vliin
			{}
			{\ltens}
			{\vdash \widecneg\lA \ltens (\widecneg\lB \ltens \widecneg \lC), \lA, \lB, \lC}
			{\vlin{}{\axrule}{\vdash  \widecneg \lA, \lA}{\vlhy{}}}
			{
				\vliin
				{}
				{\ltens}
				{\vdash\lB, \lC, \widecneg \lB \ltens \widecneg \lC}
				{\vlin{}{\axrule}{\vdash  \widecneg \lB, \lB}{\vlhy{}}}
				{\vlin{}{\axrule}{\vdash  \widecneg \lC, \lC}{\vlhy{}}}
			}
		}
	}
	\qquad
	\vlinf
	{\fequiv}
	{}
	{\lA \lpar (\lB \lpar \lC)}
	{(\lA \lpar \lB)\lpar \lC}
	\ \cotranslatesto \
	\vlderivation{
		\vliq
		{}
		{2\times \lpar}
		{\vdash (\widecneg{\lA} \ltens \widecneg\lB) \ltens \widecneg{\lC}, \lA \lpar  (\lB \lpar \lC)}
		{
			\vliin
			{}
			{\ltens}
			{\vdash(\widecneg{\lA} \ltens \widecneg\lB) \ltens \widecneg{\lC}, \lA, \lB, \lC}
			{
				\vliin
				{}
				{\ltens}
				{\vdash\lB, \lA, \widecneg \lA \ltens \widecneg \lB}
				{\vlin{}{\axrule}{\vdash  \widecneg \lA, \lA}{\vlhy{}}}
				{\vlin{}{\axrule}{\vdash  \widecneg \lB, \lB}{\vlhy{}}}
			}
			{
				\vlin
				{}{\axrule}{\vdash  \widecneg \lC, \lC}{\vlhy{}}
			}
		}
	}
	\qquad
	$$
	
	\vspace{-\baselineskip}
	$$
	\vlinf{\fequiv}{}{\lB \lpar \lA}{\lA \lpar \lB}
	\ \cotranslatesto \ 
	\vlderivation{
		\vlin
		{}
		{\lpar}
		{\vdash \widecneg \lA \ltens \widecneg \lB, \lB \lpar \lA}
		{
			\vliin
			{}
			{\ltens }
			{\vdash  \widecneg \lA \ltens \widecneg \lB, \lB, \lA}
			{
				\vlin
				{}
				{\axrule}
				{\vdash  \widecneg \lA,  \lA}
				{\vlhy{}}
			}
			{
				\vlin
				{}
				{\axrule}
				{\vdash  \widecneg \lB, \lB}
				{\vlhy{}}
			}
		}
	}
	\qquad
	\vlinf{\fequiv}{}{\lA \lpar \lbot}{\lA}
	\ \cotranslatesto \ 
	\vlderivation{
		\vlin
		{}
		{\lpar}
		{\vdash \widecneg \lA, \lA \lpar \lbot}
		{
			\vlin
			{}
			{\lbot}
			{\vdash  \widecneg \lA, \lA, \lbot}
			{
				\vlin
				{}
				{\axrule}
				{\vdash  \widecneg \lA,  \lA}
				{\vlhy{}}
			}
		}
	}
	\qquad
	\vlinf{\fequiv}{}{\lA}{\lA \lpar \lbot}
	\ \cotranslatesto \ 
	\vlderivation{
		\vliin
		{}
		{\lpar}
		{\vdash \widecneg \lA \ltens \lone, \lA}
		{
			\vlin
			{}
			{\axrule}
			{\vdash  \widecneg \lA,  \lA}
			{\vlhy{}}
		}
		{
			\vlin
			{}
			{\lone}
			{\vdash  \lone}
			{\vlhy{}}
		}
	}	
	$$	
	\caption{Interpretation of the rules in $\DMELL \cup \{\iur\}$ and of $\fequiv$ as derivations in $\SDill$ sequent calculus.}
	\label{fig:ddiToDill}
\end{figure}

\begin{theorem}[Soundness]\label{thm:SDDIstoSDill}
	For any $\MELL$ sequent $\Gamma$ and any $n \in \N$, if $\lones \provevia{\SDDIs} \toform \Gamma$
	then  $\provevia{\SDill} \Gamma$;
	and more precisely,
	if 
	$ \lone \!\!\!\provevia{\DMELL \cup \{\iur\}}\!\!\! \toform{\Gamma} $, then $\provevia{\SDill^-} \Gamma$.
\end{theorem}
\begin{proof}
%

	Clearly, for any $\MELL$ sequent $\Gamma$, there is a derivation in $\SDill$ sequent calculus:
	\begin{prooftree}[label separation = 0.3em]
		\infer0[\footnotesize$\zerorule$]{\vdash \Gamma}
	\end{prooftree}.
	
	
	Let us assume that we have a derivation $\dD$ in $\DMELL \cup \{\iur\}$  from  $\lone$ to $\toform \Gamma$.
	To define the derivation of $\vdash \Gamma$ in $\SDill^-$, we consider the formulas occurring in $\dD$ (which actually are $\MELL$ formulas) not up to $\fequiv$, so when $\fequiv$ is required, its use is made explicit 
	as if it were an inference rule (see also \Cref{footnote:equivalence}). 
	For any $\ruler \in \DMELL \cup \{\iur, \fequiv\}$, if $\vldownsmash{\vlinf{\ruler}{}{A}{B}}$
	then $ \provevia{\SDill^-} \widecneg\lB,\lA$ as shown in \Cref{fig:ddiToDill}.
	Since $\lA$ and $\lB$ are $\MELL$ formulas, many equivalences in \eqref{eq:form} cannot occur in $\dD$.
	The cases for $\ruler = {\fequiv}$ corresponding to $\lA \otimes \lB \fequiv \lB \otimes \lA$ and $\lA \otimes (\lB \otimes \lC) \fequiv (\lA \otimes \lB) \otimes \lC$ and $\lA \otimes \lone \fequiv \lA$ are omitted in \Cref{fig:ddiToDill} as they are analogous to the ones for $\lpar$.
	
	Consider $\dD$ as sequenced (\Cref{rmk:deep}).
	By induction on the $\MELL$ context $\ctx$, we prove that if 
	$\vlsmash{\vlinf{\rho}{}{\ctxp{\lA}}{\ctxp{\lB}}}$ occurs in $\dD$,
	then  $ \provevia{\SDill^-} \widecneg{\ctxp{\lB}}, \ctxp{\lA}$.
	We have just shown the case $\ctx = \cons$.
	Other cases:
$$
	\small
	\begin{array}{c@{\qquad}|@{\qquad}c}
	\ctxp{\lA} = \lD \lpar \ctxpp{\lA}
	\text{ (or similarly }
	\ctxp{\lA} = \lD \ltens \ctxpp{\lA}\text{)}
	&
	\ctxp{\lA}= \oc \ctxpp{\lA}
	\text{ (or smilarly }
	\ctxp{\lA}= \wn \ctxpp{\lA}\text{)}
	\\
	\vlderivation{
		\vlin
		{}
		{\lpar }
		{\vdash \cneg \lD \ltens \widecneg{ \ctxpp{\lB}}, \lD \lpar \ctxpp{\lA}}
		{
			\vliin
			{}
			{\ltens}
			{\vdash \cneg \lD \ltens \widecneg{ \ctxpp{\lB}}, \lD, \ctxpp{\lA}}
			{\vlin{}{\axrule}{\vdash \lD,\cneg \lD}{\vlhy{}}}
			{
				\vlpr{\IH}{}{\vdash \widecneg{ \ctxpp{\lB}}, \ctxpp{\lA} }
			}
		}
	}
&
	\vlderivation{
		\vlin
		{}
		{\ocdrule}
		{\vdash \wn\widecneg{ \ctxpp{\lB}}, \oc\ctxpp{\lA}}
		{
			\vlin
			{}
			{\wndrule}
			{\vdash \wn\widecneg{ \ctxpp{\lB}}, \ctxpp{\lA}}
			{
				\vlpr{\IH}{}{\vdash \widecneg{ \ctxpp{\lB}}, \ctxpp{\lA} }
			}
		}
	}
	\end{array}
$$

	We define a derivation of $\vdash \toform{\Gamma}$ in $\SDill^-$ by induction on the number of rules in $\dD$ as follows:
	\begin{center} 
		\footnotesize
		$
		\lone
		\ \cotranslatesto \
		\vlinf{}{\onerule}{\vdash \llone}{}
		\qquad 
		\qquad
		\qquad\qquad\qquad
		\vlsmash{\vlderivation{\vlin{\ruler}{}{\toform \Gamma}{\vlde{}{\DMELL \cup \set{\iur}}{\toform \Delta}{\vlhy{\lone}}}}}
		\quad 
		\cotranslatesto
		\quad
		\vlderivation{\vliin{}{\cutr}{\vdash \toform{\Gamma}}{\vlpr{\IH\,}{\SDill^-}{\vdash \toform{\Delta}}}{\vlpr{}{\SDill^-}{\vdash \widecneg {\toform{\Delta}}, \toform{\Gamma}}}} 
		$
	\end{center}
	By reversibility of $\lpar$ (if $\provevia{\SDill^-} \lA \lpar \lB$ then $\provevia{\SDill^-} \lA, \lB$), we have $\provevia{\SDill^-} \Gamma$.
\end{proof}

Let us sum up the correspondence between $\SDill$ sequent calculus and $\SDDIs$ deep inference system.

\begin{theorem}[Sequent calculus vs.~deep inference]
	\label{thm:SDillSDDIs}	

	Let $\Gamma$ be a $\MELL$ sequent.
	\begin{enumerate}
		\item\label{p:SDillSDDIs-cut}\emph{$\SDill$ vs.~$\SDDIs$}: $\provevia{\SDill} \Gamma $ if and only if $\lone \provevia{\SDDIs} \toform{\Gamma}$. 
		
		\item\label{p:SDillSDDIs-cut-free}\emph{$\SDill$ cut-free vs.~$\DDIdown$}: $\provevia{\SDill \setminus \set\cutr} \Gamma  $ if and only if $\lones \provevia {\DDIdown} \toform \Gamma$ for some $n \in \N$.
		
		\item\label{p:slice-cut-free}\emph{$\SDill^-$ cut-free vs.~$\DMELL$}: if $\provevia{\SDill^- \setminus \set\cutr} \Gamma  $ 
		then $\lone \provevia {\DMELL} \toform \Gamma$.
		
	\end{enumerate}
\end{theorem}

\begin{proof}
	\begin{enumerate}
		\item For $\Rightarrow$, by completeness (\Cref{thm:SDilltoSDDIs}); for $\Leftarrow$, by soundness (\Cref{thm:SDDIstoSDill}).
		\item [2.--3.] For $\Rightarrow$,	
		see \Cref{thm:SDilltoSDDIs}.
		For $\Leftarrow$ (only for \Cref{p:SDillSDDIs-cut-free}), by \Cref{thm:cutElimSequent,thm:SDillSDDIs}.\ref{p:SDillSDDIs-cut}.
		\qedhere
	\end{enumerate}
\end{proof}

%
%


\section{Normalization in $\SDDIs$}
\label{sect:normalization}

In this section we define a \emph{standard} form for derivations in $\SDDIs$ and a \emph{normalization} procedure to obtain a ``cut-free'' standard  derivation in $\DDIdown$ for any formula $A$ provable in $\SDDIs$.
	The usual approach to prove normalization in deep inference system relies on the 
\emph{splitting technique}  \cite{gug:SIS,horne2019morgan,acc:hor:str:LBF,gug:str:02,str:phd}.
However, the presence in our syntax of the connective $\lplus$ and its unit $\lzero$ prevents us to use Guglielmi and Tubella's normalization result \cite{tubella:phd,TubellaGuglielmi18}, 
which covers and generalizes the splitting proofs.
This is mainly due to the fact that 
$\lzero$ is an absorbing element for $\ltens$, $\lpar$, $\wn$ and $\oc$, together with 
the distributivity over $\lplus$.

For this reason, 
following \cite{str:phd,str:MELLinCoS},
we define the normalization process in terms of \emph{rule permutations}, which play the same role as cut-elimination steps in $\SDill$.
%
In some cases, their definition relies on the rules for the connective $\lplus $ and its unit $\lzero$. This behavior is coherent with the dynamics of cut-elimination in $\SDill$ \cite{EhrhardRegnier06,Pagani09,Gimenez11}, where the rules $\sumrule$ and $\zerorule$ step in to deal with non-deterministic choices or mismatches between ``supply and demand'' (see \Cref{sec:DiLL}).
Interestingly, these permutations in $\SDDIs$ mimic the elegant symmetries of cut-elimination steps as defined for interaction nets \cite[Fig. 4]{Pagani09}, instead of the awkward rewrite rules defined for the sequent calculus (\Cref{fig:cut-elim}).

%
The fact that  the syntax for $\SDDI$ is more flexible and symmetric than the sequent calculus, and internalizes the connective $\lplus$ and its unit $\lzero$, allows for a more fine-grained analysis of the normalization process than in $\SDill$. 
In particular, we can distinguish three kinds of rule permutations corresponding to three distinct phases in normalizing: $\MLL$ \emph{cut-elimination steps} (involving $\aiur$, $\aidr$ and $\swir$ only), \emph{resource management steps} (involving the $\wn$- and the $\oc$-rules only) and \emph{slice operations} (the process of duplicating or removing a slice, which is less evident in the $\SDill$ sequent calculus and $\SDill$ interaction nets).

\begin{definition}[Permutation]
	In $\SDDI$, a rule $\ruler$ \emph{permutes over}  a rule $\rules$ (or $\rules$ \emph{permutes under} $\ruler $) if,
	for any derivation 
	{\small $
	\vldownsmash
	{\vlderivation{\vlin{\ruler}{}{A}{\vlin{\rules}{}{B}{\vlhy C}}}}$}\,, one of the following holds:
	
	\vspace{-\baselineskip}
	$$
	\def\myskip{\hskip3em}
	\begin{array}{c@{\; ; \myskip\!\!\!\!\!\!}c@{\; ;  \myskip\!\!\!\!}c@{\; ;  \myskip\!\!\!\!}c@{\; ;  \myskip\!\!\!\!\!}c}
		\quad{\small A=C}
		&
	\quad{\small\vlinf{\rules}{}A C}
	&
	{\small \vlinf{\ruler}{}A C}
	&
	{\small\vlsmash{\vlderivation{\vlin{\rules}{}{A}{\vlin{\ruler}{}{B'}{\vlhy C}}}}} \; \; \mbox{    for some formula }B'
	&
	{\small \vlsmash{\vlderivation{\vlin{\rules}{}{A}{\vlin{2\times \ruler}{}{B'}{\vlhy C}}}}} \; \; \mbox{    for some formula }B'.
	\end{array}
$$
	
	A rule $\ruler$ \emph{permutes over} a rule $\rules$  (or $\rules$ \emph{permutes under} $\ruler $) \emph{by} a set of rules $\cS$ if $\ruler$ permutes over $\rules$, 
	or 
	$$
	\def\myskip{\hskip1.5em}
	\begin{array}{c@{\myskip}c@{\; , \ \mbox{ or } \ }c@{\; . \myskip}c@{\; .  \myskip}c}
\mbox{for any derivation $
	\vldownsmash
	{\small \vlderivation{\vlin{\ruler}{}{A}{\vlin{\rules}{}{B}{\vlhy C}}}}$\,,  one of the following holds:}
		&
{\small \vlderivation{\vlde{}{\cS}{A}{\vlhy C}}}
		&
{\small \vlupsmash{\vlderivation{\vlin{\rules}{}{A}{\vlde{}{\cS}{B'}{\vlin{\ruler}{}{B''}{\vlhy C}}}}}} \quad\mbox{ for some formulas }  B', B''
	\end{array}
	$$
\end{definition}

Roughly, permuting $\rules$ under $\ruler$ means that $\rules$ can be pushed below $\ruler$ in a derivation with same premise and conclusion. 
In this operation, $\ruler$ or $\rules$ might disappear or other rules might appear in between.
The definition of rule permutation is asymmetric: 
two $\rho$'s can be above one $\sigma$, but not two~$\sigma$'s~below~one~$\rho$.

We call \emph{trivial} the rule permutations identified by the open deduction syntax, such as the one below.
$$
\small
\vlderivation{
\vlin{\ruler_2}{}{A_1\ltens A_2}{\vlin{\ruler_1}{}{A_1\ltens B_2}{\vlhy{B_1\ltens B_2}}}
}
=
\od{\od{\odi{\odh{B_1}}{\ruler_1}{A_1}{}}\ltens \od{\odi{\odh{B_2}}{\ruler_2}{A_2}{}}}
=
\vlderivation{
	\vlin{\ruler_1}{}{A_1\ltens A_2}{\vlin{\ruler_2}{}{B_1\ltens A_2}{\vlhy{B_1\ltens B_2}}}
}
$$

The following lemma is analogous to canonicity (\Cref{fact:canon}) 
for $\SDill$ sequent calculus.
It means that, in $\SDDIs$, rules $\zdr$ and $\pdr$ can be pushed down in a derivation, and rules $\zur$ and $\pur$ can be pushed up.

\begin{lemma}[Permuting $\lzero$ and $\lplus$]
	\label{lemma:sumPermutations}
	Any rule in $\SDDIs$ permutes over $\zdr$ and $\pdr$, and  under $\pur$ and $\zur$.
\end{lemma}
\begin{proof}
We define the rule permutations below, for $\ruler, \rulet, \rules \in \SDDIs$ with $\ruler\neq \pur$,~ $\rulet\neq \zur$ and \mbox{$\rules \not\in\set{\pdr, \zdr}$}.
	\begin{equation}\label{eq:sumZeroPerm}
		\small
		\begin{array}{c@{\qquad}c@{\qquad}c}
		\odn{ B \lplus B}{\pdr}{\odn{B}{\ruler}{A}{}}{} 
		\normred
		\odn
		{\odn{B}{\ruler}{A}{} 
			\lplus   
		\odn{B}{\ruler}{A}{}}{\pdr}{A}{}
&
	\vlderivation{\vlin{\pur}{}{B\lplus B}{\vlin{\pdr}{}{B}{\vlhy{B\lplus B}}}}
		\normred
	B\lplus B
&		
	\vlderivation{\vlin{\rulet}{}{A}{\vlin{\zdr}{}{B}{\vlhy{\lzero}}}}
		\normred
	\vlinf{\zdr}{}{A}{\lzero}
\\
	\vlderivation{\vlin{\zur}{}{\lzero}{\vlin{\zdr}{}{B}{\vlhy{\lzero}}}}
		\normred
	\lzero
	&
	{\odn{\odn{B}{\rules}{A}{}}	{\pur}	{ A \lplus A}	{}}	
		\normred	
	{\odn{B}{\pur}{  \odn{B}{\rules}{A}{} \lplus   \odn{B}{\rules}{A}{} }{}}\ 
	&
	\vlsmash{\vlderivation{\vlin{\zur}{}{\lzero}{\vlin{\rules}{}{A}{\vlhy{B}}}}}
	\normred
	\vlinf{\zur}{}{\lzero}{B}
\end{array}
	\end{equation}

Note that some of the rule permutations in \eqref{eq:sumZeroPerm} may 
implicitly use formula equivalence $\fequiv$ in order to be applied, for example in the following rule permutations concerning $\aidr$ and $\pur$, or $\aidr$ and  $\zur$:
\begin{equation}
	\small
	\label{eq:fequivPerm}
	\begin{array}{c}
	\vlderivation{\vlin{\aidr}{}{ \od{\odn{a}{\pur}{a\lplus a}{}} \lpar \cneg a}{\vlhy{\lone}}}
	\quad=\;
	\vlderivation{\vlin{\aidr}{}{ \od{\odn{(a\lpar \cneg a)}{\pur}{((a\lplus a)) \lpar \cneg a}{}}}{\vlhy{\lone}}}
	\quad\fequiv\;
	\vlderivation{
		\vlin
		{\aidr}
		{}
		{ 
			\od{
				\odn
				{a\lpar \cneg a}
				{\pur}
				{((a\lpar \cneg a))\lplus((a\lpar \cneg a))}
				{}
			}
		}
		{\vlhy{\lone}}
	}
	\quad\normred\quad
	\vlderivation{
		\vlin
		{\pur}
		{}
		{
			\od{
				\odn
				{\lone}
				{\aidr}
				{a\lpar \cneg a}
				{}
			}
			\lplus 
			\od{
				\odn
				{\lone}
				{\aidr}
				{a\lpar \cneg a}
				{}
			}
		}
		{\vlhy{\lone}}
	}
\\[2em]
\vlderivation{\vlin{\aidr}{}{ \od{\odn{a}{\zur}{\lzero}{}} \lpar \cneg a}{\vlhy{\lone}}}
\quad=\quad
\vlderivation{\vlin{\aidr}{}{ \od{\odn{a\lpar \cneg a}{\zur}{\lzero\lpar \cneg a}{}}}{\vlhy{\lone}}}	
\quad\fequiv\quad
\vlderivation{\vlin{\aidr}{}{ \od{\odn{a\lpar \cneg a}{\zur}{\lzero}{}}}{\vlhy{\lone}}}
\quad\normred\quad
\vlinf{\zur}{}{\lzero}{\lone}
\end{array}
\end{equation}
Similar permutations take place when $\aidr$ is replaced with $\occur$ or $\wncur$, and in their dual configurations.
\end{proof}

As a consequence of \Cref{lemma:sumPermutations}, pushing up the rules $\zur$, $\pur$ and down $\zdr$, $\pdr$ can be interpreted as \emph{slice management} operations: it duplicates and discards the ``slices'' (the subderivations without the rule $\zur$, $\pur$, $\zdr$, $\pdr$) and extends them as much as possible, propagating the non-deterministic choice $\lplus$ and the resource mismatch $\lzero$ all along the derivation.
It generalizes \Cref{rmk:additive} (for formulas) to derivations.


\begin{corollary}[Slice management]
	\label{lemma:sliceMan}
	Let $A$, $A' $, 
	$B$, $B' $ be $\MELL$~formulas, and $\ctxp{\,}$ be a $\MELL$ context. 
	Let $\normred$ be one of the steps defined in \eqref{eq:sumZeroPerm}, and $\normred^n$ be a sequence of $n \in \N$ of such steps.
	Then, 
	$$\scriptsize
	\vlderivation{\vlde{\dD_1}{\SDDI}{A}{\vlde{\dD_2}{\SDDI}{\ctxp{\lzero}}{\vlhy{B}}}}
	\normred^{m} 
	\vlderivation{\vlin{\zdr}{}{A}{\vlin{\zur}{}{\lzero}{\vlhy{B}}}}
	\qquad \qquad
	\vlderivation{\vlde{\dD_1}{\SDDI}{A}{\vlde{\dD_2}{\SDDI}{\ctxp{\lzero}}{\vlhy{\lones}}}}
	\normred^{m+1} 
	\vlinf{\zdr}{}{A}{\lzero}
	\qquad
	\od{
		\odd{\odi{\odi{\odd{\odh{B'}}{\dD_B}{B}{\SDDI}}{\pur}
				{\od{\odd{\odh{B}}{\dD_1}{A'}{\SDDI}} \lplus \od{\odd{\odh{B}}{\dD_2}{A'}{\SDDI}}}
				{}
			}{\pdr}{A'}{}
		}{\dD_A}{A}{\SDDI}}
	\normred^{k}
	\od{
		\odi{\odi{\odh{B'}}{\pur}{
				\od{
					\odd{\odh{B'}}{\dD_B \circ \dD_1 \circ \dD_A}{A}{\SDDI}
				}
				\lplus
				\od{
					\odd{\odh{B'}}{\dD_B \circ \dD_2 \circ \dD_A}{A}{\SDDI}
				}
			}{}}{\pdr}{A}{}
	}
	$$
	where $m=\sizeof{\dD_1}+\sizeof{\dD_2}$ and $k=\sizeof{\dD_A}+\sizeof{\dD_B}$ ($\sizeof{\dD}$ is the number of inference rules in the derivation $\dD$).
\end{corollary}

Actually, we can further structure $\SDDIs$ derivations so as to separate an initial up-segment and a final down-segment (\Cref{thm:DDInorm}). 
To prove this, we use the two following lemmas.

	\begin{figure}[t!]
	\begin{center}\scriptsize
		$
		\vlderivation{\vlin{\ocdur}{}{A}{\vlin{\wnddr}{}{\wn A}{\vlhy{A}}}}
		\normred
		\vlderivation{A}
		$
		\qquad
		$
		\vlderivation{\vlin{\wndur}{}{A}{\vlin{\ocddr}{}{\oc A}{\vlhy{A}}}}
		\normred
		\vlderivation{A}
		$
		\qquad
		$
		\vlderivation{\vlin{\ocwur}{}{\lbot}{\vlin{\wnwdr}{}{\wn A}{\vlhy{\lbot}}}}
		\normred
		\vlderivation{\lbot}
		$
		\qquad
		$
		\vlderivation{\vlin{\wnwur}{}{\lone}{\vlin{\ocwdr}{}{\oc A}{\vlhy{\lone}}}}
		\normred
		\vlderivation{\lone}
		$
		\qquad
		$
		\vlderivation{\vlin{\ocdur}{}{A}{\vlin{\wnwdr}{}{\wn A}{\vlhy{\lbot}}}}
		\normred
		\vlderivation{\vlin{\zdr}{}{A}{\vlin{\zur}{}{\lzero}{\vlhy{\lbot}}}}
		$
		\qquad
		$
		\vlderivation{\vlin{\wndur}{}{A}{\vlin{\ocwdr}{}{\oc A}{\vlhy{\lone}}}}
		\normred
		\vlderivation{\vlin{\zdr}{}{A}{\vlin{\zur}{}{\lzero}{\vlhy{\lone}}}}
		$

		$
		\vlderivation{\vlin{\ocdur}{}{A}{\vlin{\wncdr}{}{\wn A}{\vlhy{\wn A \lpar \wn A}}}}
		\normred
		\vlderivation{\vlin{\pdr}{}{A}{\vlin{\pur}{}
				{
					\od{\odn{\odn{\wn A}{\ocdur}{A}{}\lpar \odn{\wn A}{\wnwur}{\lbot}{}}{\fequiv}{A}{}}
					\lplus
					\od{\odn{\odn{\wn A}{\wnwur}{\lbot}{}\lpar \odn{\wn A}{\ocdur}{A}{}}{\fequiv}{A}{}}
				}
				{\vlhy{\wn A \lpar \wn A}}}}
		$
		\qquad
		$
		\vlderivation{\vlin{\wndur}{}{A}{\vlin{\occdr}{}{\oc A}{\vlhy{\oc A \ltens \oc A}}}}
		\normred
		\vlderivation{\vlin{\pdr}{}{A}{\vlin{\pur}{}
				{
					\od{\odn{\odn{\oc A}{\wndur}{A}{}\ltens \odn{\oc A}{\ocwur}{\lone}{}}{\fequiv}{A}{}}
					\lplus
					\od{\odn{\odn{\oc A}{\ocwur}{\lone}{}\ltens \odn{\oc A}{\wndur}{A}{}}{\fequiv}{A}{}}
				}
				{\vlhy{\oc A \ltens \oc A}}}}
		$
		
		$
		\vlderivation{\vlin{\ocwur}{}{\lbot}{\vlin{\wncdr}{}{\wn A}{\vlhy{\wn A \lpar \wn A}}}}
		\normred
		\od{\odh{
				\od{\odi{\odh{\wn A}}{\ocwur}{\lbot}{}} \lpar \od{\odi{\odh{\wn A}}{\ocwur}{\lbot}{}}
		}}
		$
		\qquad
		$
		\vlderivation{\vlin{\wnwur}{}{\lone}{\vlin{\occdr}{}{\oc A}{\vlhy{\oc A \ltens \oc A}}}}
		\normred
		\od{\odh{
				\od{\odi{\odh{\oc A}}{\wnwur}{\lone}{}} \ltens \od{\odi{\odh{\oc A}}{\wnwur}{\lone}{}}
		}}
		$
		
		$
		\vlderivation{\vlin{\occur}{}{\wn A \lpar \wn A}{\vlin{\wncdr}{}{\wn A}{\vlhy{\wn A \lpar \wn A}}}}
		\normred
		\od{
			\odi{\odh{
					\odn{\wn A}{\occur}{\bclr{\wn A} \lpar \bclr{\wn A}}{} \lpar \odn{\wn A}{\occur}{\rclr{\wn A }\lpar \rclr{\wn A}}{}
			}}
			{\fequiv}{\odn{\bclr{\wn A }\lpar \rclr{\wn A}}{\wncdr}{\wn A}{} \lpar \odn{ \rclr{\wn A}    \lpar \bclr{\wn A}}{\wncdr}{\wn A}{}}{}
		}
		$
		\qquad
		$
		\vlderivation{\vlin{\wncur}{}{\oc A \ltens \oc A}{\vlin{\occdr}{}{\oc A}{\vlhy{\oc A \ltens \oc A}}}}
		\normred
		\od{
			\odi{\odh{
					\odn{\oc A}{\wncur}{\bclr{\oc A} \ltens \bclr{\oc A}}{} \ltens \odn{\oc A}{\wncur}{\rclr{\oc A }\ltens \rclr{\oc A}}{}
			}}
			{\fequiv}{\odn{\bclr{\oc A }\ltens \rclr{\oc A}}{\occdr}{\oc A}{} \ltens \odn{ \rclr{\oc A}    \ltens \bclr{\oc A}}{\occdr}{\oc A}{}}{}
		}
		$
		\qquad
		
	\end{center}
	\caption{Non-trivial rule permutations for $\oc$ and $\wn$ by $\set{\zur, \zdr, \pur, \pdr}$ in $\SDDI$.
}
	\label{fig:modRulesPerm}
\end{figure}

\begin{lemma}[Permutations of rules for $\oc$ and $\wn$]
	\label{lemma:modupVSmoddown}
		In $\SDDI$, the following rule permutations hold:
		\begin{enumerate}
			\item\label{item:INperm} \emph{Interaction-net permutations:} The rules in   $\set{\ocdur, \wndur,\ocwur,\wnwur,\occur,\wncur}$ permute over the rules in $\set{\ocddr, \wnddr, \ocwdr,\wnwdr,\occdr,\wncdr}$ by the rules in $\set{\pdr,\pur,\zdr,\zur}$;

			\item\label{item:INpermdown} 	The rules in $\set{\wnddr,\ocddr, \wnwdr, \ocwdr, \wncdr,\occdr}$ permutes under any rule in $\set{\aiur,\aidr,\swir}$;

			\item\label{item:INpermup} 	The rules in $\set{\wndur,\ocdur, \wnwur, \ocwur, \wncur,\occur}$ permutes over any rule in $\set{\aiur,\aidr,\swir}$.
			
		\end{enumerate}
\end{lemma}
\begin{proof}\hfill
	\begin{enumerate}
		\item By the (non-trivial) rule permutations  in  \Cref{fig:modRulesPerm} or by their duals obtained by up/down symmetry.
		\item 
		First, note that all rule permutations involving $\aidr$ in $\set{\wnddr,\ocddr, \wnwdr, \ocwdr, \allowbreak \wncdr,\occdr}$ are trivial.
		Moreover, any $\rho\in \set{\wnddr,\ocddr, \wnwdr, \ocwdr, \allowbreak\wncdr,\occdr}$ permutes under $\swir$ as follows.
		\begin{equation}\label{eq:switchperm}
		\od{
			\odi{\odh{\od{\odi{\odh{A'}}{\ruler}{A}{}} \ltens ((B\lpar C))}}{\swir}{((A\ltens B))\lpar C}{}
		}
		\normred
		\od{
			\odi{\odh{A' \ltens ((B\lpar C))}}{\swir}{((\od{\odi{\odh{A'}}{\ruler}{A}{}}\ltens B))\lpar C}{}
		}
		\end{equation}
		We conclude by the following rule permutations for $\wnddr$, $\wncdr$ and $\wnwdr$ (permutations for $\ocddr$, $\occdr$ and $\ocwdr$ are defined similarly).
	\end{enumerate}
		$$
		\footnotesize
		\od{\odi{\odh{\lone}}{\wnddr\!}{\wn \; \od{\odi{\odh{\lone}}{\aidr\!}{\la\lpar \cneg \la}{}}}{}}
		\normred
		\od{\odi{\odi{\odh{\lone}}{\aidr\!}{\la\lpar \cneg \la}{}}{\wnddr\!}{\wn((\la\lpar \cneg \la))}}
		\quad
		\od{
			\odi
			{\odh{\wn \lone \lpar \wn \lone}}
			{\wncdr\!}
			{\wn \; \od{\odi{\odh{\lone}}{\aidr\!}{\la\lpar \cneg \la}{}}}
			{}
		}
		\normred
		\od{
			\odi
			{
				\odh{
					\wn\; 
					\od{\odi{\odh{\lone}}{\aidr\!}{\la \lpar \cneg \la}{}}
					\lpar \wn \; 
					\od{\odi{\odh{\lone}}{\aidr\!}{\la\lpar \cneg \la}{}}
				}
			}
			{\wncdr\!}
			{\wn ((\la\lpar \cneg \la))}
			{}
		}
		\quad
		\od{
			\odi
			{\odh{\lbot}}
			{\wnwdr\!}
			{\wn \; \od{\odi{\odh{\lone}}{\aidr\!}{\la\lpar \cneg \la}{}}}
			{}
		}
		\normred
		\od{
			\odi
			{\odh{\lbot}}
			{\wnwdr\!}
			{\wn ((\la\lpar \cneg \la))}
		}
		$$
		
	\begin{enumerate}
		\item[3.] It can be obtained dually from \Cref{item:INpermdown}, by the up/down symmetry of rules.
		\qedhere
	\end{enumerate}
\end{proof}

Permutations in \Cref{fig:modRulesPerm} and their duals correspond to the cut-elimination steps for modalities $\wn$ and $\oc$ in the \emph{interaction-nets} presentation of $\SDill$, see \cite[Fig. 4]{Pagani09} and \cite{EhrhardRegnier06,Tranquilli09,Gimenez11}.
They take place when a down-rule for $\oc$ meets an up-rule for $\wn$, or vice-versa, and deal with the \emph{resource management} (see \Cref{sec:DiLL}).
Akin to interaction-nets and unlike the sequent calculus, these permutations on $\SDDI$ are perfectly symmetric.
Note the key role of the rules $\pur$, $\pdr$, $\zur$, $\zdr$ in some permutations. In particular,
\begin{itemize}
	\item formulas $\lzero$ appear when there is a mismatch between ``supply and demand'' ($\wnwdr/\ocdur$ and $\ocwdr/\wndur$),
	\item formulas with $\lplus$ appear when there is a non-deterministic choice on which request will be fed ($\wncdr/\ocdur$ and $\occdr/\wndur$).
\end{itemize}


\begin{lemma}[Linear permutations]
	\label{lemma:linearPerm}
	Let $A$ and $B$ be $\MELL$~formulas.
	\begin{enumerate}
		\item\label{MLLnorm} If  $B\provevia{\set{\aiur,  \aidr,  \swir}}A$ then $ B \provevia{\set{\aidr}} B'\provevia{\set{\swir}} A'\provevia{\set{\aiur}}A$
	for some $\MELL$~formulas $B'$ and $A'$.
	
		\item\label{cutEliMLL} If $\lone\provevia{\set{\aiur,  \aidr,  \swir}}A$ then $\lone \provevia{\set{\aidr}} A'\provevia{\set{\swir}} A$ for some $\MELL$ formula $A'$.
	\end{enumerate}
\end{lemma}
\begin{proof}[Proof (sketch)]
	This is a standard result in deep inference systems. 
	Nowadays, it is usually proved via splitting \cite{TubellaGuglielmi18,tubella:phd,tub:str:esslli19}, as it is a consequence of cut-elimination, but the hypotheses to apply the splitting technique do not hold in $\SDDI$. 
	In  \cite{guglielmi1999calculus,str:phd,str:MELLinCoS,gug:SIS}, which were written before the splitting technique was found, \Cref{MLLnorm,cutEliMLL} are proved using some sort of rule permutations.
	To prove \Cref{MLLnorm} it is enough to use the non-trivial rule permutations of $\aidr$ over $\swir$ shown in \eqref{eq:MLLcut-elim} below, and the dual rule permutations of $\aiur$ under $\swir$ obtained from \eqref{eq:MLLcut-elim} by up/down symmetry.
	\begin{equation}\label{eq:MLLcut-elim}
		\small
		\od{
			\odi{\odh{\lone \ltens ((B\lpar C))}}{\swir}{((\od{\odi{\odh{\lone}}{\aidr}{a\lpar \cneg a}{}}\ltens B))\lpar C}{}
		}
		\normred
		\od{
			\odi{\odh{\od{\odi{\odh{\lone}}{\aidr}{a\lpar \cneg a}{}} \ltens ((B\lpar C))}}{\swir}{(((a\lpar \cneg a)\ltens B))\lpar C}{}
		}
		\quad
		\quad
		\od{
			\odi{\odh{A \ltens ((B\lpar \lone))}}{\swir}{((A \ltens B))\lpar \od{\odi{\odh{\lone}}{\aidr}{a\lpar \cneg a}{}}}{}
		}
		\normred
		\od{
			\odi{\odh{A \ltens ((B\lpar \od{\odi{\odh{\lone}}{\aidr}{a\lpar \cneg a}{}}))}}{\swir}{((A\ltens B))\lpar \lone}{}
		}
	\end{equation}

	To have an intuition for the proof of \Cref{cutEliMLL}, it is enough to remark that if $\lone \provevia{\set{\aidr}} A'$, then there is a derivation of $A'$ with \emph{shallow} $\aidr$, that is, with $\aidr$ applied only in $\ltens$-context as the one below on the left:
	\vspace{-5pt}$$
	\small
	\begin{array}{c@{\qquad\qquad}|@{\qquad\qquad}c}
	\vlder{}{\set{\swir}}{A'}{
	\od{\odn{{\lone}}{\aidr}{a_1\lpar \cneg a_1}{}}
	\ltens\cdots\ltens
	\od{\odn{{\lone}}{\aidr}{a_n\lpar \cneg a_n}{}}
 	}
	&
	{\od{
		\odi{
			\odi{
				\odi{\odh{\lone}}{\fequiv}{
					\odn{\lone}{\aidr}{a\lpar \cneg a}{} \ltens {\odn{\lone}{\aidr}{a\lpar \cneg a}{}} 
				}{}
				}{2\times \swir}{a\lpar \odn{\cneg a\lpar a}{\aiur}{\lbot}{} \lpar \cneg a}{}
		}{\fequiv}{a\lpar \cneg a}{}
	}}
	\normred
	\vlinf{\aidr}{}{a\lpar \cneg a}{\lone}
	\end{array}\vspace{-15pt}
	$$
	Hence by \Cref{MLLnorm}, if $\lone\provevia{\set{\aiur,  \aidr,  \swir}}A$ then $A$ is provable by starting from shallow $\aidr$ rules; 
	and if there is a rule $\aiur$, 
	then there is at least one $\aiur$ that can be permuted up in the derivation, 
	until we can obtain a configuration as the one above on the right,
	which can be replaced by a rule $\aidr$.
\end{proof}

Rule permutations involved in the proof of \Cref{lemma:linearPerm}.\ref{cutEliMLL}
essentially correspond to \emph{$\MLL$ cut-elimination steps} in the $\SDill$ sequent calculus, indeed modalities $\oc$ and $\wn$ do not play any active role there.

\begin{definition}[Normalization step]
	\label{def:normalization}
	Any rewrite relation on $\SDDI$ derivations that is a non-trivial rule permutation used in the proofs of
	\Cref{lemma:sumPermutations,lemma:modupVSmoddown,lemma:linearPerm} is a \emph{normalization step} and denoted by $\normred$.
\end{definition}

Normalization steps 
rearrange rules in a $\DDIdown$ or $\SDDIs$ derivation in a fixed order, leaving unchanged its premise and conclusion.
So, derivability in $\DDIdown$ and $\SDDIs$ can be decomposed in~several~segments.
More precisely, every derivation in $\SDDIs$ can be rearranged in a symmetrical way so that:
\begin{enumerate}
	\item on the top there is an \emph{up-segment} where:
	\begin{enumerate}
		\item the first part consists of rules $\zur$ and $\pur$, which decompose the derivation in vertical slices;
		\item the second part consists of up-rules for $\oc$ and $\wn$, which deal with non-linear resources; 
	\end{enumerate}
	\item in the middle there is a \emph{linear segment}, roughly corresponding to $\MLL$ and to linear resources;
	\item on the bottom there is a \emph{down-segment} where:
		\begin{enumerate}
			\item the first part consists of down-rules for $\oc$ and $\wn$, which deal with non-linear resources; 
		\item the second part consists of rules $\zdr$ and $\pdr$, which merge the vertical slices of the derivation.
	\end{enumerate}
\end{enumerate}
The decomposition in $\DDIdown$ follows the same pattern but takes only down-rules, so there is \mbox{no up-segment}.

\begin{theorem}[Decomposition]
	\label{cor:decompose}
	\label{thm:DDInorm}	\label{lemma:SDDInorm}
	Let $A$ and $B$ be formulas. 
	
	\begin{enumerate}
		\item\label{p:decompose-DDI} \emph{$\DDIdown$-decomposition}: If $\deriv{\dD}{\lones} {\DDIdown}{A}$, then (for some additive normal formulas $A',A'',A'''$) there is a derivation $\dD'$ in $\SDDIdown$ such that $\dD \normred^* \dD'$ and
		$$
		\deriv{\dD'} {\lones}{\set{\aidr }}{{A'''}\provevia{\set{\swir}} A'' \provevia{\set{\wnddr, \ocddr,\wnwdr, \ocwdr,\wncdr, \occdr}}A' \provevia{\set{\zdr, \pdr}} A}.
		$$

		\item\label{p:decompose-SDDI} \emph{$\SDDI$-decomposition}:
		If $\deriv{\dD}{B}{\SDDIs}{A}$, then there is a derivation $\dD'$ (called \emph{standard}) in $\SDDIs$ from $B$ to $A$ 
		such that $\dD \normred^* \dD'$ and (for some additive normal  formulas $B',B'',B''',A''',A'',A'$):	
		$$
		\deriv{\dD'}
		{B}
		{\set{\zur, \pur}}
		{B'
		\provevia{\set{ \wndur, \ocdur,\ocwur, \wnwur, \wncur, \occur}}B''
		\provevia{\set{\aidr }}B'''
		\provevia{\set{\swir}}A'''
		\provevia{\set{\aiur}}A''
		\provevia{\set{\wnddr, \ocddr,\wnwdr,\ocwdr, \wncdr, \occdr}}A' 
		\provevia{\set{\zdr, \pdr}} A}.
		$$
	\end{enumerate}
\end{theorem}
\begin{proof}
		The decomposition of $\DDIdown$ derivations follows from \Cref{lemma:sumPermutations}, \Cref{lemma:modupVSmoddown}.\ref{item:INpermdown} and \Cref{lemma:linearPerm}.\ref{cutEliMLL}.
	
	To prove decomposition of $\SDDI$ derivations, we alternate applications of \Cref{lemma:sumPermutations,lemma:modupVSmoddown} until we obtain a derivation of the shape below.
	Then we conclude by applying \Cref{lemma:linearPerm}.\ref{MLLnorm}.
	\begin{equation*}
	B
	\provevia{\set{\zur, \pur}}B'
	\provevia{\set{ \wndur, \ocdur,\ocwur, \wnwur, \wncur, \occur}}B''
	\provevia{\set{\aidr,\swir,\aiur}}A''
	\provevia{\set{\wnddr, \ocddr,\wnwdr,\ocwdr, \wncdr, \occdr}}A' 
	\provevia{\set{\zdr, \pdr}} A.
	\qedhere
	\end{equation*}
\end{proof}



As a consequence, the up-fragment $\DDIup$ of $\SDDIs$ is \emph{superfluous} (\Cref{thm:SDDIcutElimInternal}): all that can be proved in $\SDDIs$, is already derivable in the down-fragment $\DDIdown$ of $\SDDIs$ by a standard derivation.
The existence of a standard derivation in $\DDIdown$ is obvious because the rule $\zdr$ makes every $\MELL$ formula derivable. 
The interesting part is that a standard derivation in $\DDIdown$ can be reached via normalization steps, hence in a computational way that is \emph{internal} to $\SDDI$.
Indeed, normalization of $\SDDIs$ derivations follows from $\SDDIs$ decomposition (\Cref{cor:decompose}.\ref{p:decompose-SDDI}), so it relies on the normalization steps defined on $\SDDIs$ derivations (\Cref{def:normalization}).
This normalization result is the deep inference version of \emph{cut-elimination}, since in $\DDIdown$ there is no analogue of the rule $\cutr$ ($\DDIdown$ is the ``cut-free'' fragment~of~$\SDDI$).

\begin{corollary}[Normalization]
	\label{thm:SDDIcutElimInternal}
	Let $A$ be a formula and $n \in \N$. 
	If $\deriv{\dD}{\lones}{\SDDIs}{A}$ then, for some $n' \in \N$, there exists a standard $\deriv{\dD'}{\lones\!}{\{\zur,\pur\}}{\lones' \provevia {\DDIdown}  A}$ such that $\dD \normred^* \dD'$. 
	In particular, $\lones' \provevia {\DDIdown}  A$ for some $n' \in \N$.
\end{corollary}
\begin{proof}
	By \Cref{lemma:SDDInorm}.\ref{p:decompose-SDDI}, if $\lones \provevia{\SDDIs} A$ then there is a standard derivation $$\lones \provevia {\set{\zur,\pur}} B \provevia{\set{ \wndur, \ocdur,\ocwur, \wnwur, \wncur, \occur}}B'
	\provevia{\set{\aidr,\swir,\aiur}}A''
	\provevia{\set{\wnddr, \ocddr,\wnwdr,\ocwdr, \wncdr, \occdr}}A' 
	\provevia{\set{\zdr, \pdr}} A.$$
	Moreover, $\lones \provevia {\set{\zur,\pur}} B$ implies that $B = \lones'$ for some $n' \in \N$.
	As no rule in $\set{\wnwur, \ocwur, \wncur, \occur, \wndur, \ocdur}$ can be applied to a formula of the form $\lones'$, we have $\lones' = B'$ and we conclude by \Cref{lemma:linearPerm}.\ref{cutEliMLL}.
\end{proof}

%



\section{Relation between Cut-elimination in \texorpdfstring{$\SDill$}{DiLL0} and Normalization in \texorpdfstring{$\SDDI$}{SDDI}}
\label{sect:commute}

%


In this section we investigate the correspondence between the cut-elimination procedure
in $\SDill$ sequent calculus (\Cref{sec:DiLL}) and the normalization procedure in
$\SDDI$ (\Cref{sect:normalization}).

We provided a translation $\toform{\cdot}$ of $\eta$-expanded $\SDill$ derivations to  derivations in $\DDIdown\cup\set{\iur}$ (\Cref{fig:seqToDeep}) and so in $\SDDI$ (via \Cref{lemma:idDer}). 
The translation preserves ``cut-freeness'' (\Cref{thm:SDilltoSDDIs}), and exhibits a one-to-one correspondence between weakening, contraction and dereliction rules of the two systems.
However, the translation $\toform{\cdot}$ 
does not commute with cut-elimination/\allowbreak normalization: an $\eta$-expanded derivation $\pi$ in $\SDill$ 
might reduce to a cut-free derivation $\widehat{\pi}$ via cut-elimination, but its translation $\toform{\pi}$ in $\SDDIs$ (including the transformation of the rules $\iur$ into $\aiur$, as described in \Cref{lemma:idDer}) normalizes to a $\DDIdown$ derivation $\widehat{\toform{\pi}}$ other than $\toform{\widehat{\pi}}$.
That is, diagram \eqref{eq:diagram} does~\emph{not~commute}:
\begin{equation}
\small
\label{eq:diagram}
\begin{tikzpicture}[x=30pt, y=15pt, baseline=(current  bounding  box.center)]
\node (v1) at (-1,0) {$\pi$};
\node (v3) at (1,0) {$ \toform\pi$};
\node (v4) at (1,-2) {$ \toform{\widehat\pi} \neq \widehat {\toform\pi}$};
\node (v2) at (-1,-2) {${\widehat\pi}$};
\draw[->] (v1) -- (v3) node [midway, below=.5pt, fill=white] {\tiny $\toform~$};
\draw[->]  (v2) -- (v4) node [midway, above=.5pt, fill=white] {\tiny $\toform~$};
\draw[->]  (v3) -- (v4) node [midway, right=4pt, fill=white] {\tiny $\mathsf{norm}$} node [very near end, right=.2pt] {\tiny $*$};
\draw[->]  (v1) -- (v2) node [midway, left=4pt, fill=white] {\tiny $\cutr$} node [very near end, right=.2pt] {\tiny $*$} ;
\end{tikzpicture}
\end{equation}

Technically, the lack of commutation is because
the rule $\cutr$ is translated as an instance of $\iur$, which is not a rule of $\SDDIs$ (it is not an atomic cut) and hence has to be rewritten according to \Cref{lemma:idDer}.  
But this rewriting in $\SDDIs$ might not match the resource distribution of the corresponding $\cutr$.
Consider the derivation $\pi$ below in $\SDill$ (with $\pi'$ cut-free and $\eta$-expanded) and its 
translation $\dD_\pi$ in $\SDDI$:

\vspace*{-1\baselineskip}
$$\footnotesize
\ebproofset{right label template=\tiny$\inserttext$, label separation=0.2em, separation = 1.2em}
\pi \ =  \
\begin{prooftree}
\hypo{}
\ellipsis{$\pi'$}{\vdash \Gamma}
\infer1[\wnwrule]{\vdash \Gamma, \wn a}
\infer0[\ocwrule]{\vdash \oc \widecneg{a}}
\infer2[\cutr]{\vdash \Gamma}
\end{prooftree}
\quad
\underset{\toform{\cdot}}{\overset{\Cref{fig:seqToDeep}}{\rightarrow}}
\quad
\od{
	\odi
	{
		\odi
		{\odh{\lones}}
		{\fequiv}
		{
			\od{
				\odi
				{\odd
					{\odh{\lones}}
					{\toform{\pi'}}
					{\toform\Gamma}
					{}
				}
				{\fequiv}
				{\toform{\Gamma} \lpar \odn{\lbot}{\wnwdr}{\wn a}{}}
				{}
			}
			\ltens 
			\od{
				\odi{\odh{\lone}}{\ocwdr}{\oc \cneg a}{}}
		}
		{}
	}
	{\swir}
	{\toform{\Gamma} \lpar 
		\od{
			\odi{
				\odh{
					\wn a\ltens\oc\cneg a
				}
			}
			{\iur}
			{\lbot}
			{}
		}
	}
	{}
}
\quad
\overset{\Cref{lemma:idDer}}{\rightarrow}
\quad
\od{\odi{\odi{\odh{\lones}}{\fequiv}
		{
			\od{\odi{\odd{\odh{\lones}}{\toform{\pi'}}{\toform\Gamma}{}}{\fequiv}{\toform \Gamma \lpar \odn{\lbot}{\wnwdr}{\wn a}{}}{}}
			\ltens 
			\od{\odi{\odh{\lone}}{\ocwdr}{\oc \cneg a}{}}}{}
	}{\swir}{\toform\Gamma \lpar 
		\od{\odi{\odh{\od{\odi{\odh{\wn a}}{\ocdur}{a}{}}\ltens\od{\odi{\odh{\oc\cneg a}}{\wndur}{\cneg a}{}}}}{\aiur}{\lbot}{}}
	}{}}
\ = \dD_\pi
$$

According to cut-elimination for $\SDill$, $\pi \cutred \pi'$ (one step). 
But $\dD_\pi$ in $\SDDI$ normalizes as follows:

\vspace*{-1\baselineskip}
$$
\small
\dD_\pi
\overset{\Cref{lemma:modupVSmoddown}.\ref{item:INpermup}}{\normred^*}
\od{
	\odi{\odi{
			\odi{\odh{\lones}}{\fequiv}{\od{
					\od{\odi{\odd{\odh{\lones}}{\toform{\pi'}}{\toform \Gamma}{}}{\fequiv}{\toform \Gamma \lpar 
							\od{
								\odi{\odi{\odh{\lbot}}{\wnwdr}{\wn a}{}}{\ocdur}{a}{}
							}
						}{}}
					\ltens
					\odi{\odi{\odh{\lone}}{\ocwdr}{\oc \cneg a}{}}{\wndur}{\cneg a}{}
			}}{}
		}{ \swir}{
			\toform \Gamma \lpar \od{\odi{\odh{a\ltens \cneg a}}{\aiur}{\lbot}{}}
		}{}}{\fequiv}{\toform \Gamma}{}
}
\quad \overset{\substack{\wnwdr/\ocdur\\(\Cref{fig:modRulesPerm})\\ \ }}{\normred}\quad
\od{
	\odi{\odi{
			\odi{\odh{\lones}}{\fequiv}{\od{
					\od{\odi{\odd{\odh{\lones}}{\toform{\pi'}}{\toform \Gamma}{}}{\fequiv}{\toform \Gamma \lpar 
							\od{
								\odi{\odi{\odh{\lbot}}{\zur}{\lzero}{}}{\zdr}{a}{}
							}
						}{}}
					\ltens
					\odi{\odi{\odh{\lone}}{\ocwdr}{\oc \cneg a}{}}{\wndur}{\cneg a}{}
			}}{}
		}{ \swir}{
			\toform \Gamma \lpar \od{\odi{\odh{a\ltens \cneg a}}{\aiur}{\lbot}{}}
		}{}}{\fequiv}{\toform \Gamma}{}
}
\overset{\Cref{lemma:sliceMan}}{\normred^*}
\od{\odi{\odh{\lzero}}{\zdr}{\toform \Gamma }{}}
$$

We observe that the  transformation of the general $\iur$-rule into $\aiur$ (\Cref{lemma:idDer}) converts any potential interaction of weakening, contraction and dereliction up- and down-rules to an interaction of a (weakening, contraction or dereliction) down-rule with a dereliction up-rule: it arbitrarily chooses to asks for a resource, or to make it available, exactly once.
In our example, the translation creates the ``mismatches'' $\wnwdr/\ocdur$ and $\ocwdr/\wndur$ even if in the original $\SDill$ derivation we had a ``matched'' interaction of a $\wnwrule$ with a $\ocwrule$. 
Due to these mismatches, the normal form of the derivation $\dD_\pi$ is 
{\small${\vlinf{\zdr}{}{\toform{\Gamma}}{\lzero}}$},
which is not the translation of $\pi'$ (the normal form of $\pi$ with respect to cut-elimination in $\SDill$) 
if $\pi' \neq
\begin{prooftree}[label separation = .2em]
\infer0[{\scriptsize$\zerorule$}]{\vdash \Gamma}
\end{prooftree}$.
%
More generally, this problem is related to the fact that $\SDill$ misses the promotion rule $\prule$ (\Cref{fig:scrules}), which would make a resource available at will \cite{Girard87} (see \cite{Pagani09,Tranquilli09,Gimenez11} for its cut-elimination in $\DiLL = \SDill \cup \set{\prule}$), while in the realm of $\SDill$ resources can be used only a finite number of times during cut-elimination.

In order to provide an \emph{internal} solution (\ie, in $\SDill$ without adding the 
rule $\prule$ to the system) to the commutation problem, we need to define a more sophisticated translation of $\cutr$-rules into $\SDDIs$. 

\paragraph{A commutative translation.}


We define a new translation $\toformnew{\cdot}$ from $\SDill$ to $\SDDI$ so that diagram \eqref{eq:diagram} commutes, when $\toform{\cdot}$ is replaced by $\toformnew{\cdot}$.
The idea is that the translation $\toformnew{\cdot}$ ``bends'' a derivation $\pi$ of 
$\provevia{\SDill} \Gamma, A$ to a derivation  {\small$\vlderivation{\vlde{\toformnew{\pi}}{\SDDIs}{\toform{\Gamma}}{\vlhy{\widecneg{A}}}}$} so as to avoid using the rule $\iur$.
In this way, roughly, the translation $\toformnew{\cdot}$ converts the rule $\cutr$ below (where $\pi_1$ and $\pi_2$ are $\eta$-expanded and cut-free $\SDill$ derivations) as follows:
\vspace{5pt}
\begin{center}
{\small$
\vlsmash
{\vlderivation{\vliin{}{\cutr}{\vdash \Gamma, \Delta}{\vlpr{}{\pi_1}{\vdash \Gamma, \lA}}{\vlpr{}{\pi_2}{\vdash \Delta , \widecneg \lA}}} 
\ \translatesnewto  \ \
\od{ \odd{\odh{\lones_1}}{\toform{\pi_1}}{\toform \Gamma \lpar \od{\odd{\odh{A}}{\toformnew{\pi_2}}{\toform \Delta}{}}{}}}}
$}
\qquad where $\toform{\pi_1}$ is the translation of $\pi_1$ defined	 in \Cref{fig:seqToDeep}.
\end{center}

\vspace{20pt}

To define properly the translation $\toformnew{\pi}$ of an $\eta$-expanded $\SDill$ derivation, we first need to declare which formulas in $\pi$ have to be ``bent'', selecting one of the two cut formulas for each $\cutr$ in $\pi$. 
Formally, given an $\eta$-expanded $\SDill$ derivation $\pi$, a \emph{translation} $\toformnew{\pi}$ of $\pi$ into $\SDDI$ is defined in two steps.
\begin{enumerate}
	\item For each occurrence of the rule $\cutr$ in $\pi$, we mark exactly one of its two cut formulas, say $A$, with $\isneg{A}$. 
	We propagate this mark bottom-up in $\pi$ to the subformula occurrences of $A$ in $\pi$: if the principal formula of a rule is marked, so are the active formulas in the premises (the other formulas in the sequent preserve their mark, if any).
	For instance, if the conclusion of a rule $\lpar$ is $\vdash \isneg{A}, \isneg{(B \lpar C)}, D$ with $B \lpar C$ as principal formula, then its premise is $\vdash \isneg{A}, \isneg{B}, \isneg{C}, D$.
	\item We translate $\pi$ decorated with marks $\isneg{(\cdot)}$ into a derivation $\toformnew{\pi}$ in $\SDDI$ according to the definition in \Cref{fig:seqToDeepUp} (given by induction on $\pi$).
\end{enumerate}

\begin{figure}[p]
	\def\myskip{\hskip2em}
	\def\myskiip{\hskip6em}
	\begin{center}
		\scriptsize
		$
		\vlinf{}{\axrule}{\vdash \la, \cneg \la}{}
		\ \translatesnewto \ \
		\vlinf{\aidr}{}{a \lpar \cneg a}{\lone}
		$
		\myskiip 
		$
		\vlinf{}{\axrule}{\vdash \isneg\la, \isneg{\cneg \la}}{}
		\ \translatesnewto \ \
		\vlinf{\aiur}{}{\lbot}{a \ltens \cneg a}
		$
		\myskiip 
		$
		\vlinf{}{\axrule}{\vdash \la, \isneg{\cneg \la}}{}
		\ \translatesnewto \ \
		a
		$
		\myskiip 
		$
		\vlinf{}{\axrule}{\vdash \isneg\la, {\cneg \la}}{}
		\ \translatesnewto \ \
		\cneg a
		$
		
		\smallskip
		$
		\vlderivation{
			\vlin{}{\lpar}{\vdash \Gamma_1,\isneg \Gamma_2, \isneg{(\lA \lpar \lB)}}{\vlpr{}{\pi}{\vdash \Gamma_1,\isneg \Gamma_2, \isneg\lA, \isneg\lB}}
		}
		\translatesnewto \
		\vlderivation{\od{\odd{\odh{\widecneg{\toform{\Gamma_2}} \ltens  \widecneg\lA \ltens \widecneg\lB}}{\toformnew{\pi}}{\toform{\Gamma_1}}{}} }
		$
		\myskiip 
		$
		\vlderivation{
			\vlin{}{\lpar}{\vdash \Gamma_1,\isneg \Gamma_2, \lA \lpar \lB}{\vlpr{}{\pi}{\vdash \Gamma_1,\isneg \Gamma_2, \lA, \lB}}
		}
		\translatesnewto \
		\vlderivation{\od{\odd{\odh{\widecneg{\toform{\Gamma_2}}}}{\toformnew{\pi}}{\toform{\Gamma_1} \lpar \lA \lpar \lB}{}} }
		$
		
		\smallskip
		$
		\vlderivation{\vliin{}{\ltens}{\vdash \Gamma_1, \isneg \Gamma_2,  \lA \ltens  \lB , \Delta_1, \isneg \Delta_2}{\vlpr{}{\pi_1}{\vdash \Gamma_1, \isneg \Gamma_2,  \lA}}{\vlpr{}{\pi_2}{\vdash \Delta_1\isneg \Delta_2 , \lB}}} 
		\translatesnewto \
		\vlderivation{
			\vlin
			{2\times \swir}
			{}
			{\toform{\Gamma_1} \lpar (\lA \ltens \lB) \lpar \toform{\Delta_1}} 
			{\vlhy{
					\od{\od{\odd{\odh{\widecneg{\toform{\Gamma_2}}}}{\toformnew{\pi_1}}{\toform{\Gamma_1} \lpar \lA}{}}} 
					\ltens 
					\od{\od{\odd{\odh{\widecneg{\toform{\Delta_2}}}}{\toformnew{\pi_2}}{\toform{\Delta_1} \lpar \lB}{}}} 
			}}
		}
		$
		\myskiip 
		$
		\vlderivation{\vliin{}{\ltens}{\vdash \Gamma_1, \isneg \Gamma_2, \isneg{( \lA \ltens \lB)} , \Delta_1, \isneg \Delta_2}{\vlpr{}{\pi_1}{\vdash \Gamma_1, \isneg \Gamma_2, \isneg \lA}}{\vlpr{}{\pi_2}{\vdash \Delta_1\isneg \Delta_2 , \isneg\lB}}} 
		\translatesnewto \
		\vlderivation{
			\vlin{2\times \swir}{}{
				\od{\od{\odd{\odh{\widecneg{\toform{\Gamma_2}}\ltens\cneg \lA}}{\toformnew{\pi_1}}{\toform{\Gamma_1}}{}}}
				\lpar
				\od{\od{\odd{\odh{\widecneg{\toform{\Delta_2}}\ltens\cneg \lB}}{\toformnew{\pi_2}}{\toform{\Delta_1}}{}}}
			}
			{\vlhy{\widecneg{\toform{\Gamma_2}}\ltens (\cneg A \lpar \cneg B)\ltens \widecneg{\toform{\Delta_2}}}}
		}
		$
		
		\smallskip
		$
		\vlderivation{\vlin{}{\botrule}{\vdash \Gamma_1,\isneg \Gamma_2, \lbot}{\vlpr{}{\pi}{\vdash \Gamma_1,\isneg \Gamma_2}}}
		\translatesnewto
		\vlderivation{
			\vlin
			{\fequiv}
			{}
			{\toform{\Gamma_1} \lpar \lbot}
			{\odd{\odh{\widecneg{\toform{\Gamma_2}}}}{\toformnew{\pi}}{\toform{\Gamma_1}}{} }
		}
		$
		\myskiip 
		$
		\vlderivation{\vlin{}{\botrule}{\vdash \Gamma_1,\isneg \Gamma_2, \isneg{\bot}}{\vlpr{}{\pi}{\vdash \Gamma_1,\isneg \Gamma_2}}}
		\translatesnewto \
		\od{
			\vlde
			{\toformnew{\pi}}
			{}
			{\toform{\Gamma_1}}
			{\vlin
				{\fequiv}{}
				{\widecneg{\toform{\Gamma_2}}}
				{\vlhy{\widecneg{\toform{\Gamma_2}} \ltens {\lone}}}
		}}
		$
		\myskiip 
		$
		\vlinf{}{\lone}{\vdash {\lone}}{} 
		\ \translatesnewto \ \
		{\lone}
		$
		\myskiip 
		$
		\vlinf{}{\lone}{\vdash \isneg{\lone}}{} 
		\ \translatesnewto \ \
		{\lbot}
		$
		
		\smallskip
		$
		\vlderivation{\vliin{}{\cutr}{\vdash  \Gamma_1, \isneg \Gamma_2,  \Delta_1, \isneg \Delta_2}{\vlpr{}{\pi_1}{\vdash  \Gamma_1, \isneg \Gamma_2, \lA}}{\vlpr{}{\pi_2}{\vdash  \Delta_1, \isneg \Delta_2 , \isneg{\widecneg \lA}}}} 
		\translatesnewto  \ 
		\od{
			\vlin
			{\swir}
			{}
			{\toform{\Gamma_1} \lpar 
				\od{
					\odd
					{\odh{\lA \ltens \widecneg{\toform{\Delta_2}}}}
					{\toformnew{\pi_2}}
					{\toform{\Delta_1}}
					{}
				}
			}
			{
				\vlhy{
					\od{
						\odd
						{\odh{\widecneg{\toform{\Gamma_2}}}}
						{\toformnew{\pi_1}}
						{\toform{\Gamma_1}\lpar \lA}
						{}
					}
					\ltens {\widecneg{\toform{\Delta_2}}}
				}
			}
		}
		$
		\myskiip
		$
		\vlderivation{\vliin{}{\cutr}{\vdash  \Gamma_1, \isneg \Gamma_2,  \Delta_1, \isneg \Delta_2}{\vlpr{}{\pi_1}{\vdash  \Gamma_1, \isneg \Gamma_2, \isneg\lA}}{\vlpr{}{\pi_2}{\vdash  \Delta_1, \isneg \Delta_2 , {\widecneg \lA}}}} 
		\translatesnewto  \ 
		\od{
			\vlin
			{\swir}
			{}
			{\od{
					\odd
					{\odh{\widecneg{\toform{\Gamma_2}} \ltens \cneg{\lA}}}
					{\toformnew{\pi_1}}
					{\toform{\Gamma_1}}
					{}
				}
				\lpar
				\toform{\Delta_1}
			}
			{
				\vlhy{
					{\widecneg{\toform{\Gamma_2}}}
					\ltens
					\od{
						\odd
						{\odh{\widecneg{\toform{\Delta_2}} }}
						{\toformnew{\pi_2}}
						{\cneg{\lA} \lpar \toform{\Delta_1}}
						{}
					}
				}
			}
		}
		$

		\smallskip
		$
		\vlderivation{\vliin{}{\sumrule}{\vdash   \Gamma_1,\isneg \Gamma_2}{\vlpr{}{\pi_1}{\vdash   \Gamma_1,\isneg \Gamma_2}}{\vlpr{}{\pi_2}{\vdash   \Gamma_1,\isneg \Gamma_2}}} 
		\ \translatesnewto \ \
		\vlderivation{
			\vlin
			{\pdr}
			{}
			{ \toform{\Gamma_1} } 
			{
				\vlin{\pur}{}{
					\od{\od{\odd{\odh{\widecneg{\toform{\Gamma_2}}}}{\toformnew{\pi_1}}{\toform{\Gamma_1}}{}}} 
					\lplus 
					\od{\od{\odd{\odh{\widecneg{\toform{\Gamma_2}}}}{\toformnew{\pi_2}}{\toform{\Gamma_1}}{}}}
				}
				{\vlhy{\widecneg{\toform{\Gamma_2}}}}
			}
		}
		$
		\myskiip
		$
		\vlderivation{\vlin{}{\zerorule}{\vdash \Gamma_1,\isneg \Gamma_2}{\vlhy{}}}
		\ \translatesnewto \ \
		\vlderivation{\vlin{\zdr}{}{\toform{\Gamma_1}}{\vlin{\zur}{}{\lzero}{\vlhy{\widecneg{\toform{\Gamma_2}}}}}}
		$
		
		\smallskip
		$
		\vlderivation{\vlin{}{\wnwrule}{\vdash \Gamma_1,\isneg \Gamma_2, \wn \lA}{\vlpr{}{\pi}{\vdash \Gamma_1,\isneg \Gamma_2}}}
		\translatesnewto \
		\vlderivation{
			\vlin
			{\fequiv}
			{}
			{\toform{\Gamma_1} \lpar \od{\odn{\lbot}{\wnwdr}{\wn  \lA}{}}}
			{\odd{\odh{\widecneg{\toform{\Gamma_2}}}}{\toformnew{\pi}}{\toform{\Gamma_1}}{} }
		}
		$
		\myskip 
		$
		\vlderivation{\vlin{}{\wnwrule}{\vdash \Gamma_1,\isneg \Gamma_2, \isneg{(\wn \lA)}}{\vlpr{}{\pi}{\vdash \Gamma_1,\isneg \Gamma_2}}}
		\translatesnewto \
		\od{
			\vlde
			{\toformnew{\pi}}
			{}
			{\toform{\Gamma_1}}
			{\vlin
				{\fequiv}{}
				{\widecneg{\toform{\Gamma_2}}}
				{\vlhy{\widecneg{\toform{\Gamma_2}} \ltens {\odn{\oc \cneg \lA}{\wnwur}{\lone}{}}}}
		}}
		$
		\myskip 
		$
		\vlinf{}{\ocwrule}{\vdash {\oc \lA}}{} 
		\translatesnewto \
		\vlinf{\ocwdr}{}{\oc  \lA}{\lone}
		$
		\myskip 
		$
		\vlinf{}{\ocwrule}{\vdash \isneg{(\oc \lA)}}{} 
		\translatesnewto \
		\vlinf{\ocwur}{}{\lbot}{\wn \cneg\lA}
		$
		
		$
		\vlderivation{\vlin{}{\ocdrule}{\vdash \Gamma_1,\isneg \Gamma_2, \oc \lA}{\vlpr{}{\pi}{\vdash \Gamma_1,\isneg \Gamma_2, \lA}}}
		\translatesnewto
		\vlderivation{
			\od{
				\odd{\odh{\widecneg{\toform{\Gamma_2}}}}{\toformnew{\pi}}{\toform{\Gamma_1} \lpar 
					\odn{\lA}{\ocddr}{\oc \lA}{} }{}
			} 
		}
		$
		\myskiip 
		$
		\vlderivation{\vlin{}{\ocdrule}{\vdash \Gamma_1,\isneg \Gamma_2, \isneg{(\oc \lA)}}{\vlpr{}{\pi}{\vdash \Gamma_1,\isneg \Gamma_2, \isneg\lA}}}
		\translatesnewto
		\od{
			\vlde{\toformnew{\pi}}{}{\toform{\Gamma_1}}
			{\vlhy{\widecneg{\toform{\Gamma_2}}\ltens \odn{{\wn \cneg A}}{\ocdur}{\cneg A}{}}}
		}
		$
		
		\smallskip
		$
		\vlderivation{\vlin{}{\wndrule}{\vdash  \Gamma_1,\isneg \Gamma_2, \wn \lA}{\vlpr{}{\pi}{\vdash  \Gamma_1,\isneg \Gamma_2, \lA}}}
		\translatesnewto
		\od{\odd{\odh{\widecneg{\toform{\Gamma_2}}}}{\toformnew{\pi}}{\toform{\Gamma_1} \lpar \odn{\lA}{\wnddr}{\wn \lA}{} }{}} 
		$
		\myskiip 
		$
		\vlderivation{\vlin{}{\wndrule}{\vdash  \Gamma_1,\isneg \Gamma_2,\isneg{(\wn \lA)}}{\vlpr{}{\pi}{\vdash  \Gamma_1,\isneg \Gamma_2,\isneg \lA}}}
		\translatesnewto
		\od{\odd{\odh{\widecneg{\toform{\Gamma_2}}\ltens \odn{\oc\cneg\lA}{\wndur}{\cneg\lA}{}}}{\toformnew{\pi}}{\toform{\Gamma_1}  }{}} 
		$
		
		\smallskip
		$
		\vlderivation{\vlin{}{\wncrule}{\vdash  \Gamma_1,\isneg \Gamma_2, \wn \lA}{\vlpr{}{\pi}{\vdash  \Gamma_1,\isneg \Gamma_2,  \wn \lA, \wn \lA}}}
		\translatesnewto
		\vlderivation{
			\od{
				\odd{\odh{\widecneg{\toform{\Gamma_2}}}}
				{\toformnew{\pi}}
				{\toform{\Gamma_1} \lpar \odn{\wn \lA \lpar \wn \lA}{\wncdr}{\wn \lA}{} }
				{}
			}
		}
		$
		\myskiip 
		$
		\vlderivation{\vlin{}{\wncrule}{\vdash  \Gamma_1,\isneg \Gamma_2, \isneg{(\wn \lA)}}{\vlpr{}{\pi}{\vdash  \Gamma_1,\isneg \Gamma_2,  \isneg{(\wn \lA)}, \isneg{(\wn \lA)}}}}
		\translatesnewto
		\vlderivation{
			\od{
				\odd{\odh{\widecneg{\toform{\Gamma_2}}\ltens \odn{\oc \cneg \lA \ltens \oc \cneg\lA}{\wncur}{\oc \cneg \lA}{} }}
				{\toformnew{\pi}}
				{\toform{\Gamma_1} }
				{}
			}
		}
		$
		
		\smallskip
		$
		\vlderivation{\vliin{}{\occrule}{\vdash  \Gamma_1,\isneg \Gamma_2, \oc \lA ,  \Delta_1,\isneg \Delta_2}{\vlpr{}{\pi_1}{\vdash  \Gamma_1,\isneg \Gamma_2, \oc \lA}}{\vlpr{}{\pi_2}{\vdash  \Delta_1,\isneg \Delta_2 ,\oc \lA}}} 
		\translatesnewto \ 
		\vlderivation{
			\vlin
			{2\times \swir}
			{}
			{ \toform{\Gamma_1} \lpar \od{\odn{\oc \lA \ltens \oc{\lA}}{\occdr}{\oc \lA}{}\lpar \toform{\Delta_1}} } 
			{
				\vlhy{
					\od{\od{\odd{\odh{\widecneg{\toform{\Gamma_2}}}}{\toformnew{\pi_1}}{\toform{\Gamma_1} \lpar \oc\lA}{}}}
					\ltens 
					\od{\od{\odd{\odh{\widecneg{\toform{\Delta_2}}}}{\toformnew{\pi_2}}{\toform{\Delta_1} \lpar \oc\lA}{}}}
				} 
			}
		}
		$
		\myskip 
		$
		\vlderivation{\vliin{}{\occrule}{\vdash  \Gamma_1,\isneg \Gamma_2, \isneg{(\oc \lA)} ,  \Delta_1,\isneg \Delta_2}{\vlpr{}{\pi_1}{\vdash  \Gamma_1,\isneg \Gamma_2, \isneg{(\oc \lA)}}}{\vlpr{}{\pi_2}{\vdash  \Delta_1,\isneg \Delta_2 ,\isneg{(\oc \lA)}}}} 
		\translatesnewto \ 
		\od{
			\vlin{2\times \swir}{}{{\od{\odd{\odh{\widecneg{\toform{\Gamma_2}}\ltens\wn{\cneg \lA}}}{\toformnew{\pi_1}}{\toform{\Gamma_1}}{}}}
				\lpar 
				{\od{\odd{\odh{\widecneg{\toform{\Delta_2}}\ltens\wn{\cneg \lA}}}{\toformnew{\pi_2}}{\toform{\Delta_1}}{}}}}
			{\vlhy{\widecneg{\toform{\Gamma_2}} \ltens {\odn{\wn \cneg \lA}{\occur}{\wn \cneg \lA \lpar \wn \cneg \lA}{}} \ltens \widecneg{\toform{\Delta_2}}}}}
		$
		
	\end{center}
	\caption{Translation of $\eta$-expanded $\SDill$ sequent calculus derivations into  $\SDDIs \setminus \set{\aiur}$ derivations (where if $\Gamma = A_1, \dots, A_n$, then $\isneg{\Gamma} = \isneg{A_1}, \dots, \isneg{A_n}$).}
	\label{fig:seqToDeepUp}
\end{figure}

Note that an $\eta$-expanded derivation $\pi$ in $\SDill$ may have several translations $\toformnew{\pi}$, depending on the initial selection of cut formulas in $\pi$ to mark.
We omit this dependency in the notation $\toformnew{\pi}$.  
When we state a property of a translation $\toformnew{\pi}$, we mean that it holds for \emph{any} initial selection of cut formulas~in~$\pi$.

\begin{lemma}[Target of the translation $\toformnew{\cdot}$]
	\label{lemma:target}
	Let $\pi$ be an $\eta$-expanded derivation in $\SDill$.
	Then, $\toformnew{\pi}$ is a derivation in $\SDDI$.
	If, moreover, $\pi$ is cut-free, then 
	$\toformnew{\pi} = \dD \circ \toform{\pi}$ where $\deriv{\dD}{\lone\!}{\{\zur,\pur\}}{\lones}$ 	for some $n \in \N$, 
	and $\toform{\pi}$ (defined in \Cref{fig:seqToDeep}) is a derivation in $\DDIdown$.
\end{lemma}

\begin{proof}
	By induction on $\pi$. 
	Each step in \Cref{fig:seqToDeepUp} introduces only rules in $\SDDI$, in particular no step introduces the rule $\iur$.
	If $\pi$ is cut-free, then no formula in $\pi$ is marked; 
	thus, each step in \Cref{fig:seqToDeepUp} acts like $\toform{\cdot}$ defined in \Cref{fig:seqToDeep} except possibly for adding some rules $\pur$ and $\zur$ on top, and it does not introduce any other rule in $\DDIup$ except $\swir$ (which is also in $\DDIdown$).
\end{proof}

There is trivial way to ``bend'' a derivation $\pi$ of $\provevia{\SDill} \Gamma, A$ to a derivation in $\SDDIs$: take {\small$\vlupsmash{\vlderivation{\vlin{\zdr}{}{\toform{\Gamma}}{\vlin{\zur}{}{\lzero}{\vlhy{\cneg A}}}}}$}.
However, such a translation does not make diagram \eqref{eq:diagram} commute, because it does not keep track of resources.
The translation $\toformnew{\pi}$, instead, is \emph{resource-sensitive}, thanks to a one-to-one correspondence between the occurrences of rules for weakening, dereliction, and contraction in $\pi$ and~in~$\toformnew{\pi}$ (proved by induction~on~$\pi$).

\begin{lemma}[Resources]
	\label{lemma:bend}
	For any $\eta$-expanded derivation $\pi$ in $\SDill$, 
there is a one-to-one correspondence between the rule occurrences of $\ocdrule$ in $\pi$ and the rule occurrences of $\set{\ocddr, \ocdur}$ in $\toformnew{\pi}$, and similarly between $\wndrule$ and $\set{\wnddr, \wndur}$, $\occrule$ and  $\set{\occdr,\occur}$, $\wncrule$ and $\set{\wncdr, \wncur}$, $\ocwrule$ and $\set{\ocwdr, \ocwur}$, $\wnwrule$ and $\set{\wnwdr,\wnwur}$.
\end{lemma}

Let $\normeq$ be the reflexive, transitive and symmetric closure of $\normred$.

\begin{theorem}[Simulation]
	\label{thm:cutStep}
	If $\pi$ is an $\eta$-expanded $\SDill$ derivation and $\pi\cutred \pi' $, then $\toformnew{\pi} \normeq \toformnew{\pi'}$.
\end{theorem}

\begin{proof}[Proof (sketch)]
	By case analysis of cut-elimination steps for $\SDill$, and the corresponding normalization steps in $\SDDI$.
	By \Cref{lemma:bend}, cut-elimination steps for the rules in $\set{\ocdrule,\wndrule, \occrule, \wncrule, \ocwrule, \wnwrule}$ correspond to the (non-trivial) interaction-net rule permutations in \Cref{fig:modRulesPerm} (\Cref{lemma:modupVSmoddown}.\ref{item:INperm}).
	A cut-elimination step $\ltens$/$\lpar$ corresponds to linear permutations to prove  \Cref{lemma:linearPerm}.
	Commutative cut-elimination steps correspond to trivial rule permutations.
	All permutations are interleaved by steps \eqref{eq:sumZeroPerm} and \eqref{eq:switchperm}, \mbox{in both~directions.}
\end{proof}

\Cref{thm:cutStep} says that normalization steps in $\SDDI$ (\Cref{def:normalization}) 
mimic $\SDill$ cut-elimination via translation $\toformnew{\cdot}$.
As a consequence, cut-elimination/normalization commutes with translation $\toformnew{\cdot}$.

\begin{corollary}[Commutation]
	If $\pi$ is an $\eta$-expanded $\DiLL_0$ derivation and $\pi\cutred^* \hat\pi $ with $\hat\pi$ cut-free, then $\toformnew{\pi} \normred^* \toformnew{\hat\pi}$ and $\toformnew{\hat\pi}$ is normal for $\normred$.
\end{corollary}
\begin{proof}
	By simulation (\Cref{thm:cutStep}), from $\pi\cutred^* \hat\pi $ it follows that $\toformnew{\pi} \normeq \toformnew{\hat\pi}$.
	As $\hat\pi$ is cut-free, $\toformnew{\hat\pi}$ is normal for $\normred$ according to \Cref{lemma:target}, hence $\toformnew{\pi} \normred^* \toformnew{\hat\pi}$.
\end{proof}

\section{Conclusions and Future Work}


In this paper we introduced the first sound and complete deep inference system, $\SDDI$, for the promotion-free fragment of differential linear logic, $\SDill$ \cite{EhrhardRegnier06}.
The deep inference syntax recovers the symmetry of this logic lacking in $\SDill$ sequent calculus---but which can be found in the interaction-net formalism for $\SDill$ \cite{EhrhardRegnier06}---and keeps the inductive and handy tree-like structure of sequent calculus derivations---missing in interaction nets.
The deep inference formalism  allows  us to reduce cuts to atomic formulas, and provides a tool for a more fine-grained study of cut-elimination.
Moreover, the syntax explicitly represents and internalizes the notion of slices of a derivation. 

The inference rules of $\SDDIs$ present an up/down symmetry and we proved that the up-fragment of $\SDDIs$ is derivable from the down-fragment $\DDIdown$.
To prove this result we provided a normalization procedure based on rule permutations. 
In fact, the presence of the connective $\lplus$ and its unit $\lzero$ prevent the use of the general normalization result for splittable systems \cite{TubellaGuglielmi18}.
In our normalization procedure for $\SDDI$, we are able to distinguish different kinds of rule permutations depending on their computational behavior:
some rule permutations correspond to linear (in terms of resource) cut-elimination steps, some to resource management cut-elimination steps and some to slice management operations.
Thanks to \Cref{lemma:sliceMan} we could implement a reduction strategy alternating slice management and proper cut-elimination steps inside each slice.
The internal 
normalization procedure in $\SDDIs$ to prove \Cref{thm:SDDIcutElimInternal} provides ``cut-free'' derivations. And the translation $\toform{\cdot}$ defined in \Cref{fig:seqToDeep} maps cut-free $\SDill$ derivations to $\DDIdown$, the ``cut-free'' fragment of $\SDDIs$ (\Cref{thm:SDillSDDIs}.\ref{p:SDillSDDIs-cut-free}). 
We showed that cut-elimination/normalization does not commute with translation $\toform{\cdot}$, but it does with the translation $\toformnew{\cdot}$ defined in \Cref{fig:seqToDeepUp}, a resource-sensitive refinement of $\toform{\cdot}$.

\paragraph{Translation of $\DiLL$ proof-nets.}

An ongoing work is to extend our deep inference system to represent the \emph{full} differential linear logic $\DiLL = \SDill \cup \{\prule\}$ \cite{Pagani09,Tranquilli09,PaganiTranquilli17} (including the promotion rule), possibly with the rule $\mixr$ which allows one to derive $\lA \lpar \lB$ from $\lA \otimes \lB$.
The presence of promotion $\prule$ allows us to define a translation that commutes with cut-elimination for the reasons discussed in \Cref{sect:commute}.

In this extended deep inference system, we can translate not only the $\DiLL$ sequent calculus but also $\DiLL$ proof-nets, via a direct embedding that does not pass through the sequent calculus.
Indeed, the open deduction formalism \cite{gug:gun:par:2010} allows a direct encoding of proof-nets, plus a handy and inductive syntax.



\paragraph{Computational meaning and non-determinism.}

In $\SDill$ interaction nets, when a cut-elimination step creates a new construct $\sumrule$ (expressing a non-deterministic choice) or $\zerorule$ (expressing mismatch on demanded and supplied resources), this construct is instantaneously propagated to all the interaction net where it is plugged in, without any computational step. 
It is like $\SDill$ interaction nets allow one to deal with canonical forms only, in the sense of \Cref{def:canonical}.

A feature of our deep inference formalism is that the constructs $\lplus$ (non-determinism) and $\lzero$ (resource mismatch) are internalized in the syntax, and when they appear during the normalization process, they are propagated all along the derivation by means of normalization steps (slice management, \Cref{lemma:sumPermutations} and \Cref{lemma:sliceMan}).
Is there a computational meaning in these kind of steps? Is it possible to interpret them in a model of computation which intrinsically represents non-determinism, parallelism and concurrency?

The $\pi$-calculus~\cite{milner:pi} (a model of concurrent computation) can be encoded in $\SDill$ \cite{EhrhardLaurent10}, but Mazza \cite{Mazza18} pointed out that the non-determinism expressed by $\SDill$ is too weak for true concurrency.
Deep inference may shed new light on the quest for a convincing proof-theoretic counterpart \mbox{of concurrency.}

\paragraph{Acknowledgments.} The authors thank Andrea Aler Tubella, Alessio Guglielmi, Lutz Stra\ss{}burger and the anonymous reviewers for their insightful comments.
This work has been partially supported by the EPSRC grant EP/R029121/1
``Typed Lambda-Calculi with Sharing and Unsharing''.

\bibliographystyle{eptcs}
\bibliography{biblio}

\begin{thebibliography}{10}
\providecommand{\bibitemdeclare}[2]{}
\providecommand{\surnamestart}{}
\providecommand{\surnameend}{}
\providecommand{\urlprefix}{Available at }
\providecommand{\url}[1]{\texttt{#1}}
\providecommand{\href}[2]{\texttt{#2}}
\providecommand{\urlalt}[2]{\href{#1}{#2}}
\providecommand{\doi}[1]{doi:\urlalt{http://dx.doi.org/#1}{#1}}
\providecommand{\bibinfo}[2]{#2}

\bibitemdeclare{inproceedings}{acc:hor:str:LBF}
\bibitem{acc:hor:str:LBF}
\bibinfo{author}{Matteo \surnamestart Acclavio\surnameend},
  \bibinfo{author}{Ross \surnamestart Horne\surnameend} \&
  \bibinfo{author}{Lutz \surnamestart Stra{\ss}burger\surnameend}
  (\bibinfo{year}{2020}): \emph{\bibinfo{title}{Logic beyond formulas: a proof
  system on graphs}}.
\newblock In: {\sl \bibinfo{booktitle}{35th Annual ACM/IEEE Symposium on Logic
  in Computer Science, {LICS}'20}}, \bibinfo{publisher}{{ACM}}, pp.
  \bibinfo{pages}{38--52}, \doi{10.1145/3373718.3394763}.

\bibitemdeclare{inproceedings}{brunnler:tiu:01}
\bibitem{brunnler:tiu:01}
\bibinfo{author}{Kai \surnamestart Br{\"u}nnler\surnameend} \&
  \bibinfo{author}{Alwen~Fernanto \surnamestart Tiu\surnameend}
  (\bibinfo{year}{2001}): \emph{\bibinfo{title}{A Local System for Classical
  Logic}}.
\newblock In: {\sl \bibinfo{booktitle}{Logic for Programming, Artificial
  Intelligence, and Reasoning , 8th International Conference, {LPAR} 2001}},
  {\sl \bibinfo{series}{Lecture Notes in Computer Science}}
  \bibinfo{volume}{2250}, \bibinfo{publisher}{Springer}, pp.
  \bibinfo{pages}{347--361}, \doi{10.1007/3-540-45653-8\_24}.

\bibitemdeclare{article}{Carvalho18}
\bibitem{Carvalho18}
\bibinfo{author}{Daniel \surnamestart de~Carvalho\surnameend}
  (\bibinfo{year}{2018}): \emph{\bibinfo{title}{Taylor expansion in linear
  logic is invertible}}.
\newblock {\sl \bibinfo{journal}{Log. Methods Comput. Sci.}}
  \bibinfo{volume}{14}(\bibinfo{number}{4}), \doi{10.23638/LMCS-14(4:21)2018}.

\bibitemdeclare{article}{Ehrhard05}
\bibitem{Ehrhard05}
\bibinfo{author}{Thomas \surnamestart Ehrhard\surnameend}
  (\bibinfo{year}{2005}): \emph{\bibinfo{title}{Finiteness spaces}}.
\newblock {\sl \bibinfo{journal}{Mathematical Structures in Computer Science}}
  \bibinfo{volume}{15}(\bibinfo{number}{4}), pp. \bibinfo{pages}{615--646},
  \doi{10.1017/S0960129504004645}.

\bibitemdeclare{article}{Ehrhard18}
\bibitem{Ehrhard18}
\bibinfo{author}{Thomas \surnamestart Ehrhard\surnameend}
  (\bibinfo{year}{2018}): \emph{\bibinfo{title}{An introduction to differential
  linear logic: proof-nets, models and antiderivatives}}.
\newblock {\sl \bibinfo{journal}{Mathematical Structures in Computer Science}}
  \bibinfo{volume}{28}(\bibinfo{number}{7}), pp. \bibinfo{pages}{995--1060},
  \doi{10.1017/S0960129516000372}.

\bibitemdeclare{article}{EhrhardLaurent10}
\bibitem{EhrhardLaurent10}
\bibinfo{author}{Thomas \surnamestart Ehrhard\surnameend} \&
  \bibinfo{author}{Olivier \surnamestart Laurent\surnameend}
  (\bibinfo{year}{2010}): \emph{\bibinfo{title}{Interpreting a finitary
  pi-calculus in differential interaction nets}}.
\newblock {\sl \bibinfo{journal}{Inf. Comput.}}
  \bibinfo{volume}{208}(\bibinfo{number}{6}), pp. \bibinfo{pages}{606--633},
  \doi{10.1016/j.ic.2009.06.005}.

\bibitemdeclare{article}{EhrhardRegnier03}
\bibitem{EhrhardRegnier03}
\bibinfo{author}{Thomas \surnamestart Ehrhard\surnameend} \&
  \bibinfo{author}{Laurent \surnamestart Regnier\surnameend}
  (\bibinfo{year}{2003}): \emph{\bibinfo{title}{The differential
  lambda-calculus}}.
\newblock {\sl \bibinfo{journal}{Theor. Comput. Sci.}}
  \bibinfo{volume}{309}(\bibinfo{number}{1-3}), pp. \bibinfo{pages}{1--41},
  \doi{10.1016/S0304-3975(03)00392-X}.

\bibitemdeclare{inproceedings}{EhrhardRegnier06b}
\bibitem{EhrhardRegnier06b}
\bibinfo{author}{Thomas \surnamestart Ehrhard\surnameend} \&
  \bibinfo{author}{Laurent \surnamestart Regnier\surnameend}
  (\bibinfo{year}{2006}): \emph{\bibinfo{title}{B{\"{o}}hm Trees, Krivine's
  Machine and the Taylor Expansion of Lambda-Terms}}.
\newblock In: {\sl \bibinfo{booktitle}{Second Conference on Computability in
  Europe, CiE 2006}}, {\sl \bibinfo{series}{Lecture Notes in Computer Science}}
  \bibinfo{volume}{3988}, \bibinfo{publisher}{Springer}, pp.
  \bibinfo{pages}{186--197}, \doi{10.1007/11780342\_20}.

\bibitemdeclare{article}{EhrhardRegnier06}
\bibitem{EhrhardRegnier06}
\bibinfo{author}{Thomas \surnamestart Ehrhard\surnameend} \&
  \bibinfo{author}{Laurent \surnamestart Regnier\surnameend}
  (\bibinfo{year}{2006}): \emph{\bibinfo{title}{Differential interaction
  nets}}.
\newblock {\sl \bibinfo{journal}{Theor. Comput. Sci.}}
  \bibinfo{volume}{364}(\bibinfo{number}{2}), pp. \bibinfo{pages}{166--195},
  \doi{10.1016/j.tcs.2006.08.003}.

\bibitemdeclare{article}{EhrhardRegnier08}
\bibitem{EhrhardRegnier08}
\bibinfo{author}{Thomas \surnamestart Ehrhard\surnameend} \&
  \bibinfo{author}{Laurent \surnamestart Regnier\surnameend}
  (\bibinfo{year}{2008}): \emph{\bibinfo{title}{Uniformity and the Taylor
  expansion of ordinary lambda-terms}}.
\newblock {\sl \bibinfo{journal}{Theor. Comput. Sci.}}
  \bibinfo{volume}{403}(\bibinfo{number}{2-3}), pp. \bibinfo{pages}{347--372},
  \doi{10.1016/j.tcs.2008.06.001}.

\bibitemdeclare{phdthesis}{gim:phd}
\bibitem{gim:phd}
\bibinfo{author}{St{\'e}phane \surnamestart Gimenez\surnameend}
  (\bibinfo{year}{2009}): \emph{\bibinfo{title}{{Programming, Computation and
  their Analysis using Nets from Linear Logic}}}.
\newblock \bibinfo{type}{Theses}, \bibinfo{school}{{Universit{\'e} Paris
  Diderot}}.
\newblock \urlprefix\url{https://tel.archives-ouvertes.fr/tel-00629013}.

\bibitemdeclare{inproceedings}{Gimenez11}
\bibitem{Gimenez11}
\bibinfo{author}{St{\'{e}}phane \surnamestart Gimenez\surnameend}
  (\bibinfo{year}{2011}): \emph{\bibinfo{title}{Realizability Proof for
  Normalization of Full Differential Linear Logic}}.
\newblock In: {\sl \bibinfo{booktitle}{Typed Lambda Calculi and Applications -
  10th International Conference, {TLCA} 2011}}, {\sl \bibinfo{series}{Lecture
  Notes in Computer Science}} \bibinfo{volume}{6690},
  \bibinfo{publisher}{Springer}, pp. \bibinfo{pages}{107--122},
  \doi{10.1007/978-3-642-21691-6\_11}.

\bibitemdeclare{article}{Girard87}
\bibitem{Girard87}
\bibinfo{author}{Jean{-}Yves \surnamestart Girard\surnameend}
  (\bibinfo{year}{1987}): \emph{\bibinfo{title}{Linear Logic}}.
\newblock {\sl \bibinfo{journal}{Theor. Comput. Sci.}} \bibinfo{volume}{50},
  pp. \bibinfo{pages}{1--102}, \doi{10.1016/0304-3975(87)90045-4}.

\bibitemdeclare{article}{Girard01}
\bibitem{Girard01}
\bibinfo{author}{Jean{-}Yves \surnamestart Girard\surnameend}
  (\bibinfo{year}{2001}): \emph{\bibinfo{title}{Locus Solum: From the rules of
  logic to the logic of rules}}.
\newblock {\sl \bibinfo{journal}{Mathematical Structures in Computer Science}}
  \bibinfo{volume}{11}(\bibinfo{number}{3}), pp. \bibinfo{pages}{301--506},
  \doi{10.1017/S096012950100336X}.

\bibitemdeclare{inproceedings}{GuerrieriPellissierTortora16}
\bibitem{GuerrieriPellissierTortora16}
\bibinfo{author}{Giulio \surnamestart Guerrieri\surnameend},
  \bibinfo{author}{Luc \surnamestart Pellissier\surnameend} \&
  \bibinfo{author}{Lorenzo~Tortora \surnamestart de~Falco\surnameend}
  (\bibinfo{year}{2016}): \emph{\bibinfo{title}{Computing Connected
  Proof(-Structure)s From Their Taylor Expansion}}.
\newblock In: {\sl \bibinfo{booktitle}{1st International Conference on Formal
  Structures for Computation and Deduction, {FSCD} 2016}}, {\sl
  \bibinfo{series}{LIPIcs}}~\bibinfo{volume}{52}, \bibinfo{publisher}{Schloss
  Dagstuhl}, pp. \bibinfo{pages}{20:1--20:18},
  \doi{10.4230/LIPIcs.FSCD.2016.20}.

\bibitemdeclare{inproceedings}{GuerrieriPellissierTortora19}
\bibitem{GuerrieriPellissierTortora19}
\bibinfo{author}{Giulio \surnamestart Guerrieri\surnameend},
  \bibinfo{author}{Luc \surnamestart Pellissier\surnameend} \&
  \bibinfo{author}{Lorenzo~Tortora \surnamestart de~Falco\surnameend}
  (\bibinfo{year}{2019}): \emph{\bibinfo{title}{Proof-Net as Graph, Taylor
  Expansion as Pullback}}.
\newblock In: {\sl \bibinfo{booktitle}{Logic, Language, Information, and
  Computation - 26th International Workshop, WoLLIC 2019}}, {\sl
  \bibinfo{series}{Lecture Notes in Computer Science}} \bibinfo{volume}{11541},
  \bibinfo{publisher}{Springer}, pp. \bibinfo{pages}{282--300},
  \doi{10.1007/978-3-662-59533-6\_18}.

\bibitemdeclare{inproceedings}{GuerrieriPellissierTortora20}
\bibitem{GuerrieriPellissierTortora20}
\bibinfo{author}{Giulio \surnamestart Guerrieri\surnameend},
  \bibinfo{author}{Luc \surnamestart Pellissier\surnameend} \&
  \bibinfo{author}{Lorenzo~Tortora \surnamestart de~Falco\surnameend}
  (\bibinfo{year}{2020}): \emph{\bibinfo{title}{Glueability of Resource
  Proof-Structures: Inverting the Taylor Expansion}}.
\newblock In: {\sl \bibinfo{booktitle}{28th {EACSL} Annual Conference on
  Computer Science Logic, {CSL} 2020}}, {\sl \bibinfo{series}{LIPIcs}}
  \bibinfo{volume}{152}, \bibinfo{publisher}{Schloss Dagstuhl}, pp.
  \bibinfo{pages}{24:1--24:18}, \doi{10.4230/LIPIcs.CSL.2020.24}.

\bibitemdeclare{techreport}{guglielmi1999calculus}
\bibitem{guglielmi1999calculus}
\bibinfo{author}{Alessio \surnamestart Guglielmi\surnameend}
  (\bibinfo{year}{1999}): \emph{\bibinfo{title}{A Calculus of Order and
  Interaction}}.
\newblock \bibinfo{type}{Technical Report} \bibinfo{number}{WV-1999-04},
  \bibinfo{institution}{Department of Computer Science, Dresden University of
  Technology}.
\newblock
  \urlprefix\url{http://citeseerx.ist.psu.edu/viewdoc/summary?doi=10.1.1.16.6332}.

\bibitemdeclare{article}{gug:SIS}
\bibitem{gug:SIS}
\bibinfo{author}{Alessio \surnamestart Guglielmi\surnameend}
  (\bibinfo{year}{2007}): \emph{\bibinfo{title}{A System of Interaction and
  Structure}}.
\newblock {\sl \bibinfo{journal}{{ACM} Trans. Comput. Log.}}
  \bibinfo{volume}{8}(\bibinfo{number}{1}), pp. \bibinfo{pages}{1--64},
  \doi{10.1145/1182613.1182614}.

\bibitemdeclare{inproceedings}{gug:gun:par:2010}
\bibitem{gug:gun:par:2010}
\bibinfo{author}{Alessio \surnamestart Guglielmi\surnameend},
  \bibinfo{author}{Tom \surnamestart Gundersen\surnameend} \&
  \bibinfo{author}{Michel \surnamestart Parigot\surnameend}
  (\bibinfo{year}{2010}): \emph{\bibinfo{title}{{A Proof Calculus Which Reduces
  Syntactic Bureaucracy}}}.
\newblock In: {\sl \bibinfo{booktitle}{Proceedings of the 21st International
  Conference on Rewriting Techniques and Applications, {RTA} 2010}}, {\sl
  \bibinfo{series}{LIPIcs}}~\bibinfo{volume}{6}, \bibinfo{publisher}{Schloss
  Dagstuhl - Leibniz-Zentrum f{\"{u}}r Informatik}, pp.
  \bibinfo{pages}{135--150}, \doi{10.4230/LIPIcs.RTA.2010.135}.

\bibitemdeclare{inproceedings}{gug:str:01}
\bibitem{gug:str:01}
\bibinfo{author}{Alessio \surnamestart Guglielmi\surnameend} \&
  \bibinfo{author}{Lutz \surnamestart Stra{\ss}burger\surnameend}
  (\bibinfo{year}{2001}): \emph{\bibinfo{title}{Non-commutativity and {MELL} in
  the Calculus of Structures}}.
\newblock In \bibinfo{editor}{Laurent \surnamestart Fribourg\surnameend},
  editor: {\sl \bibinfo{booktitle}{Computer Science Logic, 15th International
  Workshop, {CSL} 2001}}, {\sl \bibinfo{series}{Lecture Notes in Computer
  Science}} \bibinfo{volume}{2142}, \bibinfo{publisher}{Springer}, pp.
  \bibinfo{pages}{54--68}, \doi{10.1007/3-540-44802-0\_5}.

\bibitemdeclare{inproceedings}{gug:str:02}
\bibitem{gug:str:02}
\bibinfo{author}{Alessio \surnamestart Guglielmi\surnameend} \&
  \bibinfo{author}{Lutz \surnamestart Stra{\ss}burger\surnameend}
  (\bibinfo{year}{2002}): \emph{\bibinfo{title}{A Non-commutative Extension of
  {MELL}}}.
\newblock In: {\sl \bibinfo{booktitle}{Logic for Programming, Artificial
  Intelligence, and Reasoning, 9th International Conference, {LPAR} 2002}},
  {\sl \bibinfo{series}{Lecture Notes in Computer Science}}
  \bibinfo{volume}{2514}, \bibinfo{publisher}{Springer}, pp.
  \bibinfo{pages}{231--246}, \doi{10.1007/3-540-36078-6\_16}.

\bibitemdeclare{article}{horne2019morgan}
\bibitem{horne2019morgan}
\bibinfo{author}{Ross \surnamestart Horne\surnameend}, \bibinfo{author}{Alwen
  \surnamestart Tiu\surnameend}, \bibinfo{author}{Bogdan \surnamestart
  Aman\surnameend} \& \bibinfo{author}{Gabriel \surnamestart
  Ciobanu\surnameend} (\bibinfo{year}{2019}): \emph{\bibinfo{title}{De Morgan
  Dual Nominal Quantifiers Modelling Private Names in Non-Commutative Logic}}.
\newblock {\sl \bibinfo{journal}{ACM Transactions on Computational Logic
  (TOCL)}} \bibinfo{volume}{20}(\bibinfo{number}{4}), pp.
  \bibinfo{pages}{22:1--22:44}, \doi{10.1145/3325821}.

\bibitemdeclare{inproceedings}{Lafont90}
\bibitem{Lafont90}
\bibinfo{author}{Yves \surnamestart Lafont\surnameend} (\bibinfo{year}{1990}):
  \emph{\bibinfo{title}{Interaction Nets}}.
\newblock In: {\sl \bibinfo{booktitle}{Seventeenth Annual {ACM} Symposium on
  Principles of Programming Languages, {POPL} 1990}}, \bibinfo{publisher}{{ACM}
  Press}, pp. \bibinfo{pages}{95--108}, \doi{10.1145/96709.96718}.

\bibitemdeclare{article}{Mazza18}
\bibitem{Mazza18}
\bibinfo{author}{Damiano \surnamestart Mazza\surnameend}
  (\bibinfo{year}{2018}): \emph{\bibinfo{title}{The true concurrency of
  differential interaction nets}}.
\newblock {\sl \bibinfo{journal}{Math. Struct. Comput. Sci.}}
  \bibinfo{volume}{28}(\bibinfo{number}{7}), pp. \bibinfo{pages}{1097--1125},
  \doi{10.1017/S0960129516000402}.

\bibitemdeclare{inproceedings}{MazzaPagani07}
\bibitem{MazzaPagani07}
\bibinfo{author}{Damiano \surnamestart Mazza\surnameend} \&
  \bibinfo{author}{Michele \surnamestart Pagani\surnameend}
  (\bibinfo{year}{2007}): \emph{\bibinfo{title}{The Separation Theorem for
  Differential Interaction Nets}}.
\newblock In: {\sl \bibinfo{booktitle}{Logic for Programming, Artificial
  Intelligence, and Reasoning, 14th International Conference, {LPAR} 2007}},
  {\sl \bibinfo{series}{Lecture Notes in Computer Science}}
  \bibinfo{volume}{4790}, \bibinfo{publisher}{Springer}, pp.
  \bibinfo{pages}{393--407}, \doi{10.1007/978-3-540-75560-9\_29}.

\bibitemdeclare{book}{milner:pi}
\bibitem{milner:pi}
\bibinfo{author}{Robin \surnamestart Milner\surnameend} (\bibinfo{year}{1999}):
  \emph{\bibinfo{title}{Communicating and Mobile Systems: the Pi Calculus}}.
\newblock \bibinfo{publisher}{Cambridge University Press},
  \bibinfo{address}{New York, NY, USA}.

\bibitemdeclare{inproceedings}{Pagani09}
\bibitem{Pagani09}
\bibinfo{author}{Michele \surnamestart Pagani\surnameend}
  (\bibinfo{year}{2009}): \emph{\bibinfo{title}{The Cut-Elimination Theorem for
  Differential Nets with Promotion}}.
\newblock In: {\sl \bibinfo{booktitle}{Typed Lambda Calculi and Applications,
  9th International Conference, {TLCA} 2009}}, {\sl \bibinfo{series}{Lecture
  Notes in Computer Science}} \bibinfo{volume}{5608},
  \bibinfo{publisher}{Springer}, pp. \bibinfo{pages}{219--233},
  \doi{10.1007/978-3-642-02273-9\_17}.

\bibitemdeclare{inproceedings}{PaganiTasson09}
\bibitem{PaganiTasson09}
\bibinfo{author}{Michele \surnamestart Pagani\surnameend} \&
  \bibinfo{author}{Christine \surnamestart Tasson\surnameend}
  (\bibinfo{year}{2009}): \emph{\bibinfo{title}{The Inverse Taylor Expansion
  Problem in Linear Logic}}.
\newblock In: {\sl \bibinfo{booktitle}{Proceedings of the 24th Annual {IEEE}
  Symposium on Logic in Computer Science, {LICS} 2009, 11-14 August 2009, Los
  Angeles, CA, {USA}}}, \bibinfo{publisher}{{IEEE} Computer Society}, pp.
  \bibinfo{pages}{222--231}, \doi{10.1109/LICS.2009.35}.

\bibitemdeclare{article}{PaganiTranquilli17}
\bibitem{PaganiTranquilli17}
\bibinfo{author}{Michele \surnamestart Pagani\surnameend} \&
  \bibinfo{author}{Paolo \surnamestart Tranquilli\surnameend}
  (\bibinfo{year}{2017}): \emph{\bibinfo{title}{The conservation theorem for
  differential nets}}.
\newblock {\sl \bibinfo{journal}{Mathematical Structures in Computer Science}}
  \bibinfo{volume}{27}(\bibinfo{number}{6}), pp. \bibinfo{pages}{939--992},
  \doi{10.1017/S0960129515000456}.

\bibitemdeclare{phdthesis}{ralph:phd}
\bibitem{ralph:phd}
\bibinfo{author}{Benjamin \surnamestart Ralph\surnameend}
  (\bibinfo{year}{2019}): \emph{\bibinfo{title}{Modular Normalisation of
  Classical Proofs}}.
\newblock Ph.D. thesis, \bibinfo{school}{University of Bath}.
\newblock
  \urlprefix\url{https://researchportal.bath.ac.uk/files/189585932/thesis_ralph_final.pdf}.

\bibitemdeclare{phdthesis}{str:phd}
\bibitem{str:phd}
\bibinfo{author}{Lutz \surnamestart Stra{\ss}burger\surnameend}
  (\bibinfo{year}{2003}): \emph{\bibinfo{title}{Linear Logic and
  Noncommutativity in the Calculus of Structures}}.
\newblock Ph.D. thesis, \bibinfo{school}{Technische Universit\"at Dresden}.

\bibitemdeclare{article}{str:MELLinCoS}
\bibitem{str:MELLinCoS}
\bibinfo{author}{Lutz \surnamestart Straßburger\surnameend}
  (\bibinfo{year}{2003}): \emph{\bibinfo{title}{MELL in the calculus of
  structures}}.
\newblock {\sl \bibinfo{journal}{Theoretical Computer Science}}
  \bibinfo{volume}{309}(\bibinfo{number}{1}), pp. \bibinfo{pages}{213 -- 285},
  \doi{10.1016/S0304-3975(03)00240-8}.

\bibitemdeclare{inproceedings}{Tranquilli09}
\bibitem{Tranquilli09}
\bibinfo{author}{Paolo \surnamestart Tranquilli\surnameend}
  (\bibinfo{year}{2009}): \emph{\bibinfo{title}{Confluence of Pure Differential
  Nets with Promotion}}.
\newblock In: {\sl \bibinfo{booktitle}{Computer Science Logic, 23rd
  international Workshop, {CSL} 2009, 18th Annual Conference of the EACSLS}},
  {\sl \bibinfo{series}{Lecture Notes in Computer Science}}
  \bibinfo{volume}{5771}, \bibinfo{publisher}{Springer}, pp.
  \bibinfo{pages}{500--514}, \doi{10.1007/978-3-642-04027-6\_36}.

\bibitemdeclare{article}{Tranquilli11}
\bibitem{Tranquilli11}
\bibinfo{author}{Paolo \surnamestart Tranquilli\surnameend}
  (\bibinfo{year}{2011}): \emph{\bibinfo{title}{Intuitionistic differential
  nets and lambda-calculus}}.
\newblock {\sl \bibinfo{journal}{Theor. Comput. Sci.}}
  \bibinfo{volume}{412}(\bibinfo{number}{20}), pp. \bibinfo{pages}{1979--1997},
  \doi{10.1016/j.tcs.2010.12.022}.

\bibitemdeclare{phdthesis}{tubella:phd}
\bibitem{tubella:phd}
\bibinfo{author}{Andrea~Aler \surnamestart Tubella\surnameend}
  (\bibinfo{year}{2017}): \emph{\bibinfo{title}{A study of normalisation
  through subatomic logic}}.
\newblock Ph.D. thesis, \bibinfo{school}{University of Bath}.
\newblock
  \urlprefix\url{https://researchportal.bath.ac.uk/files/187926411/thesis.pdf}.

\bibitemdeclare{article}{TubellaGuglielmi18}
\bibitem{TubellaGuglielmi18}
\bibinfo{author}{Andrea~Aler \surnamestart Tubella\surnameend} \&
  \bibinfo{author}{Alessio \surnamestart Guglielmi\surnameend}
  (\bibinfo{year}{2018}): \emph{\bibinfo{title}{Subatomic Proof Systems:
  Splittable Systems}}.
\newblock {\sl \bibinfo{journal}{{ACM} Trans. Comput. Log.}}
  \bibinfo{volume}{19}(\bibinfo{number}{1}), pp. \bibinfo{pages}{5:1--5:33},
  \doi{10.1145/3173544}.

\bibitemdeclare{unpublished}{tub:str:esslli19}
\bibitem{tub:str:esslli19}
\bibinfo{author}{Andrea~Aler \surnamestart Tubella\surnameend} \&
  \bibinfo{author}{Lutz \surnamestart Stra{\ss}burger\surnameend}
  (\bibinfo{year}{2019}): \emph{\bibinfo{title}{{Introduction to Deep
  Inference}}}.
\newblock \urlprefix\url{https://hal.inria.fr/hal-02390267}.
\newblock \bibinfo{note}{Lecture}.

\end{thebibliography}

\newpage
\appendix

\end{document}